%% file: main.tex
\documentclass[a4paper,USenglish,cleveref, autoref, thm-restate]{lipics-v2021}%This is a template for producing LIPIcs articles. 
\nolinenumbers
\hideLIPIcs  %uncomment to remove references to LIPIcs series (logo, DOI, ...), e.g. when preparing a pre-final version to be uploaded to arXiv or another public repository

\usepackage[utf8]{inputenc}
\usepackage{todonotes}

\title{Decidable (Ac)counting with Parikh and Muller: Adding Presburger Arithmetic to Monadic Second-\\Order Logic over Tree-Interpretable Structures} %TODO Please add

\titlerunning{Adding Presburger Arithmetic to MSO Logic over Tree-Interpretable Structures} %TODO optional, please use if title is longer than one line

\author{Luisa Herrmann}{Computational Logic Group, TU Dresden, Germany \and Center for Scalable Data Analytics and Artificial Intelligence Dresden/Leipzig, Germany\and \url{https://iccl.inf.tu-dresden.de/web/Luisa_Herrmann}}{luisa.herrmann@tu-dresden.de}{https://orcid.org/0009-0004-9532-0994}{}

\author{Vincent Peth}{Département d’informatique de l’ÉNS, École normale supérieure,
CNRS, PSL University, %75005 Paris, 
France}{vincent.peth@ens.psl.eu}{https://orcid.org/0009-0007-8450-0705}{}

\author{Sebastian Rudolph}{Computational Logic Group, TU Dresden, Germany \and Center for Scalable Data Analytics and Artificial Intelligence Dresden/Leipzig, Germany \and \url{http://sebastian-rudolph.de}}{sebastian.rudolph@tu-dresden.de}{https://orcid.org/0000-0002-1609-2080}{European Research Council, Consolidator Grant DeciGUT (771779)}

\authorrunning{L. Herrmann, V. Peth, and S. Rudolph} %TODO mandatory. First: Use abbreviated first/middle names. Second (only in severe cases): Use first author plus 'et al.'

\Copyright{Luisa Herrmann, Vincent Peth, and Sebastian Rudolph} %TODO mandatory, please use full first names. LIPIcs license is "CC-BY";  http://creativecommons.org/licenses/by/3.0/


 %TODO mandatory; please add comma-separated list of keywords

\category{} %optional, e.g. invited paper

\relatedversion{} %optional, e.g. full version hosted on arXiv, HAL, or other respository/website
\funding{BMBF (SCADS22B) and SMWK by funding ScaDS.AI Dresden/Leipzig}%optional, to capture a funding statement, which applies to all authors. Please enter author specific funding statements as fifth argument of the \author macro.

\EventEditors{Sylvain Lombardy, Jérôme Leroux, and David Peleg}
\EventNoEds{2}
\EventLongTitle{48th International Symposium on Mathematical Foundations of Computer Science (MFCS 2023)}
\EventShortTitle{MFCS 2023}
\EventAcronym{MFCS}
\EventYear{2023}
\EventDate{August 28--September 01, 2023}
\EventLocation{Bordeaux, France}
\EventLogo{}
\SeriesVolume{YY}
\ArticleNo{XX}
\input{notation}

\begin{document}

\maketitle

\input{abstract}

\section{Introduction}\label{sec:Intro}
\input{introduction}

\section{Preliminaries}\label{sec:Prelims}
\input{preliminaries}

\section{Syntax and Semantics of \thegoodlogicheading}\label{sec:SyntaxSemantics}
\input{thelogic}

\input{example}

\section{Mildly Extending \thegoodlogicheading
%\thegoodlogic 
Leads to Undecidability}\label{sec:Undecidable}
\input{undecidable}

\section{Transformation into Normal Form}

\input{normalization}

%\section{Parikh MSO Logic}

%\input{luisalogic}

\section{Parikh-Muller Tree Automata}

\input{automata}

\begin{corollary}\label{cor:treedecidable}
Satisfiability of \thegoodlogic over labeled infinite binary trees is decidable.
\end{corollary}

\section{Decidability over Tree-Interpretable Classes of Structures}\label{sec:MSOInter}
\input{treeinterpretable}

\section{Incorporating Two-Variable-Logics without Width Restrictions}\label{sec:twovar}
\input{twovarlogics}

\section{Showcase: Decidability of Tame Satisfiability of the Fully Enriched µ-Calculus with Global Presburger Counting}
\input{injectivePDL} %did not compile %% SR: sorry, I was interrupted mid-writing

\section{Conclusion}\label{sec:Conc}
\input{conclusion}

\bibliography{lit.bib}

\newpage
\appendix

\input{app-undecidable}
\newpage
\input{app-normalization}
\newpage
\input{app-automata}
\newpage
\input{app-interpretations}
\newpage
\input{app-C2}
\newpage
\input{app-FEmu}

\end{document}

%% file: notation.tex
\usepackage{amsthm}
\usepackage{nimbusmononarrow} % a better distinguishable \mathtt
\usepackage[scr=boondoxo]{mathalpha} % lower case calligraphic letters
\usepackage{xspace}
\usepackage{physics}
\usepackage{mathtools}
\usepackage{stmaryrd}
\usepackage{upgreek}

\newcommand{\semof}[1]{\llbracket#1\rrbracket}

\newcommand{\consistent}{\mathrm{consistent}}
\newcommand{\FOpres}{\mathrm{FO}^2_{\mathrm{Pres}}}%not in mathcal in the original

\newcommand{\Nat}{\mathbb{N}}
\newcommand{\pos}{\mathrm{pos}}
\newcommand{\Run}{\mathrm{Run}}
\newcommand{\cnt}{\mathrm{cnt}}
\renewcommand{\inf}{\mathrm{inf}}

\newcommand{\A}{\mathcal{A}}
\newcommand{\Q}{{}Q{}}
\newcommand{\F}{\mathcal{F}}
\newcommand{\V}{\mathcal{V}}
\newcommand{\Pow}{\mathcal{P}}

\renewcommand{\L}{\mathcal{L}}

\newcommand{\To}{T^\omega}
\renewcommand{\rho}{\varrho}
\newcommand{\free}{\mathrm{free}}
\newcommand{\proj}{\mathrm{pr}}

\newcommand{\true}{\mathbf{true}}
\newcommand{\false}{\mathbf{false}}

\newcommand{\TNF}[1]{\mathsf{NF}(#1)}

\newcommand{\Varset}{\mathbf{V}}
\newcommand{\Vind}{\Varset_\mathrm{\!ind}}
\newcommand{\Vset}{\Varset_\mathrm{\!set}}
\newcommand{\Vnum}{\Varset_\mathrm{\!num}}

\newcommand{\Consts}{\mathbf{C}}
\newcommand{\Preds}{\mathbf{P}}

\newcommand{\Sig}{\mathbb{S}}
\newcommand{\SigC}{\Sig_\Consts}
\newcommand{\SigP}{\Sig_\Preds}

\newcommand{\iTerms}{\mathbf{I}}

\newcommand{\sTerms}{\mathbf{S}}
\newcommand{\nTerms}{\mathbf{N}}

\renewcommand{\Form}{\mathbf{F}}

\newcommand{\pr}[1]{\mathtt{#1}}
\newcommand{\nv}[1]{\mathscr{#1}}

\newcommand{\bin}[1]{\langle #1\rangle_{\mathrm{bin}}}

\newcommand{\nrm}{\mathsf{smpl}}
\newcommand{\nrmo}{\nrm^*}
\newcommand{\transf}{\mathsf{transf}}

\renewcommand{\int}[1]{#1^{\mathfrak{A},\nu}}
\newcommand{\plus}{\mathrel{\mbox{\rm\texttt{+}}}}
\newcommand{\tteq}{\mathrel{\mbox{\rm\texttt{=}}}}
\newcommand{\tteqfin}{\mathrel{\mbox{\rm\texttt{=}}_\texttt{\rm\tt fin}}} 
\newcommand{\ttleq}{\mathrel{\mbox{\raisebox{-0.6ex}{\rm\texttt{-}}\!\!\!\raisebox{0.1ex}{\rm\texttt{<}}}}}
\newcommand{\ttleqfin}{\mathrel{\mbox{\raisebox{-0.6ex}{\rm\texttt{-}}\!\!\!\raisebox{0.1ex}{\rm\texttt{<}}}_\texttt{\rm\tt fin}}}
\newcommand{\ttleqffin}{\mathrel{\mbox{\raisebox{-0.6ex}{\rm\texttt{-}}\!\!\!\raisebox{0.1ex}{\rm\texttt{<}}}_{(\texttt{\rm\tt fin})}}}

\newcommand{\ttinfty}{\mbox{\rm\texttt{o\hspace{-0.3ex}o}}}
\newcommand{\cnts}[1]{\mbox{\rm\texttt{\#}}\hspace{-0.2ex}#1}
\newcommand{\cntp}[1]{\mbox{\rm\texttt{\#}}\hspace{-0.2ex}(#1)}

\newcommand{\thelogic}{$\omega$MSO$\cdot$BAPA{}\xspace}

\newcommand{\thegoodlogic}{\mbox{\rm$\omega$MSO$\hspace{-0.2ex}\Join\hspace{-0.2ex}$BAPA{}}\xspace}
\newcommand{\thegoodlogicheading}{$\mathbf{\boldsymbol\omega}$MSO$\hspace{-0.2ex}\Join\hspace{-0.2ex}$BAPA{}\xspace}

\newcommand{\bbr}[1]{\mbox{\small\bf\{}#1\mbox{\small\bf\}}}
\newcommand{\cmpr}[2]{\bbr{{#1{\,\mbox{\small\bf |}\,#2}}}}

\usepackage{mathtools}

%% file: abstract.tex
% When copying the abstract in a text form, use the unicode version
%
%   ωMSO⋈BAPA
%
% for our logic

\begin{abstract}
We propose \thegoodlogic, an expressive logic for describing countable structures, which subsumes and transcends both Counting Monadic Second-Order Logic (CMSO) and Boolean Algebra with Presburger Arithmetic (BAPA). 
We show that satisfiability of \thegoodlogic is decidable over the class of labeled infinite binary trees, whereas it  becomes undecidable even for a rather mild relaxations.
The decidability result is established by an elaborate multi-step transformation into a particular normal form, followed by the deployment of \emph{Parikh-Muller Tree Automata}, a novel kind of automaton for infinite labeled binary trees, integrating and generalizing both Muller and Parikh automata while still exhibiting a decidable (in fact \textsc{PSpace}-complete) emptiness problem. By means of MSO-interpretations, we lift the decidability result to all tree-interpretable classes of structures, including the classes of finite/countable structures of bounded treewidth/cliquewidth/partitionwidth. We generalize the result further by showing that decidability is even preserved when coupling width-restricted \thegoodlogic with width-unrestricted two-variable logic with advanced counting. A final showcase demonstrates how our results can be leveraged to harvest decidability results for expressive $\mu$-calculi extended by global Presburger constraints.
\end{abstract}

% We propose ωMSO⋈BAPA, an expressive logic for describing countable structures, which subsumes and transcends both Counting Monadic Second-Order Logic (CMSO) and Boolean Algebra with Presburger Arithmetic (BAPA). 
% We show that satisfiability of ωMSO⋈BAPA is decidable over the class of labeled infinite binary trees, whereas it becomes undecidable even for a rather mild relaxations. 
% The decidability result is established by an elaborate multi-step transformation into a particular normal form, followed by the deployment of Parikh-Muller Tree Automata, a novel kind of automaton for infinite labeled binary trees, integrating and generalizing both Muller and Parikh automata while still exhibiting a decidable (in fact PSpace-complete) emptiness problem. By means of MSO-interpretations, we lift the decidability result to all tree-interpretable classes of structures, including the classes of finite/countable structures of bounded treewidth/cliquewidth/partitionwidth. We generalize the result further by showing that decidability is even preserved when coupling width-restricted ωMSO⋈BAPA with width-unrestricted two-variable logic with advanced counting. A final showcase demonstrates how our results can be leveraged to harvest decidability results for expressive μ-calculi extended by global Presburger constraints.

%% file: introduction.tex
\newcommand{\highlight}[1]{\emph{#1}}

\highlight{Monadic second-order logic} (MSO) is a popular, expressive, yet computationally reasonably well-behaved logical formalism to deal with various classes of finite or countable structures. 
%MSO plays an important role in virtually all subdisciplines of computational logic. 
%Going beyond first-order logic, 
It allows for expressing ``mildly recursive'' structural properties like connectedness or reachability, which go beyond first-order logic yet 
meet crucial modeling demands in verification, database theory, knowledge representation, and other fields of computational logic. The well-understood link between MSO and automata theory has been very fertile in theory and practice.
%In particular, the MSO theory of infinite binary trees is decidable by Rabin's famous result \cite{Rab69}, and the same holds for structures of bounded treewidth, cliquewidth, and partitionwidth.    

Unfortunately, MSO's native capabilities to express cardinality relationships are very limited; they are essentially restricted to fixed thresholds (e.g. ``there are at least 10 leaves'').  
\highlight{Counting MSO} \cite{Cou90,Cou88}, denoted CMSO, extends MSO by modulo counting and a finiteness test over sets (e.g. ``there is an even number of nodes''), which increases expressiveness in general, while over finite and infinite words or trees, CMSO can be simulated in plain MSO. In contrast, enriching MSO with \highlight{cardinality constraints} \cite{KlaRue02,KlaRue03} (as in ``all nodes have as many incoming as outgoing edges'') increases the expressivity drastically, but causes satisfiability to become undecidable even over finite words. Decidability (over finite words, trees, or graphs of bounded treewidth \cite{KotVeiZul16}) can be recovered when confining set variables occurring in cardinality constraints to those existentially quantified in front (MSO$^{\exists\text{Card}}$). One way to show this is through \highlight{Parikh automata} extending finite automata by adding finitely many counters and exploiting the relationship of Presburger arithmetic and semilinear sets \cite{GinSpa66}. 

Very recent work \cite{GuhJecLeh22,GroSabSie23,GroSie23} extended Parikh word automata to infinite words and inves-\\
tigated the impact of 
various acceptance conditions, but left a logical characterization as open question. 
As with the original Parikh automata, one central motivation behind these works is to provide automata-based approaches for specifying and verifying systems beyond regular languages. The study of $\omega$-Parikh automata is motivated by  reactive systems, whose behaviors are typically represented by infinite words. Then, the plethora of branching-time approaches in verification should call for a further generalization to $\omega$-tree-automata. Yet, to our knowledge, Parikh automata have not been studied in the context of infinite trees so far.

%\todo{Mention "open problem" of logical charac. stated in \cite{GroSie23}?} %\todo{Eigentlich wäre dazu ein eigener, ausführlicherer Absatz schön. Aber vermutlich fehlt der Platz.}
% However, different approach: Parikh condition tested infinitely often, thus (for stronger acceptance conditions than Büchi acceptance) undecidable. In \cite{GroSie23}, similar approach to consider $\omega$-Parikh word automata where Parikh condition is tested on finite prefix. Open problem regarding logical characterization was stated -- answered (and generalized) here. --> currently quite active research field

Another, orthogonal logical approach for describing sets and their cardinalities, motivated by tasks from program analysis and verification, combines the first-order theory of \highlight{Boolean algebras} (BA) with \highlight{Presburger arithmetic} (PA), resulting in the theory of \highlight{BAPA} \cite{KunNguRin05,KunNguRin06}.
As opposed to computationally benign extensions of MSO, BAPA provides stronger support for arithmetic (so one can talk about ``all selections with the same number of blue and red nodes'' or even ``all selections with a share of 70\%\,--\,80\% red nodes'', modeling statistical information).  
BAPA usually assumes a finite universe, but can be extended to the countable setting \cite{LasStu09}; satisfiability is decidable in either case. However, very regrettably, BAPA lacks non-unary relations, which is outright fatal when it comes to expressing structural properties. 

Combining both worlds, we introduce \thegoodlogic 
\raisebox{-0.6ex}{\includegraphics[scale=0.84]{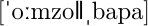}}\,,
a logic for countable structures, which extends CMSO by BAPA's set operations and Presburger statements, strictly contains MSO$^{\exists\text{Card}}$, and allows for sophisticated structural-arithmetic statements (Section 3). 
We warrant computational manageability by gently controlling the usage of variables, noting that satisfiability turns undecidable otherwise (Section 4). 
Exhibiting an elaborate transformation (Section 5), we prove that \thegoodlogic formulae over trees can be brought into a very restricted \emph{tree normal form} (TNF). 
%of the shape
%$
%\textstyle \exists X_1.\cdots\exists X_n. \bigvee_{i=1}^k \big( \varphi_i \wedge \bigwedge_{j=1}^{l_i} \chi_{i,j} \big), 
%$
%where $\varphi_i$ is a plain MSO formulae and $\chi_{i,j}$ are restricted Presburger atoms. 
We then provide a characterization showing that the sets of $\omega$-trees satisfying TNF formulae coincide with the sets of trees recognized by \emph{Parikh-Muller Tree Automata} (PMTA), a novel automata model designed by us -- and
  the first-ever automaton model on infinite trees capable of testing Parikh conditions (Section 6). PMTA generalize both Muller and Parikh automata and their emptiness is decidable. 
The decidability of \thegoodlogic over the class of labeled infinite binary trees thereby obtained is then lifted to all \highlight{tree-interpretable classes}, including vast and practically relevant classes of finite or countable structures that are bounded in terms of certain width measures (Section 7). Such width-bounded \thegoodlogic can be decidably coupled with width-unbounded two-variable logics with advanced counting (Section 8). We demonstrate how to leverage our results to gain decidability results for statistics-enhanced formalisms of the $\mu$-calculus family, which subsumes branching-time logics such as CTL$^*$ (Section 9).

%In short, this paper introduces a novel, highly \highlight{expressive logic} (with native CMSO, BAPA, and MSO$^{\exists\text{Card}}$ support), and proves \highlight{satisfiability decidable} over a broad range of classes of \emph{countable} structures by establishing a \highlight{characterization} in terms of the first-ever \emph{automaton model} on infinite trees capable of testing Parikh conditions.
%Presburger Arithmetic + Infinity: \cite{LasStu09}

%% file: preliminaries.tex
%\paragraph*{Counting to Infinity}
As usual, for any $n \in \mathbb{N}$, let $[n]\coloneqq \{1,\ldots,n\}$. 
In order to \textbf{count to infinity}, we %will %frequently 
use %the natural numbers 
$\mathbb{N}$ extended by (countable) infinity $\infty$, with arithmetics lifted
%addition, multiplication and comparison 
in the usual way;
 in particular, $\infty \plus n = \infty \plus \infty = (n+1) \cdot \infty = \infty$ and $0 \cdot \infty = 0$ as well as $n \leq \infty$ and $\infty \leq \infty$.
For countable sets $A$, let $|A|$ denote the element of $\mathbb{N} \cup \{\infty\}$ that corresponds to their cardinality.

%\paragraph*{Countable Structures}

To define \textbf{countable structures}, assume the following countable, pairwise disjoint sets: 
%\begin{itemize}
%\item 
a set $\Consts$ of \emph{(individual) constants}, denoted by $\pr{a},\pr{b},\pr{c},...$, and,
%\item 
for every $n \in \Nat$, a set $\Preds_n$ of \emph{$n$-ary predicates}, denoted by $\pr{P},\pr{R},\pr{Q},...$.
The set of all predicates will be denoted by $\Preds := \bigcup_{i\in \mathbb{N}}\Preds_n$, and we let $\mathrm{ar}:\Preds \to \mathbb{N}$ such that $\mathrm{ar}(\pr{Q})=n$ iff $\pr{Q} \in \Preds_n$.
%\end{itemize}
A \emph{(relational) signature} $\Sig$ is a union $\SigC \cup \SigP$ of finite subsets of $\Consts$ and $\Preds$, respectively.
An \emph{$\Sig$-structure} is a pair $\mathfrak{A} = (A, \cdot^\mathfrak{A})$, where $A$ is a countable, nonempty set, called the \emph{domain} of $\mathfrak{A}$ and $\cdot^\mathfrak{A}$ is a function that maps
%\begin{itemize}
%    \item 
every $\pr{a} \in \SigC$ to a domain element $\pr{a}^\mathfrak{A} \in A$, and
%    \item 
every $\pr{Q} \in \SigP$ to an $\mathrm{ar}(\pr{Q})$-ary relation $\pr{Q}^\mathfrak{A} \subseteq A^{\mathrm{ar}(\pr{Q})}$.
%\end{itemize}

%\paragraph*{Infinite Trees}

We define \textbf{infinite trees} starting from a finite, non-empty set $\Sigma$, called \emph{alphabet}.
A \emph{(full) infinite binary tree} (often simply called a \emph{tree}) labeled by some alphabet $\Sigma$ is a mapping 
$\xi\colon\{0,1\}^*\to \Sigma$. We denote the set of all trees labeled by $\Sigma$ by $T^\omega_\Sigma$. A \emph{finite tree} is a mapping
\noindent$\xi\colon X\to \Sigma$ where $X$ is a finite, prefix-closed subset of $\{0,1\}^*$. The set of all finite trees over $\Sigma$ will be denoted by $T_\Sigma$. We sometimes refer to the domain $X$ of $\xi$ by $\pos(\xi)$, whose elements we call \emph{positions} or \emph{nodes} of $\xi$. Given a tree $\xi\in T_\Sigma^\omega$ and a finite, prefix-closed set $X\subseteq\{0,1\}^*$, we denote by $\xi_{|X}$ the finite tree in $T_\Sigma$ that has $X$ as domain and coincides with $\xi$ on $X$.

An \emph{(infinite) path} $\pi$ is an infinite sequence $\pi=\pi_1\pi_2\ldots$ of positions from $\{0,1\}^*$ such that $\pi_1=\varepsilon$ and, for each $i\geq 1$, $\pi_{i+1}\in(\pi_i\cdot\{0,1\})$. Given a tree $\xi\in T^\omega_\Sigma$ and a path $\pi$, we denote by $\xi(\pi)$ the infinite word $\xi(\pi_1)\xi(\pi_2)\ldots$ obtained by concatenating the labels of $\xi$ along $\pi$. We denote by $\inf(\xi(\pi))$ the set of all labels occurring infinitely often in $\xi(\pi)$.

We will also find it convenient to represent trees over some given alphabet $\Sigma=\{a_1,\ldots,a_n\}$ as structures over the signature $\Sig = \SigP = \{\succ_0,\succ_1,\pr{P}_{a_1},\ldots,\pr{P}_{a_n}\}$: Thereby, a tree $\xi\in T^\omega_\Sigma$ will be represented by the structure $\mathfrak{A}_\xi$ with $A_\xi= \{0,1\}^*$, where $\succ_0^{\mathfrak{A}_\xi}=\{(w,w0)\mid w\in\{0,1\}^*\}$ and $\succ_1^{\mathfrak{A}_\xi}=\{(w,w1)\mid w\in\{0,1\}^*\}$ while $\pr{P}_{a_i}^{\mathfrak{A}_\xi}=\{u\in\{0,1\}^*\mid \xi(u)=a_i\}$ for each $i\in[n]$. When there is no danger of confusion, we will simply write $\xi$ instead of $\mathfrak{A}_\xi$.

%% file: thelogic.tex
We now introduce the logic \thegoodlogic. The underlying ``design principles'' for this logical formalism are to have a language that syntactically subsumes and tightly integrates CMSO and BAPA, while still exhibiting favorable computational properties, even over countably infinite structures. To this end, we will first define the language \thelogic and then obtain \thegoodlogic by imposing some syntactic restrictions on the usage of variables.

\begin{definition}[Syntax of \thelogic]
Given a signature $\Sig = \SigC \cup \SigP$, together with three countable and pairwise disjoint sets 
$\Vind$ of \emph{individual variables} (denoted $x,y,z,...$),
$\Vset$ of \emph{set variables} (denoted $X,Y,Z,...$), and
$\Vnum$ of \emph{number variables} (denoted $\nv{k},\nv{l},\nv{m},\nv{n}...$),
%\begin{itemize}
%\item $\Vind$ of \emph{individual variables}, denoted by $x,y,z,...$,
%\item $\Vset$ of \emph{set variables}, denoted by $X,Y,Z,...$, and
%\item $\Vnum$ of \emph{number variables}, denoted by $\nv{k},\nv{l},\nv{m},\nv{n}...$,
%\end{itemize}
we define the following sets of expressions by mutual induction:
\begin{itemize}
    \item the set $\iTerms$ of \emph{individual terms}:\qquad  $\iota \Coloneqq{} \pr{a} \mid x $
    \item the set $\sTerms$ of \emph{set terms} (\/$\pr{P}$ being a unary predicate):\quad $S \Coloneqq{}\ \{\pr{a}\} \mid \pr{P} \mid X \mid S^c \mid S_1 \cap S_2 \mid S_1 \cup S_2 $
    \item the set $\nTerms$ of \emph{number terms}:\footnote{We will consider number terms obtainable from each other through basic transformations (reordering, factoring, summarizing, rules for $\infty$) as syntactically equal, allowing us to focus on simplified expressions.} \quad $ t \Coloneqq{} \pr{n} \mid \ttinfty \mid \nv{k} \mid \cnts{S} \mid  \pr{m}\,t \mid t_1 \plus t_2$\\
    (with $n \in \mathbb{N}$ and $m \in \mathbb{N}\setminus\{0\}$; we use typewriter font to indicate that we mean an explicit representation of a constant natural number $n$ or $m$ rather than the symbol ``\hspace{1pt}$n$\!'' or ``\hspace{1pt}$m$\!'')
    \item the set $\Form$ of (unrestricted) \emph{formulae}:
    \begin{align*}
    \varphi \Coloneqq{}\ & \mathtt{Q}(\iota_1,\ldots,\iota_n) \mid S(\iota) \mid %\iota \tteq \iota' \mid  % REMOVED BECAUSE EXPRESSIBLE ANYWAYS
    t \ttleqfin t' \mid t \ttleq t' \mid \cnts{S} \equiv_\pr{n} \pr{m} \mid \mathrm{Fin}(S) \mid \mathbf{true} \mid \mathbf{false} \mid\\ 
    & \neg \varphi \mid \varphi \wedge \varphi' \mid \varphi \vee \varphi' \mid \exists x.\varphi \mid \forall x.\varphi \mid \exists X.\varphi \mid \forall X.\varphi \mid \exists \nv{k}.\varphi \mid \forall \nv{k}.\varphi
    \end{align*}
The first six types of atomic formulae will be referred to as \emph{predicate atoms}, \emph{set atoms}, \emph{classical Presburger atoms}, \emph{modern Presburger atoms}, \emph{modulo atoms}, and \emph{finiteness atoms}, respectively.
We use \emph{Presburger atoms} and write $t \ttleqffin t'$ to jointly refer to the classical and modern variants. A Presburger atom $t \ttleqffin t'$ is called \emph{simple}, if it contains at most one occurrence of a term of the shape $\cnts{S}$ and no occurrences of number variables.
\end{itemize}
%The \emph{signature of a \thelogic formula $\varphi$}, denoted $\Sig(\varphi)$, contains exactly the constants and predicates occurring in $\varphi$.  
\end{definition}

Note that, for notational homogeneity, we choose to write $X(\iota)$ instead of $\iota \in X$. Where convenient% 
%and when negation-normalization is not required
, we will also make use of the Boolean connectives $\Rightarrow$ and $\Leftrightarrow$ as abbreviations with the usual meaning.
While the original syntax of \thelogic does not provide an explicit equality predicate, both individual and set equality can be expressed (see further below). 
%For instance, equality of $\iota$ and $\iota'$ can be expressed by $\forall Z.(\neg Z(\iota) \vee Z(\iota'))$.

\begin{definition}[Semantics of \thelogic]
%A signature is a finite subset of $\Consts \cup \Preds$
A \emph{variable assignment} (for a structure $\mathfrak{A}$) is a function $\nu$ that maps
\begin{itemize}
    \item every individual variable $x \in \Vind$ to a domain element $\nu(x) \in A$,
    \item every set variable $X \in \Vset$ to a subset $\nu(X) \subseteq A$ of the domain,  and
    \item every number variable $\nv{k} \in \Vnum$ to a number $\nu(\nv{k}) \in \mathbb{N} \cup \{\infty\}$.
\end{itemize}
We write $\nu_{\scriptscriptstyle x{\mapsto}a}$, $\nu_{\scriptscriptstyle X{\mapsto}A'}$, and $\nu_{\scriptscriptstyle \nv{k}{\mapsto}n}$ to denote $\nu$ updated in the way indicated in the subscript.

Given an interpretation $\mathfrak{A}$ and a variable assignment $\nu$, we let the function $\int{\cdot}$ map
\begin{itemize}
    \item $\iTerms$ to $A$ by letting $\int{\pr{a}} = {}\pr{a}^\mathfrak{A}$ and $\int{x} = {}\nu(x)$,
    \item $\sTerms$ to $2^A$ by letting
\newcommand{\eqarray}[1]{
\begin{array}{r@{\  = {} \ }l}
#1
\end{array}
}
$$
\eqarray{
\int{\{\pr{a}\}} &  \{\int{\pr{a}}\} \\
\int{\pr{P}} & \pr{P}^\mathfrak{A} \\ 
}
\qquad
\eqarray{
    \int{X} & \nu(X)  \\
    \int{(S^c)} & A \setminus \int{S} \\
}
\qquad
\eqarray{
    \int{(S_1 \cap S_2)} & \int{S_1} \cap \int{S_2} \\ 
    \int{(S_1 \cup S_2)} & \int{S_1} \cup \int{S_2} \\
}
$$
%    \begin{align*}
%    \int{\{\pr{a}\}} &  = {}  \{\int{\pr{a}}\} \\
%    \int{\pr{P}} &  = {} \pr{P}^\mathfrak{A} \\ 
%    \int{X} &  = {} \nu(X)  \\
%    \int{(S^c)} &  = {} A \setminus \int{S} \\
%    \int{(S_1 \cap S_2)} &  = {} \int{S_1} \cap \int{S_2} \\ 
%    \int{(S_1 \cup S_2)} &  = {} \int{S_1} \cup \int{S_2} \\
%    \end{align*}
\item $\nTerms$ to $\mathbb{N}\cup \{\infty\}$ by letting
$$
\eqarray{
    \int{\pr{n}}  & n \\
    \int{\ttinfty}  & \infty \\ % ∞ 
}
\qquad
\eqarray{
    \int{\nv{k}}  &  \nu(\nv{k}) \\ 
    \int{(\cnts{S})}  &  |\int{S}| \\  
}
\qquad
\eqarray{
    \int{(\pr{n}\, t)}  & n \cdot \int{t} \\
    \int{(t_1 \plus t_2)}  &  \int{t_1} + \int{t_2} \\
}
$$
%    \begin{align*}
%    \int{\pr{n}}  &  = {} n \\
%    \int{\ttinfty}  &  = {} \infty \\ % ∞ 
%    \int{\nv{k}}  &  = {}  \nu(\nv{k}) \\ 
%    \int{(\cnts{S})}  &  = {}  |\int{S}| \\  
%    \int{(\pr{n}\, t)}  &  = {} n \cdot \int{t} \\
%    \int{(t_1 \plus t_2)}  &  = {}  \int{t_1} + \int{t_2} \\
%    \end{align*}
\end{itemize}

\noindent Finally we define satisfaction of formulae from $\Form$ as follows: $\mathfrak{A}, \nu$ satisfies
$$
\begin{array}{@{}l@{\ \text{iff}\ \ }l}
\pr{Q}(\iota_1,\!...,\iota_n) &  (\int{(\iota_1)}\!\!,\!...,\int{(\iota_n)}) \in \pr{Q}^\mathfrak{A} \\
S(\iota) & \int{\iota} \in \int{S} \\
%\iota_1 \tteq \iota_2 & \int{\iota_1} = \int{\iota_2}\\
t_1 \ttleq t_2 & \int{t_1} \leq \int{t_2} \\
t_1 \ttleqfin t_2 & \int{t_1} \leq \int{t_2} < \infty \\
\cnts{S} \equiv_\pr{n} \pr{m} & \begin{array}[t]{@{}l@{}}
         \int{(\cnts{S})} = m \mod n\\ 
         \text{and }\int{(\cnts{S})} < \infty \\
        \end{array}\\     
\mathrm{Fin}(S) & |\int{S}| < \infty \\
\neg \varphi & \mathfrak{A},\nu \not\models \varphi \\
\end{array}\quad
\begin{array}{@{}l@{\ \text{iff}\ \ }l}
\varphi_1 \!\wedge\! \varphi_2 & \mathfrak{A},\nu \models \varphi_1 \text{ and \ } \mathfrak{A},\nu \models \varphi_2 \\
\varphi_1 \!\vee\! \varphi_2 & \mathfrak{A},\nu \models \varphi_1 \text{ or \ } \mathfrak{A},\nu \models \varphi_2 \\
\exists x.\varphi & \mathfrak{A},\nu_{\scriptscriptstyle x{\mapsto}a} \models \varphi \text{ \ for some }a{\,\in\,} A\\
\forall x.\varphi & \mathfrak{A},\nu_{\scriptscriptstyle x{\mapsto}a} \models \varphi \text{ \ for all }a{\,\in\,} A\\
\exists X.\varphi & \mathfrak{A},\nu_{\scriptscriptstyle X{\mapsto}A'} \models \varphi \text{ \ for some }A'{\,\subseteq\,} A\\
\forall X.\varphi & \mathfrak{A},\nu_{\scriptscriptstyle X{\mapsto}A'} \models \varphi \text{ \ for all }A'{\,\subseteq\,} A\\
\exists \nv{k}.\varphi & \mathfrak{A},\nu_{\scriptscriptstyle \nv{k}{\mapsto}n} \models \varphi \text{ \ for some }n{\,\in\,} \mathbb{N} {\,\cup\,} \{\infty\}\\
\forall \nv{k}.\varphi & \mathfrak{A},\nu_{\scriptscriptstyle \nv{k}{\mapsto}n} \models \varphi \text{ \ for all }n {\,\in\,} \mathbb{N} {\,\cup\,} \{\infty\}\\
\end{array}
$$
Plus, we always let $\,\mathfrak{A},\nu \models\mathbf{true}$ and $\,\mathfrak{A},\nu \not\models\mathbf{false}$. 
For a formula $\varphi$, its \emph{free variables} (denoted $\free(\varphi)$) are defined as usual; $\varphi$ is a \emph{sentence} if $\free(\varphi)=\emptyset$. 
For sentences, $\nu$ does not influence satisfaction, which allows us to write $\mathfrak{A} \models \varphi$ and call $\mathfrak{A}$ a \emph{model} of $\varphi$ in case  $\mathfrak{A},\nu \models \varphi$ holds for any $\nu$. We call $\varphi$ \emph{satisfiable} if it has a model. 
\end{definition}

\begin{definition}[Syntax of \thegoodlogic]
From now on, we will make the following assumption (which is easily obtainable via renaming): %All considered formulae are in NNF; 
In every formula, all quantifications use different variable names and these are disjoint from the names of free variables. 
Given an \thelogic formula $\varphi$ satisfying this assumption, we analyze its constituents as follows: 
\begin{itemize}
    \item A (set or individual) variable is called \emph{assertive}, if it is free, or it is existentially quantified and the quantification does not occur inside the scope of a negation or of a universal (set, individual, or number) quantifier.
    \pagebreak
    %Moreover, every individual constant $\pr{a}$ and every unary predicate $\pr{P}$ is assertive. 
    \item The set of \emph{delicate} individual and set variables is the smallest set of (non-assertive) variables satisfying the following: 
\begin{itemize}
\item Every non-assertive set variable occurring in a non-simple Presburger atom is delicate.
\item If some atom contains a delicate (individual or set) variable, then all of this atom's non-assertive (individual or set) variables are delicate. 
%
%
%\item Whenever two distinct non-assertive individual variables $x$ and $y$ co-occur in a predicate atom $\pr{Q}(\,\cdots\,)$ then both $x$ and $y$ are tethered.
%\item If some atom contains a tethered (individual or set) variable, then all of this atom's non-assertive (individual or set) variables are tethered. 
%%%\item if $\phi$ contains an equality atom $x = y$ or $y = x$, where $x$ is dependent and $y$ is non-assertive, then $y$ is dependent.
%%%% OBSOLETE BECAUSE CAUGHT BY MORE GENERAL CASE ABOVE
%%\item if a set term $S$ contains some tethered $X$, then every non-assertive $Y$ occurring in $S$ is tethered.
%%\item if in a set atom $S(x)$, $x$ is tethered and $S$ contains some non-assertive $Y$, then $Y$ is tethered.
%%\item if in a set atom $S(x)$, $x$ is non-assertive and $S$ contains some tethered $Y$, then $x$ is tethered.
\end{itemize}
\end{itemize}
Then, $\varphi$ is an \emph{\thegoodlogic formula} iff each of its predicate atoms $\pr{Q}(\,\cdots\,)$ contains at most one delicate variable (possibly in multiple occurrences). 
\end{definition}

%We point out that the definition of \thegoodlogic generally allows tethered variables to occur in modulo atoms $\cnts{S} \equiv_\pr{n} \pr{m}$ and finiteness atoms $\mathrm{Fin}(S)$. That is, the usage of tethered variables is only prohibited in non-simple Presburger atoms $t \ttleqffin t'$.

It is easy to see that, despite the above restrictions, \thegoodlogic entirely encompasses CMSO and MSO$^{\exists\text{Card}}$ (no delicate variables) as well as BAPA (no predicates of arity ${>}1$).
%
%In the following, we find it convenient to extend the definition of set terms by ``set comprehension expressions'' of the form
%$$
%\texttt{\{} x \mid \varphi \texttt{\}}
%$$
%for $x \in \Vind$ and $\varphi \in \Form$. Any atom, $\psi(\upsilon)$ wherein such an expression $\upsilon$ occurs is meant to abbreviate 
% 

\noindent For convenience and better readability, we will make use of the following abbreviations. 
$$
\begin{array}{rl}
    x=y \ \  &\coloneqq \forall Z. Z(x) \Leftrightarrow Z(y)\\
    S \neq \emptyset\ \ &\coloneqq \exists z. S(z)\\  
    S\subseteq S' &\coloneqq \forall z. S(z)\Rightarrow S'(z)\\
    S = S' &\coloneqq (S\subseteq S') \wedge (S'\subseteq S)\\
\end{array}
\quad
\begin{array}{rl}
    \exists x{\,\in\,} S.\varphi &\coloneqq \exists x.S(x)\land\varphi\\
    \forall x{\,\in\,} S.\varphi &\coloneqq \forall x.S(x)\Rightarrow\varphi\\
    t \tteq t' &\coloneqq (t \ttleq t') \wedge (t' \ttleq t)\\ 
    t \tteqfin t' &\coloneqq (t \ttleqfin t') \wedge (t' \ttleqfin t)\\ 
\end{array}
$$
%The chosen abbreviations ensure that when used in NNF formulae, the expanded formula will be in NNF as well. Moreover, 
%For brevity, we use $\neq$, $\not\subseteq$, $\not\,\tteq$, and $\not\,\tteqfin$ to denote negated versions.  
An analysis of these abbreviations reveals that \thegoodlogic allows for the variables $x,y$ and set variables in $S,S'$ in these abbreviations to be delicate. 
We will also employ shortcuts specific to the signature $\{\succ_0,\succ_1,\pr{P}_{a} \mid a \in \Sigma\}$  for $\Sigma$-labeled trees. Contrary to above, in these shortcuts, $x$, $y$, $X$, $Y$ must not be delicate to warrant inclusion in \thegoodlogic (Obs. $\dagger$). 
$$
\begin{array}{@{}r@{\ }l@{}}
    X(x.i) &\coloneqq \ \exists y. x\succ_i y \land X(y)\\
\varphi_{\uparrow\mathrm{clsd}}(X)&\coloneqq \ \forall z. X(z.0) \vee X(z.1) \Rightarrow X(z)\\
    x\succ^* y &\coloneqq \  \forall Z. Z(y) \wedge \varphi_{\uparrow\mathrm{clsd}}(Z)  \Rightarrow Z(x)\\
\varphi_\downarrow(x,X) & \coloneqq \ \forall z. \big( X(z) \Leftrightarrow x \succ^* z \big)\\ 
\varphi_\mathrm{path}(X)&\coloneqq \  X\!\neq \emptyset \wedge \varphi_{\uparrow\mathrm{clsd}}(X)\land \forall z{\,\in\,}X.\big(X(z.0) \Leftrightarrow \neg X(z.1)\big)\\
\varphi_\mathrm{inf}(X)&\coloneqq \  \exists Z. \varphi_\mathrm{path}(Z) \land \forall z{\,\in\,} Z.\exists z'{\in}X. (z\succ^+z')\\
\varphi^\cap_\mathrm{inf}(X,Y)&\coloneqq\exists Z. Z\subseteq X\land Z\subseteq Y\land \varphi_\mathrm{inf}(Z)\\
%\varphi_\mathrm{fin}(X)&\coloneqq \  \neg \varphi_\mathrm{inf}(X)
\end{array}\hspace{-10ex} 
\begin{array}{r@{\ }l}
    x\succ y &\coloneqq \ (x\succ_0 y)\lor(x\succ_1 y)\\
\varphi_\mathrm{root}(x)&\coloneqq \  \neg\exists z.(z\succ x)\\
    x \succ^+ y&\coloneqq \  (x\succ^* y)\land (x\neq y)\\
    & \\
    & \\
    & \\
    & \\
\end{array}
$$

%
%
%$$
%\begin{array}{rl}
%\mathrm{inf}(X)&\coloneqq \exists Y. \mathrm{path}(Y) \land (\forall y{\in} Y\exists x{\in}X. y<x)\\
%    \mathrm{inf}(X,Y)&\coloneqq\exists Z. Z\subseteq X\land Z\subseteq Y\land \mathrm{inf}(Z)\\
%    \mathrm{fin}(X)&\coloneqq\neg\mathrm{inf}(X)
%\end{array}
%$$

%% file: example.tex
%Next, we provide an example demonstrating the expressive capabilities of \thegoodlogic.
\begin{example}\label{example}
\newcommand{\ela}{\text{blue}}
\newcommand{\elb}{\text{red}}
\newcommand{\elc}{\text{green}}
\newcommand{\eld}{\text{yellow}}
\newcommand{\ele}{\text{black}}
We use \thegoodlogic to specify the class of all labeled infinite binary trees over the alphabet 
$\Sigma = \{\ela,\elb,\elc,\eld,\ele\}$ 
%$\Sigma = \{\spadesuit, \heartsuit, \diamondsuit, \clubsuit,.\}$ 
satisfying the following property:

\smallskip

{\it
%\begin{quote}
\noindent ``There is a path $X$ and some node $x$ on $X$ such that the following hold:
\begin{enumerate}
    \item For every infinite selection $Y$ of $\ela$ nodes from the $x$-descendants on the path $X$, there is a selection $Y'$ of $\elb$ nodes from the whole tree, such that    
    \begin{enumerate}
    \item $Y$ and $Y'$ contain the same number of nodes with infinitely many $\elc$ descendants,
    \item $Y$ contains twice as many nodes as $Y'$ having less than 10 $\eld$ descendants.
    \end{enumerate}
    \item For every finite selection $Z$ of $\ela$ $x$-descendants, the total number of nodes lying on paths from $x$ to nodes of $Z$ is even.''  
\end{enumerate}
%\end{quote}
}
\vspace{-4ex}
\begin{align*}
& \exists X. \exists x.\varphi_\mathrm{path}(X) \wedge X(x) \wedge \exists V_0. \varphi_\downarrow(x,V_0)\  \wedge\\[-0.5ex]
& \hspace{5.5ex} \bigg( \exists V_1.\big(\forall v_1. V_1(v_1) \Leftrightarrow \exists V^{v_1}_\downarrow. \varphi_\downarrow(v_1,V^{v_1}_\downarrow) \wedge \neg \mathrm{Fin}(V^{v_1}_\downarrow \cap \pr{P}_\elc)\big) \ \wedge\\[-1ex] 
& \hspace{12ex} \exists V_2.\big(\forall v_2. V_2(v_2) \Leftrightarrow \exists V^{v_2}_\downarrow.\varphi_\downarrow(v_2,V^{v_2}_\downarrow) \wedge  \cnts{(V^{v_2}_\downarrow \cap \pr{P}_\eld)} \ttleq \pr{10}\big) \ \wedge\\
& \hspace{16ex}\Big(
\forall Y. \big( \neg\mathrm{Fin}(Y) \wedge Y \subseteq X \cap V_0 \cap \pr{P}_\ela \big) \Rightarrow \\[-2ex]
& \hspace{19ex} \exists Y'. Y' \subseteq \pr{P}_\elb \wedge \cnts{(Y\cap V_1)} \tteq \cnts{(Y'\cap V_1)} \wedge \cnts{(Y\cap V_2)} \tteq \pr{2}\,\cnts{(Y'\cap V_2)} \Big)\! \bigg)
\ \!\wedge\\[-2.5ex]
& \hspace{5.5ex}\bigg( 
\forall Z. \big( \mathrm{Fin}(Z) \wedge Z\subseteq V_0\cap\pr{P}_\ela \big)
\Rightarrow  \\[-2.5ex] 
& \hspace{19ex} \exists V_3.\big(\forall v_3.V_3(v_3) \Leftrightarrow   
( x \succ^+ v_3 \wedge \exists z\in Z.v_3 \succ^* z' )\big) \wedge
\cnts{V_3} \equiv_\pr{2} \pr{0}
\bigg)
\end{align*}
Therein, we use set variables capturing 
all descendants of $x$ ($V_0$); 
all nodes with infinitely many $\elc$ descendants ($V_1$);
all nodes with less than $10$ $\eld$ descendants ($V_2$); and
all nodes between $x$ and elements of $Z$ ($V_3$).
Analysing the variables yields that $X$, $x$, $V_0$, $V_1$, and $V_2$ are assertive, while $Y$ and $Y'$ are delicate due to their occurrence in the non-simple Presburger atoms in the fifth line.
%(the Presburger atom in the third line is simple). 
Delicacy is not inherited further, thus no two delicate variables occur in any predicate atom. Therefore the formula is indeed in \thegoodlogic. % 
%\footnote{
Note that it is crucial that $V_1$ and $V_2$ are defined ``prematurely'' outside the scope of $\forall Y$, so they become assertive and thus their occurrence in the (non-simple) Presburger atoms does not turn them delicate. This technique of ``encapsulating'' unary descriptions into assertive set variables unveils significant additional expressiveness of \thegoodlogic. See also \Cref{sec:Conc} for a discussion on a handier syntax for this.
%} 

%\lhe{I put a very small example here (that probably makes no sense) as I want to see what happens with the independent variables during normalization.}

%\begin{align*}
%    \mathrm{\pr{a}desc}(A,x)\coloneqq&\forall z. z\in A \leftrightarrow (P_\pr{a}(z)\land x<z)\\
%    \varphi\coloneqq&\exists X\exists x\exists WW'. \mathrm{\pr{a}desc}(W,x)\land\mathrm{\pr{b}desc}(W',x) \land\\ 
%    &\forall Z.Z\subseteq W\rightarrow\exists Z'.Z'\subseteq W'\land |X\cap Z|=|X\cap Z'|
%\end{align*}

\end{example}

%% file: undecidable.tex
Just slightly relaxing the syntax of \thegoodlogic allows us to express Hilbert's 10th Problem.

\begin{definition}[Positive Diophantine Equation]
A \emph{positive Diophantine equation} $\mathcal{D}$ is a tuple\\ $(NV,M,(n_w)_{w\in M},(m_w)_{w\in M})$ where 
%\begin{itemize}
%    \item 
$NV$ is a non-empty, ordered set $\{\nv{z_1},\ldots,\nv{z_k}\}$ of number variables;
%    \item 
$M$ (the variable products or \emph{monomials}) is a finite and non-empty, prefix-closed set of sorted variable sequences, i.e.,  $$M \subseteq \{ \underbrace{\nv{z}_1\ldots\nv{z}_1}_{i_1} \ldots  \underbrace{\nv{z}_k\ldots\nv{z}_k}_{i_k} \mid i_1,\ldots,i_k \in \mathbb{N}\};$$ 
%    \item 
and all $n_w$ and $m_w$ are from $\mathbb{N}$ and encode the monomial coefficients on either side of the equation.
%\end{itemize}
A positive Diophantine equation is \emph{solvable} if it admits a solution, where a \emph{solution} for $\mathcal{D}=(NV,M,(n_w)_{w\in M},(m_w)_{w\in M})$ 
is a variable assignment $\nu:NV \to \mathbb{N}$ satisfying
%$$
%\mysum{\smash{\nv{z}_1^{i_1}\ldots\nv{z}_k^{i_k}} \in M}\, lhs(\nv{z}_1^{i_1}\ldots\nv{z}_k^{i_k})\cdot\nu(\nv{z}_1)^{i_1}\cdot\ldots\cdot\nu(\nv{z}_k)^{i_k} =  
%\mysum{\smash{\nv{z}_1^{i_1}\ldots\nv{z}_k^{i_k}} \in M}\, rhs(\nv{z}_1^{i_1}\ldots\nv{z}_k^{i_k})\cdot\nu(\nv{z}_1)^{i_1}\cdot\ldots\cdot\nu(\nv{z}_k)^{i_k}
%$$
%$$\textstyle
%\sum_{\smash{\nv{z}_1^{i_1}\ldots\nv{z}_k^{i_k}} \in M}\, lhs(\nv{z}_1^{i_1}\ldots\nv{z}_k^{i_k})\cdot\nu(\nv{z}_1)^{i_1}\cdot\ldots\cdot\nu(\nv{z}_k)^{i_k} =  
%\sum_{\smash{\nv{z}_1^{i_1}\ldots\nv{z}_k^{i_k}} \in M}\, rhs(\nv{z}_1^{i_1}\ldots\nv{z}_k^{i_k})\cdot\nu(\nv{z}_1)^{i_1}\cdot\ldots\cdot\nu(\nv{z}_k)^{i_k}
%$$
\newcommand{\sldots}{.\hspace{0.1ex}.\hspace{0.1ex}.}
$$
{\sum}_{\smash{w=\nv{z}_1^{i_1}\!\!\sldots\nv{z}_k^{i_k}} {\in} M} n_w\cdot\nu(\nv{z}_1\hspace{-1pt})^{i_1}\cdot\sldots\cdot\nu(\nv{z}_k\hspace{-1pt})^{i_k} {\,=}  
{\sum}_{\smash{w=\nv{z}_1^{i_1}\!\!\sldots\nv{z}_k^{i_k}} {\in} M} m_w\cdot\nu(\nv{z}_1\hspace{-1pt})^{i_1}\cdot\sldots\cdot\nu(\nv{z}_k\hspace{-1pt})^{i_k}\, .
$$

\end{definition}

Solvability of positive Diophantine equations is undecidable, which is a straightforward consequence of the undecidability of arbitrary Diophantine equations over integers \cite{MatiyasevichDio}.

We will show that for any $\mathcal{D}$, we can compute an \thelogic sentence $\varphi_\mathcal{D}$ whose satisfiability over labeled trees coincides with solvability of $\mathcal{D}$, despite $\varphi_\mathcal{D}$ being only
``minimally outside'' \thegoodlogic –- also contrasting the fact that sentences of this shape still warrant decidable satisfiability over finite words \cite[Thm. 8.13]{Kla04}.

\begin{figure}[b]
    \centering
    \include{tikz-tree-test}
    \caption{Illustration of the intended model structure and definition of $\varphi_\mathcal{D}\coloneqq\varphi_{\mathrm{lab}}\land \varphi_{\mathrm{prod}}\land\varphi_{\mathrm{sol}}$}
    \label{fig:tree}
\end{figure}

As detailed in \Cref{fig:tree}, we let $\varphi_\mathcal{D}\coloneqq\varphi_{\mathrm{lab}}\land \varphi_{\mathrm{prod}}\land\varphi_{\mathrm{sol}}$ characterize trees labeled by $w$ and $\hat{w}$, for $w \in M$, such that each model $\xi$ of $\varphi_\mathcal{D}$ corresponds to a solution $\nu$ of $\mathcal{D}$ as follows: 
for each $\nv{z}\in NV$, the number of nodes in $\xi$ labeled with $\nv{z}$ (i.e., $\cnts{\pr{P}_\nv{z}}$) equals the number that $\nu$ assigns to $\nv{z}$. Likewise, for each variable product $w\nv{z}_i\in M$, we ensure that $\cnts{\pr{P}_{w\nv{z}_i}}=\cnts{\pr{P}_{w}}\cdot\cnts{\pr{P}_{\nv{z}_i}}$. To this end, we stipulate via $\varphi_{\mathrm{lab}}$ that for any $w$, all $w$-labeled nodes are pairwise $\succ^{\!*}$-incomparable, and every ${w\nv{z}}$-labeled node has exactly one ${w}$-labeled ancestor (using the label $\hat{w}$ for ``padding'' between $w$ and $w\nv{z}_i$), and we enforce via $\varphi_{\mathrm{prod}}$ that for any $w, w\nv{z}_i\in M$, each subtree rooted in a $w$-labeled node contains precisely as many ${w\nv{z}_i}$-labeled nodes as there are ${\nv{z}}_i$-labeled nodes in the whole tree. Finally, under the conditions enforced by $\varphi_{\mathrm{lab}}$ and $\varphi_{\mathrm{prod}}$, $\varphi_{\mathrm{sol}}$ implements that the model indeed encodes a solution of the given $\mathcal{D}$. 
% We note that the same argument and construction even works for the class of \emph{finite} trees.      

% \[
% \varphi_\mathcal{D}\coloneqq \varphi_{\mathrm{lab}} \land
% \big(\bigwedge_{\mathclap{\ w, w\nv{z}_i\in M\ }} \forall y{\in}\pr{P}_w. \exists Z. 
% \varphi_\downarrow(y,Z) \wedge \cnts{(Z\cap\pr{P}_{w\nv{z}_i})} \tteqfin \cnts{\pr{P}_{\nv{z}_i}}\big) \land \sum_{w \in M} \pr{n}_w\,\cnts{\pr{P}_w} \tteqfin  
% \sum_{w \in M} \pr{m}_w\,\cnts{\pr{P}_w}
% \]

While the first conjunct is pure MSO and the third conjunct is a variable-free Presburger atom, the second conjunct is not in \thegoodlogic: $\exists Z$ occurs inside the scope of $\forall y$, thus $Z$ is not assertive. Yet, as discussed in \Cref{sec:SyntaxSemantics} (Obs. $\dagger$), this is at odds with $Z$ occurring in  $\varphi_\downarrow(y,Z)$.   
%$Z$ co-occurs with $z$ in the set atom $Z(z)$, while $z$ must be tethered (due to its occurrence in $y \succ^+\! z$), therefore $Z$ is tethered as well. Yet, $Z$ occurs in a Presburger atom, 
%which is the (only) reason why the syntactic constraints for \thegoodlogic are violated. 

%If $y$ were quantified existentially, the formula would be in \thegoodlogic. 

%\begin{figure}[t]
%    \centering
%    \include{tikz-tree-test}
%    \caption{Illustration of the intended model structure and definition of $\varphi_\mathcal{D}\coloneqq\varphi_{\mathrm{lab}}\land \varphi_{\mathrm{prod}}\land\varphi_{\mathrm{sol}}$}
%    \label{fig:tree}
%\end{figure}

\begin{restatable}{proposition}{undecidable}
For any positive Diophantine equation $\mathcal{D}$, satisfaction of $\varphi_\mathcal{D}$ over (finite or infinite) labeled trees coincides with solvability of $\mathcal{D}$. Consequently, satisfiability of the class of \thelogic sentences of the shape $\varphi_\mathcal{D}$ is undecidable.
\end{restatable}

%% file: tikz-tree-test.tex
	% \usetikzlibrary{arrows}
	\usetikzlibrary{arrows.meta}
	% \usetikzlibrary{automata}
	% \usetikzlibrary{backgrounds}
	\usetikzlibrary{calc}
	% \usetikzlibrary{chains}
	\usetikzlibrary{decorations}
	\usetikzlibrary{decorations.pathmorphing}
	\usetikzlibrary{decorations.pathreplacing}
	% \usetikzlibrary{eval}
	% \usetikzlibrary{fadings}
	% \usetikzlibrary{fit}
	% \usetikzlibrary{graphs}
	% \usetikzlibrary{intersections}
	\usetikzlibrary{positioning}
	% \usetikzlibrary{quotes}
	% \usetikzlibrary{shapes}
	% \usetikzlibrary{shadows.blur}
	% \usetikzlibrary{shapes.arrows}
	\usetikzlibrary{shapes.geometric}
	% \usetikzlibrary{snakes}
	% \usetikzlibrary{topaths}

\newcommand{\textover}[3][l]{%
 % #1 is the alignment, default l
 % #2 is the text to be printed
 % #3 is the text for setting the width
 \makebox[\widthof{#3}][#1]{#2}%
 }

%\pgfmathsetseed{1} % To have predictable results
% Define a background layer, in which the parchment shape is drawn
\pgfdeclarelayer{background}
\pgfsetlayers{background,main}
% \usepackage{tikzpagenodes}
% \setlength{\headheight}{1cm}

%%%%%%%%%%%%%%%%%%%%%%%%%%%%%%%%%%%%%%%%%%%%%%%%%%%%%%%%%%%%%%%%%%%%%%%%%%%%%%
% TikZ
%%%%%%%%%%%%%%%%%%%%%%%%%%%%%%%%%%%%%%%%%%%%%%%%%%%%%%%%%%%%%%%%%%%%%%%%%%%%%%
\definecolor{myyellow}{RGB}{243,201,71}
\definecolor{myblue}{RGB}{67,139,168}
\definecolor{myred}{RGB}{219,76,55}

\tikzset
	{ > = Stealth
	}

\tikzset{
  inlinesf/.style={
    baseline={([yshift=-0.6ex]current bounding box.center)}
  , level distance=20pt
  , text height=1.5ex
  , text depth=0.75ex
  , sibling distance=20pt
  , inner sep=1.5pt
  }
}
\tikzset{
  inlinetree/.style={
    baseline={(current bounding box.center)}
  }
}
\tikzset{
  tri/.style = {
    inner sep=1pt,
    shape border rotate=90,
    isosceles triangle,
    isosceles triangle apex angle=60,
    % minimum width=1em,
    %draw,
  % top color=blue,
  }
}

\tikzset{
itria/.style={
  draw,dashed,shape border uses incircle,
  isosceles triangle,shape border rotate=90,yshift=-1.45cm},
rtria/.style={
  draw,dashed,shape border uses incircle,
  isosceles triangle,isosceles triangle apex angle=90,
  shape border rotate=-45,yshift=0.2cm,xshift=0.5cm},
ltria/.style={
  draw,dashed,shape border uses incircle,
  isosceles triangle,isosceles triangle apex angle=90,
  shape border rotate=225,yshift=0.2cm},
ritria/.style={
  shape border uses incircle,
  isosceles triangle,isosceles triangle apex angle=100,
  shape border rotate=-50,yshift=0.1cm, top color=myred!20},
letria/.style={
  shape border uses incircle,
  isosceles triangle,isosceles triangle apex angle=100,
  shape border rotate=230,yshift=0.1cm}
}

%\hspace{-6ex}
\begin{tikzpicture}
    [inlinesf,
    level 2/.style={sibling distance=5.5em, level distance=2em},
    pp/.style = {edge from parent path={(\tikzparentnode) -- (\tikzchildnode.north)}},
    empty/.style = {edge from parent path={}},
    ]
    \footnotesize

    \node[tri,minimum width=21em,isosceles triangle apex angle=95, top color=myblue!40] (t) {};
    \node[] (w1) at (-0.7,0.6) {};
    \node (help) at (-0.3,1) {};
    \node[] (w11) at (-0.9,-0.1) {};
    \node[] (w12) at (-1.8,-0.6) {};
    \node[] (w2) at (0.,0.06) {};
    \node[] (h1) at (0.05,-0.15) {};
    \node[] (w21) at (-0.25,-0.75) {};
    \node[] (w22) at (0.08,-0.4) {};
    \node[] (w3) at (0.99,0.33) {};
    \node[] (h2) at (1.2,0.05) {};
    \node[] (h3) at (1.8,-0.4) {};
    \node[] (w31) at (1.2,-0.35) {};
    \node[] (w32) at (1.93,-0.7) {};
	
% 	\node[tri,minimum height=3.5em,isosceles triangle apex angle=50, top color=white, below=-0.3em of t.center] (t2w) {};
	\node[tri,minimum height=4.3em,isosceles triangle apex angle=50, top color=myred!20, below=-0.3em of t.center] (t2) {};
    \node[ritria, left=1.7em of t2, minimum height=4em, yshift=-0.1em](t1) {};
    
    \node[letria, minimum height=3.5em, top color=myred!20, right=2.5em of t2, yshift=-0.4em](t3) {};
	
	\draw[gray!60,decorate, decoration={random steps, aspect=0,amplitude=1pt}] (t.north) -- (help.center) -- (w1.center) -- (w11.center);
    \draw[gray!60,decorate, decoration={random steps, aspect=0,amplitude=2pt}] (w1.center) -- (w12.center);
    \draw[gray!60,decorate, decoration={random steps, aspect=0,amplitude=2pt}] (help.center) -- (w2.center) -- (h1.center) -- (w21.center);
    \draw[gray!60,decorate, decoration={random steps, aspect=0,amplitude=2pt}] (h1.center) -- (w22.center);
    \draw[gray!60,decorate, decoration={random steps, aspect=0,amplitude=2pt}] (t.north) -- (w3.center) -- (h2.center) -- (h3.center) -- (w32.center);
    \draw[gray!60,decorate, decoration={random steps, aspect=0,amplitude=2pt}] (h2.center) -- (w31.center);

    \node[xshift=0.4em, yshift=-0.2em] at (w1.center) {\Large$\cdot_w$};
    \node[xshift=0.4em, yshift=-0.2em] at (w2.center) {\Large$\cdot_w$};
    \node[xshift=0.4em, yshift=-0.2em] at (w3.center) {\Large$\cdot_w$};

    \node[myred,xshift=0.9em, yshift=-0.15em] at (w11.center) {\Large$\cdot_{w\nv{z_i}}$};
    \node[myred,xshift=0.9em, yshift=-0.15em] at (w12.center) {\Large$\cdot_{w\nv{z_i}}$};
    \node[myred,xshift=0.9em, yshift=-0.15em] at (w21.center) {\Large$\cdot_{w\nv{z_i}}$};
    \node[myred,xshift=0.9em, yshift=-0.15em] at (w22.center) {\Large$\cdot_{w\nv{z_i}}$};
    \node[myred,xshift=0.9em, yshift=-0.15em] at (w31.center) {\Large$\cdot_{w\nv{z_i}}$};
    \node[myred,xshift=0.9em, yshift=-0.15em] at (w32.center) {\Large$\cdot_{w\nv{z_i}}$};

    \node[myblue] at (0.25,0.6) {\Large$\cdot_{\nv{z_i}}$};
    \node[myblue] at (-0.3,0.2) {\Large$\cdot_{\nv{z_i}}$};

   	\node (node) at (5.5,1.5) {$\varphi_\mathrm{lab} \coloneqq  \exists x {\in} \pr{P}_\varepsilon.\varphi_\mathrm{root}(x) \wedge \bigwedge\limits_{\mathclap{w\in M}} \big(\forall x {\in} \pr{P}_w\cup\pr{P}_{\hat{w}}.\forall y. x \!\succ\! y \Rightarrow \pr{P}_{\hat{w}}(y) \vee \bigvee\limits_{\mathclap{w\nv{z}_i \in M}}\pr{P}_{w\nv{z}_i}(y)\big)$};
   	%\node (node) at (7.8,0.8) {$\big(\forall x {\in} \pr{P}_{\hat{w}}.\forall y. x \!\succ\! y \Rightarrow \pr{P}_{\hat{w}}(y) \vee \bigvee\limits_{\mathclap{w\nv{z}_i \in M}}\pr{P}_{w\nv{z}_i}(y) \big)$};

   	\node (node) at (6.2,0.3) {$\varphi_{\mathrm{prod}}\coloneqq\bigwedge_{{w, w\nv{z}_i\in M}} \forall y{\in}\pr{P}_w. \exists Z.\varphi_\downarrow(y,Z) \wedge {\color{myred}\cnts{(Z\cap\pr{P}_{w\nv{z}_i})}} \tteqfin{\color{myblue}\cnts{\pr{P}_{\nv{z}_i}}}$};

	\node (node) at (6.8,-0.8)  {$\varphi_{\mathrm{sol}}\coloneqq\sum_{w \in M} \pr{n}_w\,\cnts{\pr{P}_w} \tteqfin \sum_{w \in M} \pr{m}_w\,\cnts{\pr{P}_w}$};

  \end{tikzpicture}  
 \vspace{-2ex}

%% file: normalization.tex
Toward establishing our decidability result, we show that \thegoodlogic formulae can be trans\-formed into a specific, very restricted normal form.
To this end, we use a variety of tech\-niques, mostly known from the literature, but with some adjustments to our setting; thus, due to space, we will restrict ourselves to a high-level description and examples. 
The nor\-mali\-za\-tion procedure is subdivided into two phases: The first phase, establishing the \emph{general normal form} (GNF), is valid independently of the underlying class of structures.
The second phase, yielding the \emph{tree normal form} (TNF), is specific to the class of labeled trees.

%\paragraph*{Generic Equivalent Transformations}

\medskip

Given an \thegoodlogic formula, substitute complex set expressions in modulo and finiteness atoms by new set variables (e.g. $\mathrm{Fin}(\pr{P}{\,\cap\,}  X)$ becomes $\exists Y.(Y\!=\pr{P}{\,\cap\,}  X) {\,\wedge\,} \mathrm{Fin}(Y)$), 
re\-move set operations from set atoms (e.g. turning $(\pr{P}^c {\,\cap\,} X)(y)$ into $\neg\pr{P}(y) {\,\wedge\,} X(y)$), and rewrite all simple Presburger atoms into plain MSO (e.g. $\pr{2}\,\cnts{\pr{P}}\ttleq \pr{3}$ becomes $\forall xy.\pr{P}(x) {\wedge} \pr{P}(y) \Rightarrow x{=}y$). 
Then, \emph{skolemize} all assertive variables (e.g. $\exists x.\exists X.\forall y.\pr{R}(x,y){\Rightarrow} X(y)$ becomes $\forall y.\pr{R}(\pr{c}_x,y){\Rightarrow} \pr{P}_X(y)$). 
Next \emph{``presburgerize''} all non-Presburger atoms containing (only) delicate variables (e.g.~replacing 
$\cnts{X}\equiv_{\pr{3}} \pr{1}$ with $\exists \nv{k}.\cnts{X}\tteqfin\pr{3}\,\nv{k}\plus\pr{1}$), which may require to turn delicate individual into set variables (e.g. $\forall y.\pr{P}(y)\Rightarrow X(y)$ becomes $\forall Y.(\cnts{Y} \tteq \pr{1}) \wedge \pr{1}\ttleq \cnts{(\pr{P}\cap Y)}\Rightarrow \pr{1}\ttleq \cnts{(X\cap Y)}$). 
The resulting formula exhibits a clear separation of variable usage: Presburger atoms use delicate and number variables, all other atoms use non-delicate variables. In a subsequent step, we \emph{``disentangle''} the quantifiers, such that the scopes of quantified number or delicate variables
are strictly separated from those of non-delicate variables.\footnote{While this transformation is not very complicated technically, it may incur nonelementary blowup.}

We next apply \emph{``vennification''}: a technique known from BAPA. In essence, we introduce new number variables to count the number of elements contained in every \emph{Venn region}, that is, every possible combination of set (non-)memberships (with this, $\cnts{(\pr{P} {\,\cup\,} X)} \ttleq \cnts{\pr{P}^c}$ becomes $\nv{k}_{\pr{P} \cap X} {\plus} \nv{k}_{\pr{P}^c \cap X} {\plus} \nv{k}_{\pr{P} \cap X^c}\ttleq \nv{k}_{\pr{P}^c \cap X} {\plus} \nv{k}_{\pr{P}^c \cap X^c}$). This allows us to remove all delicate set variables from our formula. We are now in the setting where we can apply the well-known \emph{quantifier elimination} for Presburger Arithmetic over the ``purely arithmetic'' subformulae (which may produce new modulo atoms) -- since the latter is classically defined for $\mathbb{N}$ instead of $\mathbb{N} \cup \{\infty\}$, we require a pre-processing step implementing a vast case-distinction as to which of the Venn regions are infinite. As a consequence, we obtain a formula free of number variables, with all Presburger atoms being classic and outside any quantifier scope.

Finally, we \emph{``de-skolemize''}: all constants and unary predicates introduced via the initial skolemization, but also by the intermediate transformation steps, are projected away from the signature, re-interpreting them as existentially quantified individual and set variables. We thus recover ``proper'' equivalence with the initial formula. Last, we bring the formula in disjunctive normal form and pull the trailing existential individual quantifiers inside.

\begin{definition}[General Normal Form]\label{def:GNF}
A \emph{Parikh constraint} is a classical Presburger atom without number variables and where all occurring set terms are set variables.
An \thegoodlogic formula is in \emph{general normal form} (GNF), if it is of the shape
$$
\textstyle\exists X_1.\cdots\exists X_n. \bigvee_{i=1}^k \big( \varphi_i \wedge \bigwedge_{j=1}^{l_i} \chi_{i,j} \big), 
$$
where the $\varphi_i$ are CMSO formulae,\footnote{Recall that CMSO is MSO with modulo and finiteness atoms over set variables.} whereas the $\chi_{i,j}$ are (unnegated) Parikh constraints.
\end{definition}

\begin{theorem}
For every \thegoodlogic formula $\varphi$, it is possible to compute an equivalent formula $\varphi'$ in general normal form.
%Any \thegoodlogic formula can be equivalently transformed into GNF.
\end{theorem}

%\clearpage
%\paragraph*{Tree-Specific Equivalent Transformations}

\noindent We now focus on the case of labeled trees.
Very similar to the case of CMSO, under this assumption, we can equivalently transform the GNF formula into one without occurrences of modulo and finiteness atoms. We rewrite $\cnts{X} \equiv_\pr{n} \pr{m}$ into the formula
$$
\begin{array}{r}
\mathrm{Fin}(X) \wedge \exists X_0 ... \exists X_{n-1}.\Big(
\exists x.\big(\varphi_\mathrm{root}(x) \wedge \displaystyle{\bigwedge}_{{{0 \leq i < n}\atop{i \neq m}}}\! \neg X_i(x)\big) \wedge \forall x.\big(
 (\exists y{\in} X. x \succ^* y) \vee X_0(x) \big) \ \wedge~~~\\
\displaystyle{\bigwedge}_{i,j \in \{0,\ldots,n-1\}} \forall z.\big( X_i(z.0) \wedge X_j(z.1) \Rightarrow (\neg X(z) \Rightarrow X_{i\oplus j}(z)) \wedge (X(z) \Rightarrow X_{i\oplus j \oplus 1}(z))\big)\Big),\\
\end{array}
$$
where $\oplus$ denotes addition modulo $n$. Finally, we replace all occurrences of $\mathrm{Fin}(X)$ by $\varphi_\mathrm{fin}(X)$, as defined in \Cref{sec:SyntaxSemantics}. Thus, when employing \thegoodlogic to describe labeled trees, we can confine ourselves to an even more restrictive normal form.

\begin{definition}[Tree Normal Form]\label{def:TNF}
An \thegoodlogic formula is in \emph{tree normal form} (TNF), if it is of the shape
%\vspace{-5ex}
$$
\textstyle \exists X_1.\cdots\exists X_n. \bigvee_{i=1}^k \big( \varphi_i \wedge \bigwedge_{j=1}^{l_i} \chi_{i,j} \big), 
$$
where the $\varphi_i$ are plain MSO formulae and the $\chi_{i,j}$ are (unnegated) Parikh constraints.
\end{definition}

\begin{theorem}\label{thm-tree-normal-form}
For every \thegoodlogic formula $\varphi$, it is possible to compute a formula $\varphi'$ in tree normal form that is equivalent to $\varphi$ over all labeled infinite binary trees.
\end{theorem}

%% file: automata.tex
In this section, we introduce a novel type of automata, 
%called \emph{Parikh-Muller tree automata} (PMTA)
combining and generalizing Parikh tree automata and Muller tree automata. We prove that the tree languages recognized by this automaton type coincide with those definable by TNF formulae. Moreover, we show that the emptiness problem of this automaton model is decidable. In combination,  this yields us decidable satisfiability of \thegoodlogic over labeled infinite binary trees.

\paragraph*{Variable-adorned Trees, Semilinear Sets, and Extended Parikh Maps}

Given a finite set $\Varset\subseteq(\Vind \cup\Vset)$, we denote by $\Phi_\Varset$ the set of all variable assignments of variables from $\Varset$ to elements/subsets of $\{0,1\}^*$. The \emph{set of $\Varset$-models} of a formula $\varphi$ is the set
\(\L_\Varset(\varphi)\coloneqq\{(\xi,\nu)\mid \xi\in T^\omega_\Sigma,\nu\in\Phi_\Varset,\xi,\nu\models\varphi\}\)
and by $\L(\varphi)$ we mean $\L_{\mathrm{free}(\varphi)}(\varphi)$.
%
%When considering automata, 
To represent $\Varset$-models, it is convenient to encode variable assignments $\nu\in\Phi_\Varset$ into the alphabet. For this, we let $\Sigma_\Varset=\Sigma\times 2^\Varset$ be a new alphabet and identify $\Sigma_\emptyset$ with $\Sigma$. We say that a tree $\xi\in T^\omega_{\Sigma_\Varset}$ is \emph{valid} (i.e., it encodes a variable assignment) if for each individual variable $x$ in $\Varset$ there is exactly one position in $\xi$ where $x$ occurs. As there is a bijection between $\To_\Sigma\times\Phi_\Varset$ and the set of all valid trees in $\To_{\Sigma_\Varset}$, we use these two views interchangeably.

\smallskip

A set $C\subseteq\Nat^s$, $s\geq 1$, is \emph{linear} if it is of the form
$C=\{\vec{v}_0 + \textstyle\sum_{i\in[l]}m_i \vec{v}_i \mid m_1,\ldots,m_l\in\Nat\}$
for some $l\in\Nat$ and vectors $\vec{v}_0,\ldots,\vec{v}_l\in\Nat^s$. Any finite union of linear sets is called \emph{semilinear}.
%A subset of $\Nat^s$ is called \emph{semilinear} if it is a finite union of linear sets.

Given two vectors $\vec{v}=(v_1,\ldots,v_s)\in\Nat^s$ and $\vec{v}'=(v_1',\ldots,v'_{s'})\in\Nat^{s'}$, we define their \emph{concatenation} $\vec{v}\cdot\vec{v}'$ as the vector $(v_1,\ldots,v_s,v'_1,\ldots,v'_{s'})\in\Nat^{s+s'}$. This definition is lifted to sets by letting $C\cdot C'=\{\vec{v}\cdot\vec{v}'\mid\vec{v}\in C,\ \vec{v}'\in C'\}\subseteq \Nat^{s+s'}$ for $C\subseteq\Nat^{s},C'\subseteq\Nat^{s'}$\!.

\begin{lemma}[\cite{GinSpa64,GinSpa66}] \label{lemma-semilin-presburger}
The family of semilinear sets of $\Nat^s$ coincides with the family of Pres\-bur\-ger sets of $\Nat^s$ (i.e., sets of the form $\{(x_1,\ldots,x_s)\mid \varphi(x_1,\ldots,x_s)\}$ for a Pres\-bur\-ger formula $\varphi$). Semilinear sets are closed under union, intersection, complement, and concatenation.
\end{lemma}

Given an alphabet $\Sigma$ and some finite $D\subseteq\Nat^s$ for $s\geq 1$, our automaton model works with symbols from $\Sigma\times D$. Thus we use the \emph{projections} 
$\cdot_\Sigma:{\Sigma\times D}\to \Sigma$ with $(a,d)_\Sigma=a$ and
$\cdot_D:{\Sigma\times D}\to D$ with $(a,d)_D=d$, 
which we will also apply to finite and infinite trees, resulting in a pointwise substitution of labels.  
Moreover, the \emph{extended Parikh map} $\Psi\colon T_{\Sigma\times D}\to\Nat^s$ is defined for each finite, non-empty tree $\xi\in T_{\Sigma\times D}$ by
$\Psi(\xi)=\sum_{i\in\pos(\xi)}(\xi(i))_D\,.$ 

\paragraph*{Automaton Model}

We now formally introduce our notion of a \emph{Parikh-Muller Tree Automaton (PMTA)}, which recognizes infinite trees employing a Muller acceptance condition while also testing some finite initial tree part for an arithmetic property related to Parikh's commutative image~\cite{Par66}. This is implemented by utilizing a finite number of global counters, which are ``blindly'' increased throughout the run, but are read off only once a posteriori -- when it is verified whether the tuple of the final counter values belongs to a given semilinear set.

\begin{definition}[Parikh-Muller Tree Automaton] 
Let $\Sigma$ be an alphabet, let $s \in \mathbb{N}\setminus\{0\}$, let $D\subseteq\Nat^s$ be finite, and denote $(\Sigma\times D)\cup\Sigma$ by $\Xi$. A PMTA (of \emph{dimension} $s$) is a tuple $\A=(\Q,\Xi,q_I,\Delta,\F,C)$ where $\Q=Q_P\cup Q_\omega\cup\{q_I\}$ is a finite set of \emph{states} with $Q_P$,$Q_\omega$ disjoint and $q_I$ being the \emph{initial state}, $\Delta=\Delta_P\cup\Delta_\omega$ is the \emph{transition relation} with
 \[\Delta_P\subseteq (Q_P\cup\{q_I\})\times(\Sigma\times D)\times \Q\times\Q \quad\text{ and }\quad \Delta_\omega\subseteq (Q_\omega\cup\{q_I\})\times\Sigma\times Q_\omega\times Q_\omega,\]
 $\F\subseteq 2^{Q_\omega}$ is a set of \emph{final state sets}, and $C\subseteq\Nat^s$ is  a semilinear set named \emph{final constraint}.
\end{definition}

\begin{definition}[Semantics of PMTA]
A \emph{run} of $\A$ on a tree $\zeta\in T_\Xi^\omega$ is a tree $\kappa_\zeta\in T_Q^\omega$ whose root is labeled with $q_I$ 
and which respects $\Delta$ jointly with $\zeta$. By definition of $\Delta$, if a run exists, then 
%there exists a prefix $X\subseteq\{0,1\}^*$ s.t. $\zeta(x)\in(\Sigma\times D)$ if and only if $x\in X$; in the following we denote $\zeta_{|X}$ by $\zeta_\cnt$. 
%% REWRITTEN INTO THE FOLLOWING:
$\zeta^{-1}(\Sigma\times D)$ is prefix-closed; we denote $\zeta_{|\zeta^{-1}(\Sigma\times D)}$ by $\zeta_\cnt$. 
%%
%We say that the run 
A run $\kappa_\zeta$ is \emph{accepting} if
 \begin{enumerate}
     \item for each path $\pi$, we have $\inf(\kappa_\zeta(\pi))\in \F$, and
     \item if $\pos(\zeta_\cnt)\neq\emptyset$, then $\Psi(\zeta_\cnt)\in C$.
 \end{enumerate}
Note that, by the first condition, $\kappa_\zeta$ being accepting implies finiteness of $\zeta_\cnt$ and, thus, well-definedness of the sum in $\Psi(\zeta_\cnt)$.
 The set of all accepting runs of $\A$ on $\zeta$ will be denoted by $\Run_\A(\zeta)$. Then, the \emph{tree language of} $\A$, denoted by $\L(\A)$, is the set \[\L(\A)\coloneqq\{\xi\in T_{\Sigma}^\omega\mid \exists \zeta\in T_\Xi^\omega \text{ with } \Run_\A(\zeta)\neq\emptyset \text{ and } (\zeta)_\Sigma=\xi\}\,.\]
\end{definition}

We highlight that, by choosing $\Delta_P=\emptyset$, we reobtain the well-known concept of a \emph{Muller tree automaton (MTA)}. In this case, we can drop $\Q_P$, $D$, $\Delta_P$, and $C$ from $\A$’s specification without affecting its semantics. Thus,  we define an MTA $\A$ by the tuple $(\Q_\omega,\Sigma,q_I,\Delta_\omega,\F)$.
 
For alphabets $\Sigma, \Gamma$, a \emph{relabeling} (from $\Sigma$ to $\Gamma$) is a mapping $\tau\colon\Sigma\to\Pow(\Gamma)$. We extend it to a mapping $\tau\colon\To_\Sigma\to\Pow(\To_\Gamma)$ by letting $\xi'\in\tau(\xi)$ if and only if for each position $\rho\in\{0,1\}^*$, we have $\xi'(\rho)\in\tau(\xi(\rho))$. Note that the reverse $\tau^{-1}$ of a relabeling $\tau$ is again a relabeling. %\todo{Der Absatz könnte in den Appendix.}

\begin{restatable}{proposition}{closurepmta}\label{closure-pmta}
The set of tree languages recognized by Parikh-Muller tree automata is closed under union, intersection, and relabeling.
\end{restatable}

\begin{proof}[Proof (sketch)]
As the proof techniques are rather standard and some of them were already presented in earlier work \cite{GuhJecLeh22}, we only sketch the main ideas here. Let $\A_1$ and $\A_2$ be PMTA. 

For the \emph{union}, we construct a PMTA %$\A$ 
that starts in a fresh initial state. From there, it can either enter the transitions of $\A_1$ or of $\A_2$; we keep apart the final constraints of $\A_1$ and $\A_2$ by using one additional dimension. The \emph{intersection} PMTA %$\A'$ 
is constructed as the Cartesian product of $\A_1$ and $\A_2$; it uses the concatenation of final constraints of both given PMTA and, as $\A_1$ and $\A_2$ might not ``arithmetically test'' the same initial tree part, it can nondeterministically freeze parts of its counters on different paths. \emph{Relabeling} is trivial.
\end{proof}

\paragraph*{Correspondence of PMTA and \thegoodlogicheading}

We now provide a logical characterization of  PMTAs, by showing that a tree language is recognized by a PMTA precisely if it is the set of tree models of some \thegoodlogic sentence.
The ``only if'' part is established by \Cref{prop:rectodef} and  the ``if'' part by \Cref{prop:deftorec}.
%This result also provides an answer to the open problem posed by the authors in \cite{GroSie23,GroSabSie23a} to find a logical characterization for their \emph{reachability-regular Parikh automata} (RRPA) on strings: in the usual way, our tree automata can simulate string automata (by embedding strings in particular trees) and it is not hard to see that the string version of PMTA is expressively equivalent to RRPA (details can be found in the appendix). Finally, by inspecting the following proofs we easily observe that our logical characterization can be adapted to the string setting.

%\clearpage
\begin{restatable}{proposition}{rectodef}\label{prop:rectodef}
For any PMTA $\A$, there is an \thegoodlogic sentence $\varphi$ with $\L(\A){\,=\,}\L(\varphi)$.
\end{restatable}

\begin{proof}
Given a PMTA $\A=(\Q,\Xi,q_I,\Delta,\F,C)$, we adopt (and slightly simplify) the idea from \cite[Thm. 10]{KlaRue03} of how to encode counter values and the semilinear set $C$, and combine it with the usual construction to define the behavior of an MTA by means of an MSO formula: The existence of a run is defined by a sequence of existential set quantifiers representing the states of $\A$; one additional universal set quantifier ranging over paths is used to encode the Muller acceptance condition. Furthermore, we (outermost) existentially quantify over ``counter contributions'' using set quantifiers $Z_1^0,...\,,\! Z_1^K\!,\ldots,Z_s^0,...\,,\! Z_s^K$ (with $s$ being the number of counters and $K$ the greatest counter increment occurring in $\A$'s transitions) -- the presence of a variable $Z^{d_i}_i$ at a position indicates that $d_i$ has to be added to the $i$th counter to simulate the extended Parikh map. 
Then we enforce satisfaction of the final constraint $C$ by adding the conjunct $\varphi_C$ defined as follows:
By definition of $C$, there are $k,l\in \mathbb{N}\setminus \{0\}$ and linear polynomials $p_1,\ldots,p_k\colon\Nat^l\to\Nat^s$ such that $C$ is the union of the images of $p_1,\ldots,p_k$. Assume $p_g(m_1,\ldots,m_l)=\vec{v}_0+m_1\vec{v}_1+\ldots+m_l\vec{v}_l$
with $\vec{v}_j=(v_{j,1},\ldots,v_{j,s})$.
Then, using number variables $\nv{m}_1,\ldots,\nv{m}_l$, we encode $p_g$ by
$$\textstyle\varphi_{p_g}\coloneqq\exists \nv{m}_1\ldots \exists\nv{m}_l.\bigwedge_{i=1}^s \Bigl(\sum_{d=0}^K\pr{d}\cnts{Z_i^d}\tteqfin\pr{v}_{0,i}+\pr{v}_{1,i}\nv{m}_1+\ldots+\pr{v}_{l,i}\nv{m}_l\Bigr),$$
and let
$\textstyle\varphi_C\coloneqq\big(
%\bigwedge_{\tiny\raisebox{1ex}{$\substack{1\leq i\leq s\\ 0\leq d\leq K}$}}
\bigwedge_{i=1}^s
\bigwedge_{d=0}^K
\forall x. \neg Z_i^d(x)\big)\lor\varphi_{p_1}\lor\ldots\lor\varphi_{p_k}$. This finishes the construction of the overall sentence specifying $\L(\A)$, which can be easily shown to be in \thegoodlogic. 
\end{proof}

The other direction is proved by an induction on the structure of TNF formulae involving the closure properties of PMTA. The last piece that needs to be shown for this is the recognizability of the models of a Parikh constraint.

\begin{restatable}{lemma}{parikhtorec}\label{parikh-to-rec}
For each Parikh constraint $\chi$ there is a PMTA $\A$ with $\L(\A)=\L(\chi)$.
\end{restatable}

\begin{proof} 
We assume w.l.o.g.~that $\chi$ is of the form $\pr{c} + \sum_{i\in[r]} \pr{c_i}\,\cnts{X_i} \ttleqfin \pr{d} + \sum_{j\in[k]} \pr{d_j}\,\cnts{Y_j}$ 
where all $X_i$ are pairwise distinct, and all $Y_j$ likewise.
%where on both sides of $\ttleqfin$ duplicates of set variables are resolved by multiplication, respectively. 
Given a subset $\theta\subseteq\mathrm{free}(\chi)$, we denote by $|\theta|_X$ the number $\sum_{X_i\in\theta}c_i$ (and similar for $|\theta|_Y$). Then, assuming $\xi(\rho)=(\sigma^\xi_\rho,\theta^\xi_\rho)$, we get
\[\L(\chi)=\{\xi\in\To_{\Sigma_{\mathrm{free}(\chi)}}\mid  c+ \textstyle\sum_{{\rho\in\pos(\xi)}}|\theta^\xi_{\rho}|_X\leq d + \sum_{{\rho\in\pos(\xi)}} |\theta^\xi_{\rho}|_Y<\infty\}\]
and, by the condition ${<\,}\infty$, both sums can add up only finitely many non-zero elements. Therefore, $\xi\in\L(\chi)$ holds exactly if there is a non-empty, finite, prefix-closed $Z\subset\{0,1\}^*$ that comprises all positions holding variable assignments and for which $\xi|_Z$ satisfies $\chi$. This condition can be verified by a PMTA defined in the following.

Let $D=\{(i,j)\mid 0\leq i\leq \sum_{l\in[r]} c_l, 0\leq j\leq \sum_{l\in[k]} d_l\}$. We construct the PMTA $\A=(\{q_I,q_f\},\Xi,q_I,\Delta,\{\{q_f\}\},C)$ with $\Xi=(\Sigma_{\mathrm{free}(\chi)}\times D)\cup\Sigma_{\mathrm{free}(\chi)}$, $\Delta=\Delta_P\cup\Delta_\omega$ where
\begin{itemize}
    \item $\Delta_P=\{(q_I,\bigl((\sigma,\theta),(|\theta|_X,|\theta|_Y)\bigr),q',q')\mid (\sigma,\theta)\in\Sigma_{\mathrm{free}(\chi)}, q'\in\{q_I,q_f\}\}$ and
    \item $\Delta_\omega=\{(q_f,(\sigma,\emptyset),q_f,q_f)\mid\sigma\in\Sigma\}$
\end{itemize}
and $C=\{(z_1,z_2)\mid \pr{c}+z_1\ttleqfin \pr{d}+z_2\}$.\footnote{Note that by \cref{lemma-semilin-presburger} we can use this description for a semilinear set.}
Then, one can easily show that $\L(\chi)=\L(\A)$.
\end{proof}

\begin{restatable}{proposition}{deftorec}\label{prop:deftorec}
For every \thegoodlogic formula $\varphi$ there is a PMTA $\A$ with $\L(\A)=\L(\varphi)$.
\end{restatable}

\begin{proof}
Let $\varphi$ be an \thegoodlogic formula. By \cref{thm-tree-normal-form}, we can assume that $\varphi$ is in tree normal form, i.e., of the form $\exists X_1.\cdots\exists X_n. \bigvee_{i=1}^k \big( \varphi_i \wedge \bigwedge_{j=1}^{l_i} \chi_{i,j} \big)$,
where $\varphi_i$ are plain MSO sentences and the $\chi_{i,j}$ are (unnegated) Parikh constraints. The proof of the statement is an induction on the (now restricted) structure of $\varphi$ using the well-known recognizability of MSO sentences \cite{Rab69}, \cref{parikh-to-rec}, and \cref{closure-pmta}. 
\end{proof}

The characterization obtained through \Cref{prop:rectodef} and \Cref{prop:deftorec} also provides an answer to the open problem posed by the authors in \cite{GroSie23,GroSabSie23a} to find a logical characterization for their \emph{reachability-regular Parikh automata} (RRPA) on words: in the usual way, our tree automata can simulate word automata (by embedding words in particular trees) and it is not too hard to see that the word version of PMTA is expressively equivalent to RRPA (details can be found in the appendix). Finally, by a routine inspection of the corresponding proofs we easily observe that our logical characterization also applies to the word setting.

\paragraph*{Deciding Emptiness of Parikh-Muller Tree Automata}

% \begin{theorem}[\cite{Rab69,HunDaw05}]\label{emptiness-mta}
% Given an MTA $\A$, deciding $\L(\A)\neq\emptyset$ is \textsc{PSpace}-complete.
% \end{theorem}

Our proof of decidability (and complexity) of the emptiness problem of PMTA rests on the respective results for the two components it combines, MTA and PTA. Thus, let us first recall the definition of Parikh tree automata \cite{KlaRue02,Kla04}, slightly adjusted to our setting.

\begin{definition}[Parikh tree automaton \cite{KlaRue03}]
Let $\Sigma$ be an alphabet, let $s\geq 1$, and let $D\subseteq\Nat^s$ be finite. A \emph{Parikh tree automaton (PTA)} is a tuple $\A=(Q,\Sigma\times D,\delta,q_I,F,C)$ where $Q$ is a finite set of \emph{states}, $\delta\subseteq Q\times(\Sigma\times D)\times Q\times Q$ is the \emph{transition relation}, $q_I$ is the \emph{initial state}, $F\subseteq Q$ is a set of \emph{final states}, and $C\subseteq\Nat^s$ is a semilinear set.\footnote{We note that the PTAs defined in \cite{KlaRue03} were total, i.e., $\delta$ is a function of type $Q\times(\Sigma\times D)\to\Pow(Q\times Q)$. Each PTA as defined here can be made total by using an additional sink state.}
Given a finite tree $\xi\in T_{\Sigma\times D}$, a \emph{run} of $\A$ on $\xi$ is a tree $\kappa_\xi\in T_Q$ with $\pos(\kappa_\xi)=\{\varepsilon\}\cup\{ui\mid u\in\pos(\xi),i\in\{0,1\}\}$ and $\kappa(\varepsilon)=q_I$ that respects the transition relation of $\A$. The run $\kappa_\xi$ is said to be \emph{accepting} if $\Psi(\xi)\in C$ and $\kappa_\xi(u)\in F$ for each leaf $u\in\pos(\kappa_\xi)\setminus\pos(\xi)$; we denote the set of all accepting runs of $\A$ on $\xi$ by $\Run_\A(\xi)$. Finally, the \emph{tree language of $\A$}, denoted $\L(\A)$, is the set
\[\L(\A)\coloneqq\{\xi\in T_\Sigma\mid \exists \xi'\in T_{\Sigma\times D}\text{ with } \Run_\A(\xi')\neq\emptyset\text{ and } (\xi')_\Sigma=\xi\}\,.\]
\end{definition}

It was shown in \cite{Kla04} that non-emptiness is decidable for PTA. The exact complexity can be obtained by adopting \cite[Proposition III.2.]{FigLib15a} to the tree setting. This ultimately enables us to establish the desired result for our automaton model.

\begin{restatable}[based on \cite{Kla04,FigLib15a}]{proposition}{ptaempt}\label{emptiness-pta}
Given a PTA $\A$, deciding $\L(\A)\neq\emptyset$ is \textsc{NP}-complete.
\end{restatable}

\vspace{-0.8em}

\begin{restatable}{theorem}{pmtaempt}\label{pmtaempt}
Given a PMTA $\A$, deciding $\L(\A)\neq\emptyset$ is \textsc{PSpace}-complete.
\end{restatable}

\begin{proof}[Proof (sketch)]
Let $\A=(\Q,\Xi,q_{I},\Delta,\F,C)$ be a PMTA with $\Q= Q_{P}\cup Q_{\omega}\cup\{q_I\}$, $\Xi=(\Sigma\times D)\cup\Sigma$, and $\Delta=\Delta_{P}\cup \Delta_{\omega}$. We observe that each tree in the language of $\A$ can be seen as some finite tree over $\Sigma\times D$ (on which the Parikh constraint is tested), 
having infinite trees from $T_\Sigma$ attached to all its leafs. This allows us to reduce PMTA non-emptiness testing to deciding non-emptiness of Muller tree automata and Parikh tree automata.
To this end, consider 
\begin{itemize}
    \item the Muller tree automaton $\A_{q_I}=(Q_\omega\cup\{q_I\},\Sigma,q_I,\Delta_\omega,\F)$, 
    \item the Muller tree automata $\A_q=(Q_\omega,\Sigma,q,\Delta_\omega,\F)$ for all $q\in Q_\omega$, and
    \item the Parikh tree automaton $\A_P=(\Q,\Sigma\times D\!,q_I,\Delta_P,F_P,C)$ with $F_P=\{q{\,\in\,} Q_\omega \,|\, \L(\A_q){\,\neq\,}\emptyset\}$.
%    \item let $F_P=\{q\in Q_\omega\mid \L(\A_q)\neq\emptyset\}$, and
%    \item $\A_P=(\Q,\Sigma\times D,q_I,\Delta_P,F_P,C)$ be a Parikh tree automaton.
\end{itemize}
As deciding $\L(\A_q)\neq\emptyset$ is \textsc{PSpace}-complete \cite{Rab69,HunDaw05}, $\A_P$ can be constructed in \textsc{PSpace} and, by \cref{emptiness-pta}, its  non-emptiness can be decided in \textsc{NP}. Thus, the overall \textsc{PSpace} complexity follows from the observation that
$\L(\A)\neq\emptyset \quad \text{iff} \quad \L(\A_{q_I})\neq\emptyset\ \text{or}\ \L(\A_P)\neq\emptyset.$
\end{proof}

%% file: treeinterpretable.tex
Finally, we lift the obtained decidability result for labeled trees to much more general classes of structures, leveraging the well-known technique of \emph{MSO-interpretations} 
(also referred to as MSO-transductions or MSO-definable functions in the literature \cite{ArnborgLS88,Courcelle91a,Engelfriet90,Courcelle94,CourEngBook}). 

\begin{definition}[MSO-Interpretation]
Given two signatures \,$\Sig$ and \,$\Sig'$\!, an \emph{MSO-interpretation} is a sequence
$
\mathcal{I} = ( 
\varphi_\mathrm{Dom}(x), 
(\varphi_\pr{c}(x))_{\pr{c}\in \SigC}, 
(\varphi_\pr{Q}(x_1,\ldots,x_{ar(\pr{Q})}))_{\pr{Q}\in \SigP}
)
$
of MSO-formulae over %the signature 
$\Sig'$ (with free variables as indicated). We identify $\mathcal{I}$ with the partial function satisfying $\mathcal{I}(\mathfrak{A}) = \mathfrak{B}$ for an $\Sig'$-structure $\mathfrak{A}$ and an $\Sig$-structure $\mathfrak{B}$
%, we write $\mathcal{I}(\mathfrak{A}) = \mathfrak{B}$ 
if
%\begin{itemize}
%\item 
$\{ a \in A \mid \mathfrak{A},\{x \mapsto a\} \models \varphi_\mathrm{Dom}(x) \} = B$ as well as 
%\item 
$\{ a \in B \mid \mathfrak{A},\{x \mapsto a\} \models \varphi_{\pr{c}}(x) \} = \{\pr{c}^\mathfrak{B}\}$ for every $\pr{c} \in \SigC$, and, for every $\pr{Q} \in \SigP$, we have
%\item 
$\pr{Q}^\mathfrak{B} = \{ (a_1,\ldots,a_{ar(\pr{Q})}) \in B^{ar(\pr{Q})} \mid \mathfrak{A},\{x_i \mapsto a_i\}_{1\leq i \leq ar(\pr{Q})} 
\models \varphi_{\pr{Q}}(x_1,\ldots,x_{ar(\pr{Q})}) \}$.
For a class $\mathscr{S}$ of $\Sig'$-structures, let $\mathcal{I}(\mathscr{S}) \coloneqq \{\mathfrak{B} \mid \mathcal{I}(\mathfrak{A})=\mathfrak{B}, \mathfrak{A}\in \mathscr{S}\}$.
A class $\mathscr{T}$\! of $\Sig$-structures is \emph{tree-inter-\\pretable}, if it coincides with $\mathcal{I}(T^\omega_\Sigma)$ for some $\Sigma$ and corresponding MSO-interpretation $\mathcal{I}$. 
\end{definition}

The key insight for our result is that the well-known rewritability of MSO formulae under MSO-interpretations can be lifted to \thegoodlogic without much effort. 

%\begin{lemma}
%Let $\mathcal{I}$ be an MSO-interpretation. Then, for every \thegoodlogic sentence $\varphi$ over $\Sig'$ one can compute an
%\thegoodlogic sentence $\varphi^\mathcal{I}$ over $\Sig$ satisfying 
%$ \mathfrak{A} \models \varphi^\mathcal{I} \Longleftrightarrow \mathfrak{B} \models \varphi$
%for every $\Sig'$-structure $\mathfrak{A}$ and $\Sig$-structure $\mathfrak{B}$ with $\mathcal{I}(\mathfrak{A}) \cong \mathfrak{B}$.
%\end{lemma}

\begin{restatable}{lemma}{msointerlemma}
\label{lem:msointer}%
Let $\mathcal{I}$ be an MSO-interpretation. Then, for every \thegoodlogic sentence $\varphi$ over $\Sig$ one can compute an
\thegoodlogic sentence $\varphi^\mathcal{I}$ over $\Sig'$ satisfying 
$ \mathfrak{A} \models \varphi^\mathcal{I} \Longleftrightarrow \mathfrak{B} \models \varphi$
for every $\Sig'$-structure $\mathfrak{A}$ and $\Sig$-structure $\mathfrak{B}$ with $\mathcal{I}(\mathfrak{A}) \cong \mathfrak{B}$.
\end{restatable}

This insight can be used to show that decidability is propagated through MSO-interpreta\-tions, and thus can be guaranteed for all tree-interpretable classes, thanks to \Cref{cor:treedecidable}.

\begin{restatable}{theorem}{msointerthm}
\label{thm:msointer}%
Let $\mathscr{S}$ be a class of structures over which satisfiability of \thegoodlogic is decidable, let $\mathcal{I}$ be an MSO-interpretation. Then satisfiability of \thegoodlogic over $\mathcal{I}(\mathscr{S})$ is decidable as well. 
In particular, \thegoodlogic is decidable over any tree-interpretable class. 
\end{restatable}

%\begin{theorem}
%Let $\Sigma$ be a finite alphabet. Then, for every MSO-interpretation $\mathcal{I}$, satisfiability
%of \thegoodlogic sentences over $\mathcal{I}(T^\omega_\Sigma)$ is decidable.
%\end{theorem}

This result allows us, in one go, to harvest several decidability results, as tree-interpreta\-bility is able to capture classes of (finite or countable) structures whose treewidth \cite{ROBERTSON198449}, cliquewidth \cite{ICDT2023,CourEngBook,Cou04,GroheT04}, or partitionwidth \cite{Blumensath03,Blumensath06,feller2023decidability} is bounded by some value $k \in \mathbb{N}$.  

\begin{corollary}\label{cor:whateverwidth}
Given a signature $\Sig$, satisfiability of \thegoodlogic is decidable over the classes of finite or countable $\Sig$-structures of bounded
%\begin{itemize}
%    \item 
treewidth,
%    \item 
cliquewidth, and
%    \item 
partitionwidth.
%\end{itemize}
\end{corollary}

%% file: twovarlogics.tex
\Cref{cor:whateverwidth} constitutes a strong decidability result, also in view of the fact that lifting the width restriction immediately leads to undecidability even for much weaker logics like FO.
A feasible way to nevertheless relax this restriction without putting decidability at risk and yet maintaining all the expressive power of \thegoodlogic is to ``couple'' it with another logic~$\mathbb{L}$ whose satisfiability problem is decidable over arbitrary structures.
Then, one considers sentences $\varphi \wedge \psi$, where $\varphi$ is an \thegoodlogic sentence while $\psi$ is an $\mathbb{L}$-sentence, and asks for models whose reduct to the signature of $\varphi$ adheres to the width restriction.
That way, signature elements of $\psi$ not occurring in $\varphi$ can ``behave freely'' and are not subject to the imposed width constraint.\footnote{We refer to Kotek et al. \cite{KotVeiZul16} for a result that is similar in spirit, establishing decidability of finite satisfiability of treewidth-bounded MSO$_2$ coupled with $\mathrm{C}^2$.} 
Such a ``coupling'' of  \thegoodlogic and $\mathbb{L}$ can be made more or less ``tight'' depending on the arity of the predicates allowed to be shared between $\varphi$ and $\psi$.

We can show that a decidable coupling with shared unary predicates can be done for $\mathbb{L} = \FOpres$ \cite{Ben20}, an expressive extension of 2-variable first-order logic by Presburger-like counting quantifiers of the form $\exists^S$, where $S \subseteq \mathbb{N} \cup \{\infty\}$ is an ultimately periodic set from $\mathbb{N} \cup \{\infty\}$ with the semantics defined by $\mathfrak{A},\nu \models \exists^S x.\varphi$ iff $  |\{ a \in A \mid \mathfrak{A},\nu_{x \mapsto a} \models \varphi\}| \in S$. $\FOpres$ subsumes the prominent counting 2-variable first-order fragment $\mathrm{C}^2$ \cite{GradelOR97}, but goes beyond first-order logic. Its satisfiability problem was shown to be decidable only recently \cite{Ben20}.

\begin{restatable}{theorem}{mixed}\label{thm:mixed}
%\begin{theorem}\label{thm:mixed}
Let $w$ be any of treewidth, cliquewidth, or partitionwidth, and let $n \in \mathbb{N}$.
Let $\Sig_a$ and $\Sig_b$ be signatures whose only joint elements are unary predicates.
Then the following problem is decidable: 
Given a \thegoodlogic sentence $\varphi$ over $\Sig_a$ and a $\FOpres$ sentence $\psi$ over $\Sig_b$, does there exist a countable $\Sig_a {\cup}\, \Sig_b$-structure
$\mathfrak{C}$ satisfying $w(\mathfrak{C}|_{\Sig_a} ) \leq n$ and $\mathfrak{C} \models \varphi \wedge \psi$.  
%\end{theorem}
\end{restatable}

In a nutshell, this result is obtained by exploiting the fact that, for every $\FOpres$ formula $\psi$ over $\Sig_b$, one can construct a \thegoodlogic formula $\psi'$ over the purely unary signature $\Sig_a \cap \Sig_b$ that is satisfied by precisely those $\Sig_a$-structures that are ``$\Sig_a \!\cap \Sig_b$-compatible'' with some model of $\psi$. Consequently, the \thegoodlogic formula $\varphi \wedge \psi'$ over $\Sig_a$ is such that for any of its models $\mathfrak{A}$ one finds a ``$\Sig_a \!\cap \Sig_b$-compatible'' model $\mathfrak{B}$ of $\psi$. Then, superimposing $\mathfrak{A}$ and $\mathfrak{B}$ would yield a model $\mathfrak{C}$ of $\varphi \wedge \psi$, which by construction satisfies $w(\mathfrak{C}|_{\Sig_a}) = w(\mathfrak{A})$. Consequently, to solve the decision problem of \Cref{thm:mixed}, it suffices to check if the \thegoodlogic formula $\varphi \wedge \psi'$ has a model $\mathfrak{A}$ satisfying $w(\mathfrak{A}|_{\Sig_a} ) \leq n$ which is decidable by \Cref{cor:whateverwidth}. We note that the extended arithmetic capabilities of \thegoodlogic are essential for this result, as $\psi'$ needs to encode linear inequalities over counts of realized atomic 1-types. % over $\Sig_a \!\cap \Sig_b$.

%% file: injectivePDL.tex
%An important and practically relevant class of expressive logical formalisms, which play a pivotal role in logic-based knowledge representation and verification, is obtained from variations and extensions of modal logics and description logics. To name just a few, this class contains most ontology languages as well as PDL, CTL$^*$, modal $\mu$-calculus, as well as their graded and/or hybrid variants. Modulo some representational variations, all these logics' model-theoretic semantics rest on structures over unary and binary predicates (often interpreted as a transition system's state space). While the simpler variants of this family can be seen as fragments of first-order logic, the more expressive ones cannot, as they feature fixed-point capabilities (through regular path expressions or explicit fixed-point operators). In most cases, decidability of the satisfiability problem in these logics can be obtained from the fact that they exhibit some sort of the tree-model property. Many of these logics exhibit some limited counting capabilities, but recently, there has been an increased interest in accommodating more advanced arithmetic constraints, including \emph{global constraints} expressing statistical information such as ``more than 50\% of the state space's final states are successful''.

An important and practically relevant class of expressive logical formalisms, which play a pivotal role in logic-based knowledge representation and verification, is obtained from variations and extensions of propositional modal logics \cite{ModalLogicBook,ModalLogicHandbook} and description logics \cite{baader_horrocks_lutz_sattler_2017,Rudolph11}. %To name just a few, 
This class contains most ontology languages as well as PDL \cite{PDL}, CTL$^*$ \cite{CTLstar}, the propositional modal $\mu$-calculus \cite{mucalculus} and their extensions. 
%Modulo some representational variations, all these logics' model-theoretic semantics rest on structures over unary and binary predicates (often interpreted as a transition system's state space). %While the simpler variants of this family can be seen as fragments of first-order logic, the more expressive ones cannot, as they feature fixed-point capabilities (through regular path expressions or explicit fixed-point operators). 
%Typically, decidability of the satisfiability problem in these logics follows from some sort of the tree-model property. 
Many of these logics exhibit some limited local counting capabilities \cite{Tobies00}, but recently, there has been an increased interest in accommodating more advanced arithmetic constraints \cite{DemriL10,mucounting,BaaderBR20,BednarczykOPT21}, including \emph{global constraints} \cite{BaaderBH96,Rudolph19} expressing statistical information such as ``more than 50\% of the state space's final states are successful''.

We will demonstrate the usefulness of \thegoodlogic for establishing decidability results at the example of adding global Presburger constraints to the \emph{fully enriched $\mu$-calculus}, a very powerful formalism used in verification. 
We first introduce syntax and semantics.\footnote{For brevity and coherence, we slightly adjust the syntax and use classical model-theoretic semantics (structures with unary and binary predicates) instead of the original one of modal logic (Kripke structures with propositional variables and programs), as the two are well known to be equivalent.}  

\begin{definition}
Given a signature 
$\Sig = \Sig_\Consts \cup \Sig_{\Preds,1} \cup \Sig_{\Preds,2}$ of constants, unary predicates and binary predicates,
the formulas of the fully enriched $\mu$-calculus (FEµ) are defined by

\medskip
%    \begin{align*}
$  \varphi \Coloneqq{}\ \mathbf{true} \mid \mathbf{false} \mid X \mid \pr{c} \mid \neg \pr{c} 
\mid \pr{P} \mid \neg\pr{P} \mid \varphi \wedge \varphi' \mid  \varphi \vee \varphi' \mid 
\langle n,\alpha \rangle\varphi \mid [ n,\alpha ]\varphi \mid {\rm\mathtt\rm \upmu} X.\varphi \mid {\mathrm \upnu} X.\varphi$ 
%    \end{align*}

\medskip

\noindent where $X$ is a set variable from some countable set $\Vset$, $\pr{P} \in \Sig_{\Preds,1}$, $n\in \mathbb{N}$ and $\alpha$ has the form $\pr{R}$ or $\pr{R}^-$ for some $\pr{R} \in \Sig_{\Preds,2}$. 
For ease of presentation, we assume positive normal form.
Given a structure $\mathfrak{A}$ and a set variable assignment $\nu: \Vset \to 2^A$, the semantics $\semof{\varphi}^{\mathfrak{A}}_\nu \subseteq A$ of formulae $\varphi$ is defined by the following function (stipulating $(\pr{R}^-)^\mathfrak{A} = \{(a,a') \mid (a',a) \in \pr{R}^\mathfrak{A}\}$):

\medskip
\noindent$
\begin{array}{r@{\ \mapsto \ }l}
~\mathbf{true} & A \\ 
~\mathbf{false} & \emptyset \\ 
\end{array}\ 
\begin{array}[h]{r@{\ \mapsto \ }l}
X & \nu(X)\\[3ex] 
\end{array}\ 
\begin{array}{r@{\ \mapsto \ }l}
\pr{c} & \{\pr{c}^\mathfrak{A}\}\\ 
\neg \pr{c} & A \setminus \{\pr{c}^\mathfrak{A}\} \\ 
\end{array}\ 
\begin{array}{r@{\ \mapsto \ }l}
\pr{P} & \pr{P}^\mathfrak{A}\\ 
\neg\pr{P} & A \setminus \pr{P}^\mathfrak{A} \\ 
\end{array}\ 
\begin{array}{r@{\ \mapsto \ }l}
\varphi \wedge \varphi' & \semof{\varphi}^{\mathfrak{A}}_\nu \cap \semof{\varphi'}^{\mathfrak{A}}_\nu \\  
\varphi \vee \varphi' & \semof{\varphi}^{\mathfrak{A}}_\nu \cup \semof{\varphi'}^{\mathfrak{A}}_\nu \\ 
\end{array}$

\smallskip

\noindent$
\begin{array}{r@{\ \mapsto \ }l}
~\langle n,\alpha \rangle\varphi & \{ a \mid |\{ \alpha^\mathfrak{A} \cap (\{a\} {\times} \semof{\varphi}^{\mathfrak{A}}_\nu) \} | \geq n \} \\ 
~[n,\alpha ]\varphi & \{ a \mid |\{ \alpha^\mathfrak{A} \cap (\{a\} {\times} ( A \setminus \semof{\varphi}^{\mathfrak{A}}_\nu)) \} | \leq n \} \\ 
\end{array}
\begin{array}{r@{\ \mapsto \ }l}
\upmu X.\varphi & \bigcap \{A' \subseteq A \mid \semof{\varphi}^{\mathfrak{A}}_{\nu_{X\mapsto A'}} \subseteq A' \}\\ 
\upnu X.\varphi  & \bigcup \{A' \subseteq A \mid A' \subseteq \semof{\varphi}^{\mathfrak{A}}_{\nu_{X\mapsto A'}} \} \\
\end{array}
$

\medskip

A FEµ formula is \emph{closed} if all occurrences of set variables are in the scope of some $\upmu$ or~$\upnu$. 
A \emph{global FEµ Presburger constraint} is a Parikh constraint (cf. \Cref{def:GNF}), where all set variables have been replaced by closed FEµ formulae.
Given a set $\Pi$ of global FEµ Presburger constraints, we let 
$\mathfrak{A} \models \Pi$ if for every element of 
$\Pi$, replacing each of its closed FEµ formulae 
$\psi$ by 
$\semof{\psi}^{\mathfrak{A}}_\emptyset$ produces a statement valid in 
$\mathfrak{A}$.
A closed FEµ formula $\varphi$ is satisfiable wrt. $\Pi$ if there is some structure $\mathfrak{A} \models \Pi$ with $\semof{\varphi}^{\mathfrak{A}}_\emptyset \neq \emptyset$, in which case we call
$\mathfrak{A}$ a model of $(\varphi,\Pi)$.
\end{definition}

In fact, unrestricted satisfiability in FEµ (even without Presburger constraints) is undecidable \cite{BonattiP04}.
Decidability can be regained, however, when restricting to \emph{tame structures}, also commonly known as ``quasi-forests'' \cite{CalvaneseEO07,BonattiLMV08,ZOIQ,BednarczykR19}.

%\begin{definition}[tame structures]
%Let $\Sig = \Sig_\Consts \cup \Sig_{\Preds,1} \cup \Sig_{\Preds,2}$ be a signature as above.
%A \emph{tame structure} $\mathfrak{A}$ over $\Sig$ is a countable structure such that, for some finite set $Roots$,
%	\begin{itemize}
%		\item the domain $A$ of $\mathfrak{A}$ is a forest, i.e., a
%		prefix-closed subset of $Roots \times \mathbb{N}^*$, 
%        \item $Roots = \{\pr{a}^\mathfrak{A} \mid \pr{a} \in \SigC\}$ (we conveniently identify $r \in Roots$ with $r\varepsilon \in Roots\times \mathbb{N}^*$), and
%		\item for every $a, a'\in A$ with $(a,a') \in \pr{R}^\mathfrak{A}$ for some
%		$\pr{R}\in \Sig_{\Preds,2}$, either (i) $\{a,a'\} \cap
%		Roots \neq \emptyset$, or (ii) $a = a',$ or
%		(iii) $a$ is a child of $a'$, or (iv) $a'$ is a child
%		of $a$.
%		\end{itemize}
%
%A logic has the \emph{tame model property} if every satisfiable formula $\varphi$ has a model that is tame over the signature used by $\varphi$.
%The \emph{tame satisfiability} problem consists in deciding if a given formula has a tame model.  
%\end{definition}

\begin{definition}[tame structures]
%\begin{definition}[quasi-forest model, \citeauthor{DBLP:conf/aaai/CalvaneseEO07}, \citeyear{DBLP:conf/aaai/CalvaneseEO07}]
Let $\Sig = \Sig_\Consts \cup \Sig_{\Preds,1} \cup \Sig_{\Preds,2}$ be a signature as above.
A \emph{tame structure} $\mathfrak{A}$ over $\Sig$ is a countable structure such that, for some finite set $Roots$,
	\begin{itemize}
		\item the domain $A$ of $\mathfrak{A}$ is a forest, i.e., a
		prefix-closed subset of $\{rw \mid r\in Roots, w \in \mathbb{N}^*\}$, 
        \item the roots coincide with the named elements, i.e., $Roots = \{\pr{a}^\mathfrak{A} \mid \pr{a} \in \SigC\}$, %(we conveniently identify $r \in Roots$ with $r\varepsilon \in Roots\times \mathbb{N}^*$) 
        and
		\item for every $a, a'\in A$ with $(a,a') \in \pr{R}^\mathfrak{A}$ for some
		$\pr{R}\in \Sig_{\Preds,2}$, either (i) $\{a,a'\} \cap
		Roots \neq \emptyset$, or (ii) $a = a',$ or
		(iii) $a$ is a child of $a'$, or (iv) $a'$ is a child
		of $a$.
		\end{itemize}
A logic has the \emph{tame model property} if every satisfiable formula $\varphi$ has a model that is tame over the signature used by $\varphi$.
The \emph{tame satisfiability} problem consists in deciding if a given formula has a tame model.  
\end{definition}

While the restriction to tame structures may seem somewhat arbitrary at first, it is well justified: three maximal decidable sublogics of FEµ have the tame-model-property \cite{BonattiLMV08}, in which case satisfiability over arbitrary structures and tame structures coincide. Also, the structural restriction has some plausibility from a transition system perspective in that one distinguishes between a finite set of ``named'' states with arbitrary transitions between them and potentially infinitely many ``anonymous'' states with more restricted access.
It is easy to see that all tame structures over $\Sig = \SigC \cup \SigP$ have a treewidth not larger than $|\SigC|+1$.

\begin{restatable}{theorem}{FEmu}\label{thm:FEmu}
%\begin{theorem}
The tame satisfiability problem of the fully enriched $\mu$-calculus with global Presburger constraints is decidable.
%\end{theorem}
\end{restatable}

\begin{proof}[Proof (sketch)] 
Let $\Sig$ be a finite signature, $\varphi$ a closed FEµ formula over $\Sig$, and $\Pi$ a finite set of global FEµ Presburger constraints.
Being a tame structure over $\Sig$ can be expressed by an MSO sentence $\psi_\mathrm{tame}$. We define a translation $\mathrm{trans}_x$ mapping closed FEµ formulae to \thegoodlogic formulae with free variable $x$ such that $\mathfrak{A},\{x \mapsto a\} \models \mathrm{trans}_x(\varphi)$ iff $a \in \semof{\varphi}^{\mathfrak{A}}_\emptyset$. Based on this, we exhibit another translation $\mathrm{trans}$, which maps global FEµ Presburger constraints to equivalent $\thegoodlogic$ sentences. 
Then, tame satisfiability of  $(\varphi,\Pi)$ corresponds to satisfiability of the \thegoodlogic sentence $\psi_\mathrm{tame} \wedge \exists x.\mathrm{trans}_x(\varphi) \wedge \bigwedge \mathrm{trans}(\Pi)$ over all countable structures of treewidth $\leq |\SigC|+1$, which is decidable by \Cref{cor:whateverwidth}.
\end{proof}

Thanks to the expressive power of FEµ, the above result transfers to numerous other prominent logics (and their fragments), including PDL and CTL$^*$  as well as the description logics $\mu\mathcal{ALCOIQ}$ and $\mathcal{ALCOIQ}^\mathrm{reg}$ \cite{calvanese_degiacomo_2007}, for all of which tame satisfiability is thus decidable even in the presence of global Presburger constraints. 
The argument easily extends to the description logic $\mathcal{ZOIQ}$ \cite{ZOIQ}, adding Boolean combinations of binary predicates (programs).

%% file: conclusion.tex
We have proposed \thegoodlogic, a logic with a high combined structural and arithmetic expressivity, subsuming and properly extending existing popular formalisms for either purpose. We have established decidability of the satisfiability of \thegoodlogic formulae over arbitrary tree-interpretable classes of structures. A key role is played by Parikh-Muller Tree Automata, a novel type of automaton over labeled infinite binary trees with decidable emptiness.

For improving readability and succinctness, the syntax of our formalism could be extended by ``comprehension expressions'': set terms of the form
$
\cmpr{x}{\psi}
$
with $x \in \Vind$ and $\psi \in \Form$, whose semantics is straightforwardly defined by $\int{\cmpr{x}{\psi}} = \{ a \in A \mid \mathfrak{A},\nu_{\tiny x \mapsto a} \models \psi\}$.
E.g., this allows us to write
$\pr{2}\,\cnts{\cmpr{x}{\exists y. \pr{R}(x,y)}} \tteq \pr{3}\,\cnts{\cmpr{y}{\exists x. \pr{R}(x,y)}}$ rather than the more unwieldy
$$\exists V_1. (\forall x.V_1(x) \Leftrightarrow \exists y. \pr{R}(x,y)) \wedge
  \exists V_2. (\forall y.V_2(y) \Leftrightarrow \exists x. \pr{R}(x,y)) \wedge
\pr{2}\,\cnts{V_1} \tteq \pr{3}\,\cnts{V_2}.$$
Note that comprehension expressions do not increase expressivity; they can be removed from a formula~$\varphi$ yielding an equivalent formula $\varphi'$ as follows:
Let $\chi$ be the largest subformula of~$\varphi$ that contains the expression $\cmpr{x}{\psi}$ but no quantifiers binding any of the free variables of~$\psi$. Then, obtain $\varphi'$ from $\varphi$ by replacing $\chi$ by $\chi'$, where
$\chi' \coloneqq \exists Z.(\forall x.Z(x) \Leftrightarrow \psi) \wedge \chi[ \cmpr{x}{\psi} \mapsto Z].$
\thegoodlogic membership of such extended formulae can then be decided based on their ``purified'' variant,\footnote{The described removal technique is optimized toward producing formulae in \thegoodlogic.} or by means of an elaborately refined analysis of variable interactions. 
 
%If $\varphi$ contains several such expressions, this procedure needs to be repeated until all of them are eliminated. 

Concluding, we are quite confident that this paper's findings and techniques will prove useful as a generic tool for establishing decidability results for formalisms from various areas of computer science such as knowledge representation or verification. That said, in view of the non-elementary blow-ups abounding in our methods, we concede that they are unlikely to be helpful in more fine-grained complexity analyses, once decidability is established.

%% file: app-undecidable.tex
\pagebreak

\section{Mildly Extending \thegoodlogicheading
%\thegoodlogic 
Leads to Undecidability}

Given infinite trees $\xi, \xi_1\in T_\Sigma^\omega$ and a position $\rho\in\pos(\xi)$, we denote by $\xi[\rho\to\xi_1]$ the tree $\xi'$ resulting from substituting $\xi_1$ at position $\rho$ into $\xi$, i.e., for each $\rho'\in\pos(\xi')$ we obtain $\xi'(\rho')=\xi_1(\rho_1)$ if $\rho'=\rho\rho_1$ for some $\rho_1\in\{0,1\}^*$ and $\xi'(\rho')=\xi(\rho')$ otherwise. We abbreviate $(\ldots(\xi[\rho_1\to\xi_1])\ldots)[\rho_n\to\xi_n]$ by $\xi[\rho_1\to\xi_1,\ldots,\rho_n\to\xi_n]$.
Furthermore, we denote by $[\xi]_\rho$ the subtree of $\xi$ at position $\rho$, i.e., $[\xi]_\rho(\rho')=\xi(\rho\rho')$ for each $\rho'\in\{0,1\}^*$. We let $|\xi|_\sigma=|\{\rho\in\{0,1\}^*\mid\xi(\rho)=\sigma\}|$ for each $\sigma\in\Sigma$.

Given two numbers $m,n\in\Nat$ such that $m< 2^n$ we denote by $\bin{m}^n$ the binary representation of $m$ using $n$ digits (with the least significant digit on the right).

\undecidable*

\begin{proof} Let $\mathcal{D}=(NV,M,(n_w)_{w\in M},(m_w)_{w\in M})$ be a positive Diophantine equation with number variables $NV=\{\nv{z_1},\ldots,\nv{z_k}\}$\\

$\Rightarrow:$  Assume that $\nu:NV \to \mathbb{N}$ is a solution of $\mathcal{D}$.

Let $\Sigma=M\cup NV$ and $\hat{\Sigma}=\{\hat\alpha\mid\alpha\in\Sigma\}$. We construct a model $\xi\in T_{\Sigma\cup\hat\Sigma}^\omega$ of $\varphi_\mathcal{D}$ with the following intuition: for each position $\rho$ in $\xi$ labeled by $w$ and all variables $\nv{z}_{i_1},\ldots,\nv{z}_{i_l}$ with $w\nv{z}_{i_j}\in M$ we pick up a level $n$ in $[\xi]_\rho$ with enough space to label in parallel nodes from $\rho\cdot\{0,1\}^n$ by $w\nv{z}_{i_j}$ -- such that, for each $j\in[l]$, $w\nv{z}_{i_j}$ occurs $\nu(\nv{z}_{i_j})$ many times in $[\xi]_\rho$. The nodes in between are labeled by the padding symbol $\hat{w}$. \\

Let $w\in M\cup NV$ and $\zeta_w \in T_{\{w,\hat w\}}$ such that $\zeta_w(\varepsilon)=w$ and $\zeta_w(\rho)=\hat{w}$ for each $\rho\in\{0,1\}^+$. Then $\xi=\xi_\varepsilon$ with $\xi_\varepsilon$ recursively defined as follows. For the base case, let $w\in M\cup NV$ such that there is no $\nv{z}\in NV$ with $w\nv{z}\in M$. Then $\xi_w=\zeta_w$.

For the inductive case let $w\in M \cup NV$ such that there is at least one $\nv{z}_i\in NV$ with $w\nv{z}_i\in M \cup NV$. For every such $w\in M \cup NV$, let 
\begin{itemize}
%    \item $l\leq k$ and $\vec{i}=(i_1,\ldots, i_l)\in[k]^l$ such that $i_1<\ldots<i_l$ and $s\in\{i_1,\ldots,i_l\}$ iff $w\nv{z}_{s}\in M$ for each $s\in[k]$, {\small\textcolor{red}{($\vec{i}$ is the vector of all number variables extending $w$)}}
    \item $m_{w}=\displaystyle\sum_{w\nv{z}_i\in M\cup NV} \nu(\nv{z}_i)$ and let $n\geq 1$ such that $2^{n-1}\leq m_{w}\leq 2^n$, and 
    % {\small\textcolor{red}{($n$ is the level in which we substitute trees $\xi_{w\nv{z}}$)}}
    \item for every $j\leq k$, let $u^w_{j}=\displaystyle\sum_{w\nv{z}_s\in M\cup NV,\, 1\leq s<j}\nu(\nv{z}_s)$. 
    % {\small\textcolor{red}{(the number of positions that are already occupied by smaller variables (substitution offset))}}
\end{itemize}

Then $\xi_w=\zeta_w[\mathrm{SUB}_w]$, where $\mathrm{SUB}_w$ is the substitution sequence composed of all subsequences $\mathrm{sub}_{w\nv{z}_j}$ for all $w\nv{z}_j \in M\cup NV$ where 
$$ \mathrm{sub}_{w\nv{z}_j} \ \ \coloneqq \ \ \bin{u^w_j}^n\to\xi_{w\nv{z}_j},\ \ldots\ ,\bin{u^w_j+\nu(\nv{z}_{j})-1}^n\to\xi_{w\nv{z}_j} $$
if $\nu(\nv{z}_j)\neq 0$ and $\mathrm{sub}_{w\nv{z}_j}\coloneqq\varepsilon$ otherwise.

\color{black}
 The next two observations follow easily from the fact that $\xi_w$ is constructed by substituting trees $\xi_{w\nv{z}_i}$ in the $n$th level ($n>1$) of $\zeta_w$ with leaving the root at level $0$ unchanged.

\begin{observation}\label{obs1}
    For each $w\in M$ we have $\xi_w(\varepsilon)=w$.
\end{observation}

\begin{observation}\label{obs2}
    For each $u\in\pos(\xi_\varepsilon)$ we have $\xi_\varepsilon(u)=w\nv{z}$ for some $w\in M$ and $\nv{z}\in NV$ if and only if \begin{enumerate}
    \item[(a)] $[\xi_\varepsilon]_u=\xi_{w\nv{z}}$ and
    \item[(b)] there is a $v\in\{0,1\}^*$ with $v<_{\mathrm{prefix}}u$ and $\xi_\varepsilon(v)=w$, and for each $\rho\in\{0,1\}^*$ with $v<_{\mathrm{prefix}}\rho<_{\mathrm{prefix}}u$ we have $\xi_\varepsilon(\rho)=\hat{w}$. 
\end{enumerate}
\end{observation}

The mapping $\nu\colon NV\to \Nat$ can be extended to $\hat\nu\colon NV^*\to \Nat$ by letting $\hat\nu(\nv{z}_{i_1}\ldots\nv{z}_{i_l})=\nu(\nv{z}_{i_1})\cdot\ldots\cdot\nu(\nv{z}_{i_l})$; in the following we identify $\hat\nu$ and $\nu$.
\begin{claim}\label{claim}\label{claim}
    For each $w\in M$ and $\nv{z_i}\in NV$ such that $w\nv{z_i}\in M$ we obtain $|\xi_\varepsilon|_{w\nv{z_i}}=\nu(w\nv{z_i})$.
\end{claim}

\begin{proof}
    By induction on $w$: For the induction base assume that $w=\varepsilon$. By construction and \cref{obs2}, each node labeled $\nv{z}_i$ is the root of a subtree $\xi_{\nv{z}_i}$ and originates from the substitution of $\mathrm{sub}_{\nv{z}_{i}}$ into $\zeta_\varepsilon[\mathrm{sub}_{\nv{z}_{1}},\ldots,\mathrm{sub}_{\nv{z}_{i-1}}]$. By definition of $\mathrm{sub}_{\nv{z}_{i}}$, this substitution takes place $\nu(\nv{z}_{i})$ many times at parallel positions. Thus, $|\xi_\varepsilon|_{\nv{z_i}}=\nu(\nv{z_i})$.

    Now assume that $w=u\nv{z}\in M$ with $\nv{z}\in NV$ and $|\xi_\varepsilon|_w=\nu(w)$. Moreover, let $\nv{z}_i\in NV$ such that $w\nv{z}_i\in M$. By \cref{obs2} (and with the same argumentation as above), each node labeled $w\nv{z}_i$ occurs $\nu(\nv{z}_i)$ many times in $\xi_w$. As, by assumption, $w$ and, thus, also $\xi_w$ occurs $\nu(w)$ many times in $\xi_\varepsilon$, we obtain
    \begin{align*}
    |\xi_\varepsilon|_{w\nv{z_i}}=|\xi_w|_{\nv{z}_i}\cdot|\xi_\varepsilon|_w=\nu(\nv{z_i})\cdot\nu(w)=\nu(w\nv{z}_i)\,.
    \end{align*}
\end{proof}

% \begin{proof}
%     By induction on $w$: For the induction base assume that $w=\varepsilon$. By construction and \cref{obs2}, each node labeled $\nv{z}_i$ is the root of a subtree $\xi_{\nv{z}_i}$ and originates from the substitution of $\mathrm{sub}_{\nv{z}_{i}}^{(1,\ldots,k)}$ into $\zeta_\varepsilon[\mathrm{sub}_{\nv{z}_{1}}^{(1,\ldots,k)},\ldots,\mathrm{sub}_{\nv{z}_{i-1}}^{(1,\ldots,k)}]$. By definition of $\mathrm{sub}_{\nv{z}_{i}}^{(1,\ldots,k)}$, this substitution takes place $\nu(\nv{z}_{i})$ many times at parallel positions\todo{explain?}. Thus, $|\xi_\varepsilon|_{\nv{z_i}}=\nu(\nv{z_i})$.

%     Now assume that $w=u\nv{z}\in M$ with $\nv{z}\in NV$ and $|\xi_\varepsilon|_w=\nu(w)$. Moreover, let $\nv{z}_i\in NV$ such that $w\nv{z}_i\in M$. By \cref{obs2} (and with the same argumentation as above), each node labeled $w\nv{z}_i$ occurs $\nu(\nv{z}_i)$ many times in $\xi_w$. As, by assumption, $w$ and, thus, also $\xi_w$ occurs $\nu(w)$ many times in $\xi_\varepsilon$, we obtain
%     \begin{align*}
%     |\xi_\varepsilon|_{w\nv{z_i}}=|\xi_w|_{\nv{z}_i}\cdot|\xi_\varepsilon|_w=\nu(\nv{z_i})\cdot\nu(w)=\nu(w\nv{z}_i)\,.
%     \end{align*}
% \end{proof}

Now we want to show that $\xi_\varepsilon$ is indeed a model of $\varphi_\mathcal{D}$, i.e., $\xi_\varepsilon\models\varphi_{\mathrm{lab}}$, $\xi_\varepsilon\models\varphi_{\mathrm{prod}}$, and $\xi_\varepsilon\models\varphi_{\mathrm{sol}}$.

$\xi_\varepsilon\models\varphi_{\mathrm{lab}}$: It follows from \cref{obs1} that $\xi_\varepsilon\models\exists x\in\pr{P}_\varepsilon.\varphi_\mathrm{root}(x)$. Now let $w\in M$, $u\in\pos(\xi_\varepsilon)$ such that $\xi_\varepsilon(u)\in\{w,\hat{w}\}$, and $v\in\{u0,u1\}$. If $\xi_\varepsilon(u)=w$, by \cref{obs2}, $[\xi_{\varepsilon}]_u=\xi_w$ and, thus, either $\xi_\varepsilon(v)=\hat{w}$ or $[\xi_\varepsilon]_v=\xi_{w\nv{z}_j}$ and $\xi_\varepsilon(v)=\xi_{w\nv{z}_j}(\varepsilon)$ for some $j\in[k]$. By \cref{obs1}, $\xi_{w\nv{z}_j}(\varepsilon)=w\nv{z}_j$. If $\xi_\varepsilon(u)=\hat{w}$, we can argue similarly. Hence, $\xi_\varepsilon\models \bigwedge\limits_{\mathclap{w\in M}} \big(\forall x {\in} \pr{P}_w\cup\pr{P}_{\hat{w}}.\forall y. x \!\succ\! y \Rightarrow \pr{P}_{\hat{w}}(y) \vee \bigvee\limits_{\mathclap{w\nv{z}_i \in M}}\pr{P}_{w\nv{z}_i}(y)\big)$.

$\xi_\varepsilon\models\varphi_{\mathrm{prod}}$: By analyzing $\varphi_{\mathrm{prod}}$ we observe that $\xi_\varepsilon\models\varphi_{\mathrm{prod}}$ iff for each $w,w\nv{z}_i\in M$ and $u\in \pos(\xi_\varepsilon)$ with $\xi_\varepsilon(u)=w$ it holds that \[|[\xi_\varepsilon]_u|_{w\nv{z}_i}=|\xi_\varepsilon|_{\nv{z}_i} <\infty.\]
By \cref{claim}, $|\xi_\varepsilon|_{\nv{z}_i}=\nu(\nv{z}_i)$. On the other hand, by \cref{obs2}, $|[\xi_\varepsilon]_u|_{w\nv{z}_i}$ corresponds to the number of occurrences of $\xi_{w\nv{z}_i}$ in $\xi_w$ which is ensured by $\mathrm{sub}_{w\nv{z}_{i_j}}$ to be $\nu(\nv{z}_i)$, too.

$\xi_\varepsilon\models\varphi_{\mathrm{sol}}$: By assumption we know that the equation 
$$
{\sum}_{\smash{w=\nv{z}_1^{i_1}\!\ldots\nv{z}_k^{i_k}} {\in} M} n_w\cdot\nu(\nv{z}_1\hspace{-1pt})^{i_1}\cdot\ldots\cdot\nu(\nv{z}_k\hspace{-1pt})^{i_k} {\,=}  
{\sum}_{\smash{w=\nv{z}_1^{i_1}\!\ldots\nv{z}_k^{i_k}} {\in} M} m_w\cdot\nu(\nv{z}_1\hspace{-1pt})^{i_1}\cdot\ldots\cdot\nu(\nv{z}_k\hspace{-1pt})^{i_k}\, .
$$ is satisfied. Moreover, by \cref{claim}, $|\xi_\varepsilon|_w=\nu(w)$ for each $w\in M$. It follows that $\xi_\varepsilon\models\sum_{w \in M} \pr{n}_w\,\cnts{\pr{P}_w} \tteqfin \sum_{w \in M} \pr{m}_w\,\cnts{\pr{P}_w}$.\\

$\Leftarrow:$ The proof for the other direction works with a similar argumentation. Essentially, as \cref{claim} can be shown for an arbitrary model $\xi$ of $\varphi_D$ and a mapping $\nu:NV \to \mathbb{N}$ given by $\nu(\nv{z})=|\xi|_\nv{z}$, we obtain that $\nu$ is a solution of $\mathcal{D}$.
\end{proof}

%% file: app-normalization.tex
\section{Stepwise Simplification of \thegoodlogicheading Formulae}

Throughout the transformation, we will make use of formulae of the shape $\varphi_S(\iota)$ where $S$ is a set term and $\iota$ is an individual term.
This is defined inductively as follows:
\begin{align*}
  \varphi_{\{\pr{a}\}}(\iota) &\ := \  \{\pr{a}\}(\iota) \\
\varphi_{\pr{P}}(\iota) &\ := \  \pr{P}(\iota) \\    
 \varphi_{X}(\iota) &\ := \  X(\iota) \\
 \varphi_{S^c}(\iota) &\ := \  \mathrm{NNF}(\neg \varphi_{S}(\iota)) \\
 \varphi_{S_1 \cup S_2}(\iota) &\ := \   \varphi_{S_1}(\iota) \wedge  \varphi_{S_2}(\iota) \\
 \varphi_{S_1 \cap S_2}(\iota) &\ := \   \varphi_{S_1}(\iota) \vee  \varphi_{S_2}(\iota)
\end{align*}
It should be clear that $\varphi_S(\iota)$ and $S(\iota)$ are equivalent.

\medskip

Now we describe the normalization process. In the course of our treatise, we will often speak of (and introduce) \emph{fresh} variables or signature elements. By this, we mean symbols that are entirely new and have not been seen before; in particular, they are neither already present in the current formula, nor in any of the formula's previous ``versions'' throughout the normalization process.  

\subsection{Simplification and Skolemization}

\begin{itemize}
\item Remove complex set terms from finiteness and modulo atoms:
\begin{itemize}
    \item Replace every $\mathrm{Fin}(S)$ where $S$ is not a set variable by $\exists Z.Z=S \wedge \mathrm{Fin}(Z)$ where $Z$ is a fresh set variable. 
    %\todo{note that fresh means that $Z$ neither occurs as variable nor as index of a new predicate $P_X$ (vielleicht sollte das auch erst nach dem ersten Einführen von neuen Prädikaten gesagt werden)}
    \item Replace every $\cnts{S}\equiv_{\pr{m}}\pr{n}$ where $S$ is not a set variable by $\exists Z.Z=S \wedge \cnts{Z}\equiv_{\pr{m}}\pr{n}$ where $Z$ is a fresh set variable. 
\end{itemize}
We note that these replacements do not change the status of any variable inside $S$ and the result will remain in \thegoodlogic.
\item Remove set operations from set atoms by replacing every $S(\iota)$ by $\varphi_S(\iota)$, whenever $S$ contains any of $\cdot^c$, $\cap$, $\cup$. Again, the transformation clearly preserves membership in \thegoodlogic.
\item Remove simple Presburger atoms by the following subsequent equivalent transformations:
\begin{itemize}
    \item Evaluate all (simple) Presburger atoms without occurrence of any $\cnts{S}$, and, depending on the result, replace them with $\mathbf{true}$ or $\mathbf{false}$.   
    \item Replace all simple Presburger atoms of the shape $t \ttleq \ttinfty$ by $\mathbf{true}$.
    \item Replace all simple Presburger atoms of the shape $t \ttleqfin \ttinfty$ by $\mathbf{false}$.
    \item Replace all simple Presburger atoms of the shape $t \plus \pr{n} \ttleqffin t' \plus \pr{m}$ by 
\begin{itemize}
    \item $t \plus \pr{k} \ttleqffin t'$ if $n > m$, and
    \item $t \ttleqffin t' \plus \pr{k}$ otherwise,
\end{itemize}
     where $k = |m-n|$.
    \item Replace $\pr{m} \ttleqffin \pr{n}\, \cnts{S}$ by $\pr{k} \ttleqffin \cnts{S}$, where $k = \lceil \frac{m}{n} \rceil$.
    \item Replace $\pr{n}\, \cnts{S} \ttleqffin \pr{m}$ by $\cnts{S} \ttleqffin \pr{k}$, where $k = \lfloor \frac{m}{n} \rfloor$.
    \item Replace $\pr{0} \ttleq \cnts{S}$ and $\pr{0} \ttleq \cnts{S}\plus\pr{n}$ by $\true$.
    \item Replace $\pr{0} \ttleqfin \cnts{S}$ and $\pr{0} \ttleqfin \cnts{S}\plus\pr{n}$ by $\mathrm{Fin}(S)$.
    \item Replace $\pr{m} \ttleqfin \cnts{S}$ by $\mathrm{Fin}(S) \wedge \pr{m} \ttleq \cnts{S}$.
    \item Replace $\pr{m} \ttleq \cnts{S}$ by 
          $$\exists x_1\ldots x_m. \bigwedge_{1\leq i\leq m}\varphi_S(x_i) \wedge \bigwedge_{1\leq i<j\leq m} x_i \neq x_j.$$
    \item Replace $\cnts{S}\ttleqffin \pr{n}$ 
          by $$\forall x_1\ldots x_{n+1}. \bigvee_{0\leq i\leq n+1}\varphi_{S^c}(x_i) \vee \bigvee_{0\leq i<j\leq n+1} x_i = x_j.$$
          (recall that the empty conjunction equals $\true$ and the empty disjunction equals $\false$).
    \item Replace $\ttinfty \ttleq \cnts{S}$ by $\neg \mathrm{Fin}(S)$.
\end{itemize}
It should be clear that this transformation produces an equivalent formula and entirely removes all simple Presburger atoms from the considered formula.
\item Skolemize all assertive set and individual variables: 
\begin{itemize}
\item remove trailing $\exists X$ and $\exists x$
\item replace all free occurrences of $X$ by $\pr{P}_X$ (fresh unary predicate)
\item replace all free occurrences of $x$ by $\pr{c}_x$ (fresh individual constant)
\end{itemize}
This is a satisfiability-preserving transformation and afterward all assertive variables are gone.  
In particular, all variables that now still occur in Presburger atoms are delicate. 
\end{itemize}

\subsection{Presburgerization: Separation of Variables}

We \emph{``presburgerize''} all non-Presburger atoms with delicate set variables. This may require to introduce further auxiliary unary predicates, extending the signature, and is done in two steps:
\begin{enumerate}
    \item 
One by one, turn each delicate individual variable $y$ into a fresh delicate set variable $Y$ through the following procedure: 
\begin{itemize}
\item in case $y$ is existentially quantified, replace the subformula $\exists y.\varphi$ by $$\exists Y.\big((\cnts{Y}\tteq\mathtt{1}) \wedge \varphi^y_{Y}\big)$$ 
\item in case $y$ is universally quantified, replace $\forall y.\varphi$ by $$\forall Y.\big((\cnts{Y} \tteq \pr{1}) \Rightarrow \varphi^y_{Y}\big)$$ 
\item thereby, given any formula $\varphi$ with a free variable $y$, and any fresh set variable $Y$, we obtain $\varphi^y_Y$ from $\varphi$ as the end product of an exhaustive transformation sequence $\varphi=\varphi_0 \leadsto \varphi_1 \leadsto \ldots \leadsto \varphi_h=\varphi^y_Y$ where $\varphi_{i+1}$ is obtained from $\varphi_{i}$ by one of the following actions: 
\begin{itemize}
\item replacing any set atom $S(y)$ by $\mathtt{1} \ttleq \cntp{S\cap Y}$
%\item replacing any atom of the form $y\tteq y$ by $\true$
%\item replacing any atom of the form $y\tteq \iota$ or $\iota\tteq y$ by $Y(\iota)$
\item picking an atom $\pr{Q}(\iota_1,\ldots,\iota_n)$ in $\varphi_i$ that contains $y$ and
letting 
$$
\varphi_{i+1} = \varphi'_i \wedge \forall y'. ( \pr{Q}(\iota_1,\ldots,\iota_n)[y\mapsto y']\Leftrightarrow \pr{P}_{\pr{Q}(\iota_1,\ldots,\iota_n)}(y')),
$$
where $y'$ is a fresh (and obviously non-delicate) individual variable and $\varphi'_i$ is obtained from $\varphi_i$ by replacing  
$\pr{Q}(\iota_1,\ldots,\iota_n)$ with $\mathtt{1} \ttleq \cntp{\pr{P}_{\pr{Q}(\iota_1,\ldots,\iota_n)}\cap Y}$.\footnote{Note that, by the assumption of \thegoodlogic (i.e., each predicate atom contains at most one delicate variable), $y$ is the only variable occurring in $\pr{Q}(\iota_1,\ldots,\iota_n)$.}
\end{itemize}
\end{itemize}
Under the given circumstances, the transformation produces an equivalent formula in which all delicate individual variables have been removed.
\item Turn all remaining non-Presburger atoms containing some delicate set variable (and for which then \emph{all} contained individual and set variables must be delicate by definition) into subformulae with only Presburger atoms $t \ttleqffin t'$:  
\begin{itemize}
    \item Replace any $\cnts{S} \equiv_\pr{n} \pr{m}$ with delicate variables by 
    $$\exists \nv{k}. \big( \cnts{S} \tteq \pr{n}\,\nv{k} \plus \pr{m} \big).$$
    \item replace any finiteness atom $\mathrm{Fin}(S)$ with delicate variables by $\neg(\ttinfty \ttleq \cnts{S})$.
    \item replace any set atom of the shape $S(\pr{a})$ with delicate variables and constant $\pr{a}$ by $\mathtt{1} \ttleq \cntp{S\cap \{\pr{a}\}}$.  
\end{itemize}
\end{enumerate}

At the end of this sequence of transformations, we have obtained a separation of variable types: 
Next to number variables, Presburger atoms $t \ttleqffin t'$ contain exclusively delicate set variables, while all other kinds of atoms contain neither delicate nor number variables.  

\subsection{Disentangling Quantifiers}

An \thelogic formula is said to be in \emph{negation normal form} (short: NNF), if $\neg$ only occurs directly in front of predicate, set, or finiteness atoms.
Every \thelogic formula can be equivalently transformed into NNF, using the commonly known equivalences plus:
\begin{align*}
\neg (t \ttleqfin t') &\quad \equiv \quad (t' \plus \pr{1} \ttleqfin t) \vee (\ttinfty \ttleq t \plus t') \\[-0.5ex]
\neg (t \ttleq t') &\quad \equiv \quad (t' \plus \pr{1} \ttleqfin t) \vee ((\ttinfty \ttleq t)  \wedge (t' \ttleqfin t')) \\[-0.5ex]
\neg ( \cnts{S} \equiv_\pr{n} \pr{m} ) &\quad \equiv \quad (\ttinfty \ttleq \cnts{S}) \vee  \textstyle\bigvee_{\!\!{k \in \{0,\ldots,n-1\}}\atop{k \neq m \!\!\mod n}} \cnts{S} \equiv_\pr{n} \pr{k}
\end{align*}
Also it is easy to see that the normal form of a \thegoodlogic formula is again in \thegoodlogic. 
For what follows, we assume that the considered formulas are in NNF.

The purpose of the \emph{``disentangeling''} step of our normalization procedure is the following:
While the previous transformations have ensured a variable separation between the different atom types
(Presburger atoms use exclusively number and delicate set variables, whereas all other atoms use exclusively non-delicate individual and set variables), the scopes of the respective quantifiers are still containing atoms of either type. Our goal is to make sure that the quantifier scopes of number and delicate variables contain exclusively Presburger atoms, whereas the scopes of non-delicate variables contain none.

Consequently, we will call a formula \emph{entangled}, if one or both of the following is the case: 
\begin{itemize}
\item It contains a subformula of the form $\exists X.\psi$ or $\forall X.\psi$ with delicate $X$, or $\exists \nv{k}.\psi$ or $\forall \nv{k}.\psi$ such that $\psi$ contains a predicate, set, modulo or finiteness atom.
\item It contains a subformula of the form $\exists X.\psi$ or $\forall X.\psi$ with non-delicate $X$, or $\exists x.\psi$ or $\forall x.\psi$ such that $\psi$ contains a Presburger atom.
\end{itemize}
A formula is \emph{disentangled} if it is not entangled. 
Also, we call a disentangled \thegoodlogic formula \emph{arithmetic} if it only contains Presburger atoms, while we call it \emph{(plain) CMSO} if it only contains predicate, set, modulo and finiteness atoms.
It is not hard to see that, by virtue of the negation normal form, every disentangled sentence can be written as a positive Boolean combination of arithmetic and plain CMSO formulae. Applying the distributive law and intelligent grouping, we can be even more restrictive: every disentangled formula $\psi$ allows for a \emph{disentangled disjunctive normal form} (DDNF) as well as a 
\emph{disentangled conjunctive normal form} (DCNF) with the shapes
$$
\mathsf{DDNF}(\psi) = \bigvee_{i=1}^k \psi^\mathrm{arith}_{\vee i} \wedge \psi^\mathrm{plain}_{\vee i}
\qquad
\mathsf{DCNF}(\psi) = \bigwedge_{j=1}^l \psi^\mathrm{arith}_{\wedge j} \vee \psi^\mathrm{plain}_{\wedge j}
$$
where the $\psi^\mathrm{arith}_{\vee i}$ and $\psi^\mathrm{arith}_{\wedge j}$ are arithmetic formulae while the $\psi^\mathrm{plain}_{\vee i}$ and $\psi^\mathrm{plain}_{\wedge j}$ are plain CMSO formulae.

\bigskip

Now, in order to disentangle an impure \thegoodlogic formula $\varphi$, we repeatedly do the following:
We pick an occurrence of a minimal entangled subformula $\chi$ of $\varphi$ (i.e., one all of whose proper subformulae are disentangled), and replace it with $\chi'$, the ``disentangled equivalent variant'' of $\psi$. Obviously, this procedure terminates -- the number of replacement steps needed is bounded by the number of quantifiers of $\varphi$ -- and produces a disentangled formula.

It remains to provide a method to obtain $\chi'$ from $\chi$. First observe that by assumption of minimality, $\chi$ must start with a quantifier followed by some disentangled formula $\psi$. We now make a case distinction depending on the quantifier: obtain $\chi'$ from $\chi$ through the following function (where $X$ is delicate while $Y$ isn't, and we assume $\mathsf{DDNF}(\psi)$ and $\mathsf{DCNF}(\psi)$ as above):
$$
\begin{array}{rl}
\exists \nv{k}.\psi
& \mapsto \ \bigvee\limits_{i=1}^k \exists \nv{k}.\psi^\mathrm{arith}_{\vee i} \wedge \psi^\mathrm{plain}_{\vee i}\\ 
\exists X.\psi 
& \mapsto \ \bigvee\limits_{i=1}^k \exists X.\psi^\mathrm{arith}_{\vee i} \wedge \psi^\mathrm{plain}_{\vee i}\\ 
\exists x.\psi 
& \mapsto \ \bigvee\limits_{i=1}^k \psi^\mathrm{arith}_{\vee i} \wedge \exists x.\psi^\mathrm{plain}_{\vee i}\\ 
\exists Y.\psi 
& \mapsto \ \bigvee\limits_{i=1}^k \psi^\mathrm{arith}_{\vee i} \wedge \exists Y.\psi^\mathrm{plain}_{\vee i}\\ 
\end{array}\qquad\qquad
\begin{array}{rl}
\forall \nv{k}.\psi
& \mapsto \ \bigwedge\limits_{j=1}^l \forall \nv{k}.\psi^\mathrm{arith}_{\wedge j} \vee \psi^\mathrm{plain}_{\wedge j}\\ 
\forall X.\psi 
& \mapsto \ \bigwedge\limits_{j=1}^l \forall X.\psi^\mathrm{arith}_{\wedge j} \vee \psi^\mathrm{plain}_{\wedge j}\\ 
\forall x.\psi 
& \mapsto \ \bigwedge\limits_{j=1}^l \psi^\mathrm{arith}_{\wedge j} \vee \forall x.\psi^\mathrm{plain}_{\wedge j}\\ 
\forall Y.\psi 
& \mapsto \ \bigwedge\limits_{j=1}^l \psi^\mathrm{arith}_{\wedge j} \vee \forall Y.\psi^\mathrm{plain}_{\wedge j}\\ 
\end{array}
$$
Obviously, given the structure of $\psi$ with the ensured variable separation, $\chi$ and $\chi'$ are equivalent, and $\chi'$ is indeed disentangled as claimed.   

While it may seem rather innocuous at the first glance, it should be noted that disentangling a formula may incur non-elementary blowup caused by the alternating transformations into DDNF and DCNF.  

\subsection{Vennification: Eliminating Delicate Variables}
The strategy for removing delicate variables is inspired by a very similar technique used for treating BAPA \cite{KunNguRin05}.
The basic idea is to replace them by number variables, exploiting the fact that Presburger atoms only talk about cardinalities. However, in order to properly account for set operations (and even the mere fact that sets can overlap in different ways), some pre-processing is required, where the interplay of sets is broken down to disjoint indivisible smallest subsets, usually referred to as \emph{Venn regions}.\footnote{Another way to think of this is that the cardinality of each of these Venn regions is the number of elements realizing one complete unary type.}

Given some finite set $\mathbf{U} \subseteq \Preds_1 \cup \{ \{\pr{a}\} \mid \pr{a}\in \Consts\} \cup \Vset$ of unary predicates, set variables, and constant-singleton sets, a \emph{Venn region over $\mathbf{U}$} is a set term $R$ that is an intersection containing for each $U \in \mathbf{U}$ either $U^c$ or $U$. For convenience, we will consider any two Venn regions whose set expressions contain the same $U$ and $U^c$ as syntactically equal.
We denote the set of all Venn regions over $\mathbf{U}$ by $\mathrm{VR}(\mathbf{U})$. Given some $R \in \mathrm{VR}(\mathbf{U})$ and a set term $S$ that is a Boolean combination of (a selection of) elements from $\mathbf{U}$ (short: $S$ is a \emph{set term over $\mathbf{U}$}), we write $R \models S$ to denote that $\int{R} \subseteq \int{S}$ for all structures $\mathfrak{A}$ and corresponding variable assignments $\nu$. We note that $R \models S$ can be easily decided by solving the corresponding entailment problem for Boolean formulas. Furthermore $R \not\models S$ holds exactly if $\int{R} \cap \int{S} = \emptyset$ for all structures $\mathfrak{A}$ and corresponding variable assignments $\nu$. Likewise, for any $R_1,R_2 \in \mathrm{VR}(\mathbf{U})$ with $R_1 \neq R_2$, we have that $\int{R_1} \cap \int{R_2} = \emptyset$ holds for all structures $\mathfrak{A}$ and corresponding variable assignments $\nu$. These insights allow us to equivalently re-write, for any set term $S$ over $\mathbf{U}$, the numerical term $\cnts{S}$ into 
$$ \sum_{{R \in \mathrm{VR}(\mathbf{U})}\atop{R \models S}} \cnts{R}.$$   

Given a sentence $\varphi$, in order to have direct ``access'' to the Venn regions over $\mathbf{U}(\varphi) := \Preds_1(\varphi) \cup \{ \{\pr{a}\} \mid \pr{a}\in \Consts(\varphi)\}$, we introduce one auxiliary unary predicate $\pr{V}_{\!R}$ for every $R \in \mathrm{VR}(\mathbf{U}(\varphi))$, and define the corresponding definitorial description as
the sentence
$$
\psi^\mathrm{VennDef}_\varphi:= \bigwedge_{R \in \mathrm{VR}(\mathbf{U}(\varphi))} \forall x. (\pr{V}_{\!R}(x) \Leftrightarrow \varphi_{R}(x)). 
$$
Then, $\varphi' := \varphi \wedge \psi^\mathrm{VennDef}_\varphi$ is a model-conservative extension of $\varphi$ and every model $\mathfrak{A}$ of $\varphi'$ and every variable assignment $\nu$ satisfy the correspondency 
$$  
\pr{V}_{\!R}^\mathfrak{A} = \int{\pr{V}_{\!R}} = \int{R}.
$$

After these preparations, we finally remove the delicate variables from $\varphi'$. 
The intuitive idea behind this technique is to maintain a ``histogram'' storing all sizes of all Venn Regions constructible from the unary Predicates and “currently active variables”. Thereby, the size of each Venn region $R$ is ``stored'' in a number variable called $\nv{k}_R$. Whenever a new (necessarily delicate) set variable $X$ is supposed to be introduced via an existential or universal quantifier, we instead ``refine'' the histogram of Venn regions: every prior Venn region is split in two (depending on membership or non-membership in $X$), giving rise to newly introduced number variables $\nv{k}_{R\cap X}$ and $\nv{k}_{R\cap X^c}$ which need to add up to $\nv{k}_{R}$. As discussed above, this numerical information is sufficient to provide enough information for evaluating the Presburger atoms.

Formally, this transformation is realized by replacing any of its arithmetic sub-sentences $\chi$ by
$$
\big(\exists \nv{k}_R\big)_{R \in \mathrm{VR}(\mathbf{U}(\varphi))}.\Big( 
\bigwedge_{R \in \mathrm{VR}(\mathbf{U}(\varphi))} \nv{k}_R\tteq \cnts{\pr{V}_{\!R}}\Big)  
\wedge \transf{}(\chi,\mathbf{U}(\varphi)),
$$
where $\transf{}$ is recursively defined by

$$
\begin{array}{@{\transf{}(}c@{,\mathbf{U})\ := \ }l}
\psi_1 \wedge \psi_2 & \transf{}(\psi_1,\mathbf{U}) \wedge \transf{}(\psi_1,\mathbf{U}) \\
\psi_1 \vee \psi_2 & \transf{}(\psi_1,\mathbf{U}) \vee \transf{}(\psi_1,\mathbf{U}) \\
\exists \nv{k}.\psi & \exists \nv{k}.\transf{}(\psi,\mathbf{U}) \\
\forall \nv{k}.\psi & \forall \nv{k}.\transf{}(\psi,\mathbf{U}) \\
\exists X.\psi & 
\begin{array}[t]{@{}r}
(\exists \nv{k}_{R \cap X}.\exists \nv{k}_{R \cap X^c})_{R \in \mathrm{VR}(\mathbf{U})}.
\big(\bigwedge_{R \in \mathrm{VR}(\mathbf{U})} \nv{k}_{R} \tteq \nv{k}_{R \cap X} \plus \nv{k}_{R \cap X^c}\big)\ \wedge \ \ \hfill~\\
\transf{}(\psi,\mathbf{U}\cup\{X\}) \\ 
\end{array}
\\
\forall X.\psi & 
\begin{array}[t]{@{}r}
(\forall \nv{k}_{R \cap X}.\forall \nv{k}_{R \cap X^c})_{R \in \mathrm{VR}(\mathbf{U})}.
\big(\bigvee_{R \in \mathrm{VR}(\mathbf{U})} \nv{k}_{R} \! \not\,\tteq \nv{k}_{R \cap X} \plus \nv{k}_{R \cap X^c}\big)\ \vee \ \ \hfill~\\
\transf{}(\psi,\mathbf{U}\cup\{X\}) \\ 
\end{array}
\\
\end{array}
$$

$$
\begin{array}{@{\transf{}(}c@{,\mathbf{U})\ := \ }l}
t_1 \ttleq t_2 &  \transf{}(t_1,\mathbf{U}) \ttleq \transf{}(t_2,\mathbf{U}) \\
t_1 \ttleqfin t_2 &  \transf{}(t_1,\mathbf{U}) \ttleqfin \transf{}(t_2,\mathbf{U}) \\
t_1 \plus t_2 & \transf{}(t_1,\mathbf{U}) \plus \transf{}(t_2,\mathbf{U}) \\
\pr{m}\,t &  \pr{m}\, \transf{}(t_2,\mathbf{U}) \\
\pr{n} & \pr{n}\\ 
\ttinfty & \ttinfty\\
\nv{k} & \nv{k}\\
\cnts{S} & \displaystyle\sum_{{R \in \mathrm{VR}(\mathbf{U})}\atop{R \models S}} \nv{k}_R\\ 
\end{array}
$$

After this transformation, we arrive at a modified $\varphi$ wherein all arithmetic subsentences are entirely free of individual and set variables.

\subsection{Eliminating Number Variables}

In the next step, we will eliminate all number variables from arithmetic subsentences.
To this end, we will use the well known quantifier elimination procedure from Presburger arithmetic.
As classical Presburger arithmetic is defined over the natural numbers (without $\infty$), this requires another preprocessing step. Intuitively, this preprocessing implements a case distinction for every variable introduced through quantification, as to if the variable is instatiated by a natural number proper, or $\infty$.  

Clearly, via standard transformation (extended to $\infty$ in the straightforward way), every remaining number term can be written in summarized, maximally simplified form of the shape $$\pr{n}_0 + \sum_{i=1}^k \pr{n}_i t^*_i \qquad \mbox{or} \qquad \ttinfty $$ where $n_0 \in \mathbb{N}$, each $n_i \in \mathbb{N}\setminus\{0\}$,  and every $t^*_i$ is either a number variable or of the form $\cnts{\pr{P}}$.
For any term $t$, let $\TNF{t}$ denote the corresponding equivalent term of this form.
For formulae $\varphi$, let $\TNF{\varphi}$ denote $\varphi$ with every maximal term $t$ occurring in $\varphi$ replaced by $\TNF{t}$. 

%A term is called 
%\begin{itemize}
%    \item \emph{fin-constrained} if there are no occurrences of $\nv{k}$ and $\infty$, 
%    \item \emph{fin-unconstrained} if there are no occurrences of $\nv{k}_\fin$, and
%    \item \emph{fin-mixed} otherwise.
%\end{itemize}

%A formula is called 
%\begin{itemize}
%    \item \emph{classic}, if there are no occurrences of $\exists$, $\forall$, or $\infty$, 
%    \item \emph{modern}, if there are no occurrences of $\exists^\fin$ or $\forall^\fin$, and
%    \item \emph{eclectic} otherwise.
%\end{itemize}

We now define the function $\nrm$ which simplifies arithmetic formulae as follows (where $t^*$, $t_1^*$, $t_2^*$ stand for number terms whose NF is not $\ttinfty$):

\noindent{\small
\begin{minipage}{0.5\textwidth}
\begin{align*}
\true & \mapsto \   \true \\
\false & \mapsto \   \false \\
\varphi \wedge \varphi' & \mapsto \   \nrm(\varphi) \wedge \nrm(\varphi')\\
\varphi \vee \varphi' & \mapsto \   \nrm(\varphi) \vee \nrm(\varphi')\\
\exists \nv{k}.\varphi & \mapsto \   \nrm(\TNF{\varphi[\nv{k}\mapsto\ttinfty]}) \vee \exists\nv{k}.\nrm(\varphi)\\
\forall \nv{k}.\varphi & \mapsto \   \nrm(\TNF{\varphi[\nv{k}\mapsto\ttinfty]}) \wedge \forall\nv{k}.\nrm(\varphi)\\
\end{align*}
\end{minipage}
\hspace{-0ex}
\begin{minipage}{0.5\textwidth}
\begin{align*}
t \ttleqfin \ttinfty & \mapsto \   \false\\
t \ttleq \ttinfty & \mapsto \   \true\\
\ttinfty \ttleqfin t^* & \mapsto \   \false\\
\ttinfty \ttleq t^* & \mapsto \  \false\\
t^*_1 \ttleqffin t^*_2 & \mapsto \   t^*_1 \ttleqfin t^*_2\\
\end{align*}
\end{minipage}}

%{\small
%\begin{minipage}{0.5\textwidth}
%\begin{align*}
%\nrm(\true) & := \true \\
%\nrm(\false) & := \false \\
%\nrm(\varphi \wedge \varphi') & := \nrm(\varphi) \wedge \nrm(\varphi')\\
%\nrm(\varphi \vee \varphi') & := \nrm(\varphi) \vee \nrm(\varphi')\\
%\nrm(\exists^\fin \nv{k}.\varphi) & := \exists^\fin\nv{k}.\nrm(\varphi)\\
%\nrm(\forall^\fin \nv{k}.\varphi) & := \forall^\fin\nv{k}.\nrm(\varphi)\\
%\nrm(\exists \nv{k}.\varphi) & := \nrm(\TNF{\varphi[\nv{k}\mapsto\infty]}) \vee \exists^\fin\nv{k}.\nrm(\varphi)\\
%\nrm(\forall \nv{k}.\varphi) & := \nrm(\TNF{\varphi[\nv{k}\mapsto\infty]}) \wedge \forall^\fin\nv{k}.\nrm(\varphi)\\
%\end{align*}
%\end{minipage}
%\begin{minipage}{0.5\textwidth}
%\begin{align*}
%\nrm(\infty \leq \infty) & := \true\\
%\nrm(t \leq \infty) & := \true\\
%\nrm(\infty \leq t) & := \false\\
%\nrm(t \leq t') & := t \leq t'\\
%\nrm(\neg(\infty \leq \infty)) & := \false\\
%\nrm(\neg(t \leq \infty)) & := \false\\
%\nrm(\neg(\infty \leq t)) & := \true\\
%\nrm(\neg(t \leq t')) & := \TNF{t'+1} \leq t\\
%\end{align*}
%\end{minipage}}

This is an equivalent transformation when confining the attention to structures that map all occurring $\cnts{\pr{P}}$ to finite numbers. 
Building on that, we define the function $\nrmo$ that equivalently rewrites arbitrary arithmetic subsentences as follows.
For an arithmetic subsentence $\chi$ in NNF of the overall formula $\varphi$, let 
$$
\nrmo(\chi) := \bigvee_{\mathbf{U} \subseteq \mathbf{U}(\varphi)}\hspace{-1ex} 
\bigg(\hspace{-0.5ex} 
\big(\bigwedge_{S\in \mathbf{U}} \neg\mathrm{Fin}(S)\big) \wedge \big(\bigwedge_{S\in \mathbf{U}(\varphi) \setminus \mathbf{U}} \mathrm{Fin}(S)\big) \wedge \chi^\mathbf{U} \bigg)
$$
where
$$
\chi^\mathbf{U} := \nrm\Big(\TNF{\varphi\left[\cnts{S} \mapsto \infty\right]_{S{\in}\mathbf{U}}}\Big).
$$
We note that every $\chi^\mathbf{U}$ can be seen as a formula of standard Presburger arithmetic, as all still occurring $\cnts{\pr{P}}$ are constrained to be finite by the environment of $\chi^\mathbf{U}$ in $\nrmo(\chi)$. This allows us to apply classical quantifier elimination \cite{Pre29} to each $\chi^\mathbf{U}$, yielding a sentence $\mathrm{QE}(\chi^\mathbf{U})$ free of any number variables, but possibly containing modulo atoms next to classical Presburger atoms. Letting $\mathrm{QE}(\nrmo(\chi))$ denote $\nrmo(\chi)$ with every subsentence $\chi^\mathbf{U}$ replaced by $\mathrm{QE}(\chi^\mathbf{U})$, we finally obtain the new formula $\varphi$ by replacing all arithmetic subsentences $\chi$ by $\mathrm{QE}(\nrmo(\chi))$. This yields us a sentence entirely free of number variables, which is a positive Boolean combination of sentences of the following two kinds:
\begin{itemize}
    \item sentences containing (possibly negated) predicate, set, finiteness and modulo atoms, and
    \item unnegated classical Presburger atoms without variables.
\end{itemize}

Using distributivity, we can ensure that the obtained sentence has the form
$$
\bigvee_{i=1}^k \big( \varphi_i \wedge \bigwedge_{j=1}^{l_i} \chi_{i,j} \big), 
$$
where the $\varphi_i$ are CMSO sentences, i.e., MSO sentences with possible additional occurrences of modulo and finiteness atoms, while the $\chi_{i,j}$ are unnegated classical Presburger atoms without variables.

\subsection{De-Skolemization}

At this point, the formula (say, $\varphi^\mathrm{new}$) we have arrived at uses a signature $\Sig(\varphi^\mathrm{new})$ that is a superset of the signature $\Sig(\varphi^\mathrm{orig})$ used by the original formula  $\varphi^\mathrm{orig}$. The initial Skolemization has introduced auxiliary constant and unary predicate names, and more unary predicates were introduced in the course of ``Presburgerization'' and ``Vennization''. The formula $\varphi^\mathrm{new}$ is then a model-conservative extension of the original formula $\varphi^\mathrm{orig}$: every model of the obtained formula is a model of the original one and every model of the original one can be turned into a model of the current one by picking appropriate interpretations for the added signature elements.\footnote{In case the original formula had free variables, a similar statement holds where the models are accompanied by variable assignments.} In order to re-gain equivalence with $\varphi^\mathrm{orig}$ we therefore need to ``project away'' these additional signature elements by conceiving them as individual and set variables and existentially quantifying over them. An exception is made for the representatives of the formerly free variables, which will be free again. Formally%, given 
%$$
%\Sig^\mathrm{diff}=\Sig(\varphi^\mathrm{new}) \setminus ( \Sig(\varphi^\mathrm{orig}) \cup \{\pr{P}_X %\mid X \in \free(\varphi^\mathrm{orig})\} \cup \{\pr{c}_x \mid x \in \free(\varphi^\mathrm{orig})\})
%$$
, we obtain
$$
\varphi^\text{free-again} \coloneqq 
\varphi^\mathrm{new}\big[\pr{P}_X \mapsto X\big]_{X \in \free(\varphi^\mathrm{orig})\cap\Vset}\big[\pr{c}_x \mapsto x\big]_{x \in \free(\varphi^\mathrm{orig})\cap\Vind}
$$
and, with $\Sig^\mathrm{diff}=\Sig(\varphi^\mathrm{new}) \setminus \Sig(\varphi^\mathrm{orig})$ we let
$$
\varphi^\text{de-skol} \coloneqq \big(\exists X_\pr{P}\big)_{\pr{P} \in \SigP^\mathrm{diff}}.\big(\exists x_\pr{c}\big)_{\pr{c} \in \SigC^\mathrm{diff}}.
\varphi^\text{free-again}\big[\pr{P} \mapsto X_\pr{P}\big]_{\pr{P} \in \SigP^\mathrm{diff}}\big[\pr{c} \mapsto x_\pr{c}\big]_{\pr{c} \in \SigC^\mathrm{diff}},
$$
and conclude that $\varphi^\text{de-skol}$ is equivalent to $\varphi^\mathrm{orig}$.
Now having in mind that by construction, $\varphi^\text{de-skol}$ has the shape 
$$
\exists X_1. \cdots \exists X_n.\exists x_1. \cdots \exists x_m.\bigvee_{i=1}^k \big( \varphi_i \wedge \bigwedge_{j=1}^{l_i} \chi_{i,j} \big), 
$$
where all $\chi_{i,j}$ are free of individual variables, we see that we can pull the quantifier block $\exists x_1. \cdots \exists x_m$ past the disjunction and then inside the $\varphi_i$.
Last not least, we can, for every $\cnts{\pr{P}}$ still occurring in some $\chi_{i,j}$, replace it by $\cnts{X_\pr{P}}$ (where $X_\pr{P}$ is a fresh set variable), add the conjunct $X_\pr{P}=\pr{P}$ to every $\varphi_i$ and put the quantifier $\exists X_\pr{P}$ in front of the whole formula.

With this we have – at last – arrived at a formula in GNF as specified in \cref{def:GNF}.

%% file: app-automata.tex
\section{Properties of PMTA}

Given some $n\in\Nat$, a set $A$, and a tuple $a=(a_1,\ldots,a_n)\in A^n$, we let $\proj_i(a)=a_i$ for every $i\in[n]$, that is, $\proj_i$ denotes the \emph{$i^\mathit{th}$\! projection}. Projections are lifted to sets of tuples as usual.

By writing $\xi[\xi_1,\ldots,\xi_{k}]$ we denote the composition of a finite tree $\xi$ with trees $\xi_1,\ldots,\xi_{k}$ in the usual way, i.e., we equip $\xi$ with variables (two variables under each leaf and one variable under each position with only one successor) and we replace the $i^\mathrm{th}$ variable of $\xi$ by the tree $\xi_{i}$ -- under the assumption that $k$ corresponds to the number of variables $\xi$ is equipped with. If $\pos(\xi)=\emptyset$, we let $\xi[\xi_1]=\xi_1$.

\subsection{Closure Properties}

\closurepmta*
\begin{proof}
The result follows from the subsequent Lemmas \ref{union}, \ref{intersec}, and \ref{relab}.
\end{proof}

\begin{lemma}\label{union}
PMTA are closed under union.
\end{lemma}

\begin{proof}
Let $\A_i=(\Q_i,\Xi_i,q_{I\!,i},\Delta_i,\F_i,C_i)$ for $i\in[2]$ be two PMTA with $\Q_i= Q_{P,i}\cup Q_{\omega,i}\cup\{q_{I\!,i}\}$, $\Xi_i=(\Sigma\times D_i)\cup\Sigma$, and $\Delta_i=\Delta_{P,i}\cup \Delta_{\omega,i}$. W.l.o.g. we can assume that $\A_1$ and $\A_2$ are both of dimension $s$ for some $s\geq 1$ (semi-linear sets are closed under concatenation by \cref{lemma-semilin-presburger}), that $\Q_1$ and $\Q_2$ are disjoint, and that $q_I$ is a fresh state not in $\Q_1\cup\Q_2$. 

We construct the PMTA $\A=(\Q,\Xi,q_I,\Delta,\F,C)$ where
\begin{itemize}
    \item $\Q=\Q_1\cup\Q_2\cup\{q_I\}$,
    \item $\Xi=\Bigl(\Sigma\times (D_1\cdot\{(0),(1)\}\cup D_2\cdot\{(0),(2)\})\Bigr)\cup\Sigma$,
    \item $\F=\F_1\cup \F_2$,
    \item $C=C_1\cdot\{(1)\}\cup C_2\cdot\{(2)\}$, and
    \item $\Delta=\Delta_P\cup\Delta_\omega$ consists of the transitions
    \begin{itemize}
        \item $\Delta_P=\ \{(q,(\sigma,\vec{d}\cdot(0)),q_1,q_2)\mid(q,(\sigma,\vec{d}),q_1,q_2)\in\Delta_{P,1}\cup\Delta_{P,2}\}\ \cup$\\
        \hbox{}\quad\qquad$\{(q_I,(\sigma,\vec{d}\cdot(i)),q_1,q_2)\mid(q_{I\!,i},(\sigma,\vec{d}),q_1,q_2)\in\Delta_{P,i} \text{ for some } i\in[2]\}$ and
        \item $\Delta_\omega=\ \Delta_{\omega,1}\cup\Delta_{\omega,2}\cup\{(q_I,\gamma,q_1,q_2)\mid(q_{I\!,i},\gamma,q_1,q_2)\in\Delta_{\omega,i}\text{ for some } i\in[2]\}$.
    \end{itemize}
\end{itemize}    

$\L(\A_1)\cup\L(\A_2)\subseteq\L(\A)$: Let $\xi\in\L(\A_i)$ for some $i\in[2]$. Then there exists some $\zeta\in \To_{\Xi_i}$ s.t. $(\zeta)_\Sigma=\xi$ and there is an accepting run $\kappa_\zeta$ of $\A_i$ on $\zeta$. By def., $\kappa_\zeta(\varepsilon)=q_{I\!,i}$.

(1) Assume that $\zeta_\cnt\neq\emptyset$. Consider the tree $\zeta'\in\To_\Xi$ with
\[\zeta'(\rho)=
\begin{cases}
(\sigma,\vec{d}\cdot(i)) &\text{ if } \rho=\varepsilon, \zeta(\varepsilon)=(\sigma,\vec{d})\\
(\sigma,\vec{d}\cdot(0)) &\text{ if } \rho\in\pos(\zeta_\cnt)\setminus\{\varepsilon\}, \zeta(\rho)=(\sigma,\vec{d})\\
\zeta(\rho) &\text{ otherwise}
\end{cases}.\]
Clearly, $(\zeta')_\Sigma=(\zeta)_\Sigma=\xi$. Now consider the tree $\kappa'$ with $\kappa'(\varepsilon)=q_I$ and $\kappa'(\rho)=\kappa_\zeta(\rho)$ for each $\rho\in\{0,1\}^+$. By construction and as $\kappa_\zeta$ respects the transition relation of $\A_i$, $(q_I,\zeta'(\varepsilon),\kappa'(0),\kappa'(1))\in\Delta$ and $(\kappa'(\rho),\zeta'(\rho),\kappa'(\rho0),\kappa'(\rho1))\in\Delta$ for each $\rho\in\{0,1\}^+$. Thus, $\kappa'$ is a run of $\A$ on $\zeta'$. As $\Psi(\zeta'_\cnt)=\Psi(\zeta_\cnt)\cdot(i)$ and $\Psi(\zeta_\cnt)\in C_i$, $\Psi(\zeta'_\cnt)\in (C_i\cdot \{(i)\})\subseteq C$. Finally, as $\inf(\kappa'(\pi))=\inf(\kappa_\zeta(\pi))$ for each path $\pi$, we obtain that $\kappa'$ is accepting. Hence, $\xi\in\L(\A)$.

(2) Assume that $\zeta_\cnt=\emptyset$. Similar argumentation as above.\\

$\L(\A)\subseteq\L(\A_1)\cup\L(\A_2)$: Let $\xi\in\L(\A)$. Then there exists some $\zeta\in\To_\Xi$ s.t. $(\zeta)_\Sigma=\xi$ and there is some accepting run $\kappa_\zeta$ of $\A$ on $\zeta$. 

(1) Assume $\zeta_\cnt\neq\emptyset$. Then $\zeta(\varepsilon)=(\sigma,\vec{d}\cdot(i))$ for some $i\in[2]$, $\sigma\in\Sigma$, $\vec{d}\in D_i$, and, for each $\rho\in\pos(\zeta_\cnt)\setminus\{\varepsilon\}$, $\zeta(\rho)=(\sigma_\rho,\vec{d}_\rho\cdot(0))$ for some $\sigma_\rho\in\Sigma$, $\vec{d}_\rho\in D_i$. By construction of $\Delta$,
\begin{itemize}
    \item $(q_{I\!,i},(\sigma,\vec{d}),\kappa_\zeta(0),\kappa_\zeta(1))\in\Delta_{P,i}$,
    \item $(\kappa_\zeta(\rho),(\sigma_\rho,\vec{d}_\rho),\kappa_\zeta(\rho0),\kappa_\zeta(\rho1))\in\Delta_{P,i}$ for each $\rho\in\pos(\zeta_\cnt)\setminus\{\varepsilon\}$, and 
    \item $(\kappa_\zeta(\rho'),\zeta(\rho'),\kappa_\zeta(\rho'0),\kappa_\zeta(\rho'1))\in\Delta_{\omega,i}$ for each $\rho'\in\{0,1\}^*\setminus\pos(\zeta_\cnt)$.
\end{itemize}
It follows that the tree $\kappa'$ with $\kappa'(\varepsilon)=q_{I\!,i}$ and $\kappa'(\rho)=\kappa_\zeta$ for all $\rho\in\{0,1\}^+$ is a run of $\A_i$ on the tree $\zeta'$ obtained from $\zeta$ by cropping the last element of each $\vec{d}\in D_i\cdot\{(0),(i)\}$. Clearly, $(\zeta')_\Sigma=(\zeta)_\Sigma=\xi$. As $\Psi(\zeta_\cnt)=\vec{w}\cdot(i)$ and $\Psi(\zeta_\cnt)\in C_i\cdot\{(i)\}$, we obtain that $\Psi(\zeta'_\cnt)=\vec{w}\in C_i$. Moreover, $\inf(\kappa'(\pi))=\inf(\kappa_\zeta(\pi))$ for each path $\pi$. Thus, $\kappa'$ is accepting and $\xi\in\L(\A)$.

(2) Assume $\zeta_\cnt = \emptyset$. Similar argumentation as above.
\end{proof}

\begin{lemma}\label{intersec}
PMTA are closed under intersection.
\end{lemma} 

\begin{proof}
Let, for each $i\in[2]$, $\A_i=(\Q_i,\Xi_i,q_{I\!,i},\Delta_i,\F_i,C_i)$ be a PMTA of dimension $s_i$ with $\Q_i= Q_{P,i}\cup Q_{\omega,i}\cup\{q_{I\!,i}\}$, $\Xi_i=(\Sigma\times D_i)\cup\Sigma$, and $\Delta_i=\Delta_{P,i}\cup \Delta_{\omega,i}$. W.l.o.g. we can assume that $(0)^{s_i}\in D_i$. Note that on a tree $\xi\in\L(\A_1)\cap\L(\A_2)$, $\A_1$ and $\A_2$ might not test the same initial part of $\xi$ arithmetically. In order to simulate this ``superposition'', we need to test a bigger initial part and let non-active counters ``idle'' using zero-increments $(0)^{s_i}$.

Thus, we define the Cartesian product $\A=(\Q,\Xi,(q_{I,1},q_{I,2}),\Delta,\F,C)$ where 
\begin{itemize}
    \item $\Q_\omega=Q_{\omega,1}\times Q_{\omega,2}$, and $\Q_P=(\Q_1\times\Q_2)\setminus\Q_\omega$,
    \item $\Xi=(\Sigma\times D_1\cdot D_2)\cup\Sigma$,
    \item $\F=\{F\subseteq \Q_\omega\mid \mathrm{pr}_1(F)\in\F_1 \text{ and } \mathrm{pr}_2(F)\in\F_2\}$,
    \item $C=C_1\cdot C_2$,
\end{itemize}
   and $\Delta=\Delta_P\cup\Delta_\omega$ where
\begin{itemize}
    \item $\Delta_P$ is the smallest set such that
    \begin{itemize}
        \item for each $(q,(\sigma,\vec{d}),q_1,q_2)\in\Delta_{P,1}$ and $(q',(\sigma,\vec{d}'),q_1',q_2')\in\Delta_{P,2}$ we have\\ $((q,q'),(\sigma,\vec{d}\cdot\vec{d}'),(q_1,q_1'),(q_2,q_2'))\in\Delta_P$,
        \item for each $(q,(\sigma,\vec{d}),q_1,q_2)\in\Delta_{P,1}$ and $(q',\sigma,q_1',q_2')\in\Delta_{\omega,2}$ we have\\ $((q,q'),(\sigma,\vec{d}\cdot (0)^{s_2}),(q_1,q_1'),(q_2,q_2'))\in\Delta_P$,
        \item for each $(q,\sigma,q_1,q_2)\in\Delta_{\omega,1}$ and $(q',(\sigma,\vec{d}'),q_1',q_2')\in\Delta_{P,2}$ we have\\ $((q,q'),(\sigma,(0)^{s_1}\cdot\vec{d}'),(q_1,q_1'),(q_2,q_2'))\in\Delta_P$, and
    \end{itemize}
    \item $\Delta_\omega=\{((q,q'),\gamma,(q_1,q_1'),(q_2,q_2'))\mid (q,\gamma,q_1,q_2)\in\Delta_{\omega,1}, (q',\gamma,q_1',q_2')\in\Delta_{\omega,2}\}$.
\end{itemize}

$\L(\A_1)\cap\L(\A_2)\subseteq\L(\A):$ Let $\xi\in\L(\A_1)\cap\L(\A_2)$. Then there are $\zeta_1\in \To_{\Xi_1}$, $\zeta_2\in \To_{\Xi_2}$ such that $(\zeta_1)_\Sigma=(\zeta_2)_\Sigma=\xi$ and for each $i\in[2]$ there exists an accepting run $\kappa_i$ of $\A_i$ on $\zeta_i$. Note that $\zeta_1$ and $\zeta_2$ might not have the same prefix enriched by Parikh vectors. Now consider the tree $\zeta\in\To_\Xi$ with
\[\zeta(\rho)=
\begin{cases}
(\sigma,\vec{d}\cdot\vec{d}') & \text{if } \rho\in\pos(\zeta_{1,\cnt})\cap\pos(\zeta_{2,\cnt}), \zeta_1(\rho)=(\sigma,\vec{d}), \zeta_2(\rho)=(\sigma,\vec{d}')\\
(\sigma,\vec{d}\cdot(0)^{s_2}) & \text{if } \rho\in\pos(\zeta_{1,\cnt})\setminus\pos(\zeta_{2,\cnt}), \zeta_1(\rho)=(\sigma,\vec{d}), \zeta_2(\rho)=\sigma\\
(\sigma,(0)^{s_1}\cdot\vec{d}') & \text{if } \rho\in\pos(\zeta_{2,\cnt})\setminus\pos(\zeta_{1,\cnt}), \zeta_1(\rho)=\sigma, \zeta_2(\rho)=(\sigma,\vec{d}')\\
\sigma & \text{if } \rho\notin\pos(\zeta_{1,\cnt})\cup\pos(\zeta_{2,\cnt}), \zeta_1(\rho)=\zeta_2(\rho)=\sigma
\end{cases}\] as well as the tree $\kappa\in\To_\Q$ given by $\kappa(\rho)=(\kappa_1(\rho),\kappa_2(\rho))$ for each $\rho\in\{0,1\}^*$. As $\kappa_1(\varepsilon)=q_{I,1}$ and $\kappa_2(\varepsilon)=q_{I,2}$, $\kappa(\varepsilon)=(q_{I,1},q_{I,2})$. Moreover, we obtain for each position $\rho$
\begin{itemize}
    \item if $\rho\in\pos(\zeta_{1,\cnt})\cap\pos(\zeta_{2,\cnt})$, then $(\kappa_i(\rho),\zeta_i(\rho),\kappa_i(\rho 0),\kappa_i(\rho1))\in\Delta_{P,i}$ for each $i\in[2]$. By construction of $\Delta_P$, $((\kappa_1(\rho),\kappa_2(\rho)),\zeta(\rho),(\kappa_1(\rho 0)\kappa_2(\rho 0)),(\kappa_1(\rho 1)\kappa_2(\rho 1)))\in\Delta_{P,i}$,
\end{itemize}
and for all other positions we can argue for the existence of the respective transition in a similar way. Hence, $\kappa$ is a run of $\A$ on $\zeta$. 

Furthermore, if $\pos(\zeta_{1,\cnt})\cup\pos(\zeta_{2,\cnt})\neq\emptyset$, we obtain by assuming $\zeta_\cnt(j)=(\sigma_j,\vec{d_j}\cdot\vec{d_j}')$
\[\Psi(\zeta_\cnt)=\sum_{j\in\pos(\zeta_\cnt)}\vec{d_j}\cdot\vec{d_j}'=(\sum_{j\in\pos(\zeta_\cnt)}\vec{d_j})\cdot(\sum_{j\in\pos(\zeta_\cnt)}\vec{d_j}')=\Psi(\zeta_{1,\cnt})\cdot\Psi(\zeta_{2,\cnt})\] where the last equality holds since $\pos(\zeta_{i,\cnt})\subseteq\pos(\zeta_\cnt)$ and $\vec{d}+(0)^{s_i}=\vec{d}$ for each $i\in[2]$ and $\vec{d}\in\Nat^{s_i}$. As $\Psi(\zeta_{1,\cnt})\in C_1$ and $\Psi(\zeta_{2,\cnt})\in C_2$, we obtain $\Psi(\zeta_{1,\cnt})\cdot\Psi(\zeta_{2,\cnt})\in C_1\cdot C_2$.

Finally, let $\pi$ be an arbitrary path. Then $\inf(\kappa_1(\pi))=F_1$ and $\inf(\kappa_2(\pi))=F_2$ for some $F_1\in\F_1$ and $F_2\in\F_2$. As $\proj_i(\kappa(\pi))=\kappa_i(\pi)$ and $\proj_i(\inf(\kappa(\pi)))=\inf(\proj_i(\kappa(\pi)))$, we obtain that $\proj_i(\inf(\kappa(\pi)))=\inf(\kappa_i(\pi))$ for each $i\in[2]$. Thus,  $\inf(\kappa(\pi))\in\F$ and $\kappa$ is accepting on $\zeta$. As $(\zeta)_\Sigma=\xi$, $\xi\in\L(\A)$.\\

$\L(\A)\subseteq\L(\A_1)\cap\L(\A_2):$ Similar. By construction of $\Delta$, each $\zeta$ and accepting run $\kappa\in\Run_\A(\zeta)$ are of a particular form: $\zeta$ can be decomposed into \[\hat{\zeta}[\zeta_1[\zeta^1_1,\ldots,\zeta^1_{m_1}],\ldots,\zeta_k[\zeta^k_1,\ldots,\zeta^k_{m_k}]]\] where $\hat{\zeta}$ comprises all positions that are assigned by $\kappa$ states from $Q_{P,1}\times Q_{P,2}$, each $\zeta_j$ consists completely of labels $(\sigma,\vec{d}\cdot(0)^{s_2})$ and the corresponding positions of $\kappa$ assign states from $Q_{P,1}\times Q_{\omega,2}$ (or similar the other way around), and each $\zeta^u_v$ has labels from $\Sigma$ and gets assigned states from $Q_\omega$. From this, we easily can reconstruct two trees $\zeta_1'$ and $\zeta_2'$ with $(\zeta_1')_\Sigma=(\zeta_2')_\Sigma=(\zeta)_\Sigma$ and accepting runs $\kappa_1'\in\Run_{\A_1}(\zeta_1')$ and $\kappa_2'\in\Run_{\A_2}(\zeta_2')$.
\end{proof}

\begin{lemma}\label{relab}
PMTA are closed under relabeling.
\end{lemma}

\begin{proof}[Proof (sketch)]
Let $\Sigma, \Gamma$ be alphabets, let $\A=(\Q,\Xi,q_{I},\Delta,\F,C)$ be a PMTA with $\Q= Q_{P}\cup Q_{\omega}$, $\Xi=(\Sigma\times D)\cup\Sigma$, and $\Delta=\Delta_{P}\cup \Delta_{\omega}$,
and let $\tau\colon\Sigma\to\Pow(\Gamma)$ be a relabeling. We construct the WPT3 $\A'=(\Q,\Xi',q_{I},\Delta',\F,C)$ with $\Xi'=(\Gamma\times D)\cup\Gamma$ and $\Delta'=\Delta_{P}'\cup \Delta_{\omega}'$ where
\begin{itemize}
    \item $\Delta_P'=\{(q,(\gamma,\vec{d}),q_1,q_2)\mid(q,(\sigma,\vec{d}),q_1,q_2)\in\Delta_P,\gamma\in\tau(\sigma)\}$ and
    \item $\Delta_\omega'=\{(q,\gamma,q_1,q_2)\mid(q,\sigma,q_1,q_2)\in\Delta_\omega,\gamma\in\tau(\sigma)\}$.
\end{itemize}
\end{proof}

\subsection{Correspondence with \thegoodlogicheading}

% \begin{lemma}[Lemma 13]
% For every PMTA $\A$ there is an \thegoodlogic sentence $\varphi$ with $\L(\A)=\L(\varphi)$.
% \end{lemma}
\rectodef*

\begin{proof} This proof adopts (and slightly simplifies) the idea in \cite[Thm. 10]{KlaRue03} how to logically encode counter values and the semilinear set $C$ to our setting and uses the well-known construction to define the behavior of a Muller tree automaton by means of an MSO formula.

Let $\Sigma$ be an alphabet, let $D\subseteq\Nat^s$ be finite for some $s\geq 1$, and let $\A=(\Q,\Xi,q_I,\Delta,\F,C)$. W.l.o.g. assume $\Q=\{1,\ldots,r\}$ for some $r\geq1$. As usual, we use set variables $X_{1},\ldots,X_{r}$ to encode a run of $\A$ on a tree $\zeta$.

Let $K$ be the maximal value occurring in vectors from $D$. Each position $i \in [s]$ of a vector $d=(d_1,\ldots,d_s)\in D$ carries a value $d_i\leq K$, in the following represented by the set variable $Z^{d_i}_i$. Thus, the presence of a variable $Z^{d_i}_i$ at a position $x$ of $\zeta$ indicates that $d_i$ has to be added to the $i^\mathrm{th}$ counter to simulate the extended Parikh map $\Psi(\zeta)$.

Formally, let
$$\varphi\coloneqq\exists Z_1^0\ldots Z_1^K\ldots Z_s^K\exists X_{1}\ldots X_{r}\forall P. \varphi_{\mathrm{part}_Q}\land\varphi_{\mathrm{part}_Z}\land\varphi_{\mathrm{run}}\land\varphi_{\mathrm{acc}}\land\varphi_C$$
where
\begin{itemize}
    \item $\displaystyle\varphi_{\mathrm{part}_Q}\coloneqq\forall x.\bigl( \bigwedge_{p,q\in Q, p\neq q}\neg(X_p(x) \land X_q(x))\bigr)$
    \item $\displaystyle\varphi_{\mathrm{part}_Z}\coloneqq\bigwedge_{1\leq i\leq s}\bigl(\forall x.\bigwedge_{u,v\in \{0,\ldots,K\}, u\neq v}\neg(Z_i^u(x) \land Z_i^v(x))\bigr)$
    % \item \(
    %         \begin{aligned}[t]
    %           \varphi_{\mathrm{trans}} = &\ (\varepsilon\in X_{q_I})\land \forall x.\Bigl(\\
    %             & \ \bigvee_{(q,(\sigma,\vec{d}),q_1,q_2)\in\Delta_P}\hspace{-2em} \bigl(x\in X_q \land P_\sigma(x)\land x0\in X_{q_1}\land x1\in X_{q_2}
    %             \land \bigwedge_{1\leq i\leq s}x\in Z^{(\vec{d})_i}\bigr)\\
    %             & \ \lor\\
    %             & \bigvee_{(q,\sigma,q_1,q_2)\in\Delta_\omega}\hspace{-0.5em} \bigl(x\in X_q \land P_\sigma(x) \land x0\in X_{q_1}\land x1\in X_{q_2}\\
    %             & \hspace{10em} \land \bigwedge_{\substack{1\leq i\leq s\\ 0\leq d\leq K}}\neg(x\in Z^d_i)\bigr)
    %         \Bigr)\end{aligned}
    %       \) 
    \item \(\varphi_{\mathrm{run}} \coloneqq \ \forall x.\Bigl((\varphi_\mathrm{root}(x)\Rightarrow X_{q_I}(x))\land (\varphi_{\mathrm{trans}_P}\lor \varphi_{\mathrm{trans}_\omega})\Bigr) \)
        \begin{itemize}
            \item $\varphi_{\mathrm{trans}_P}\coloneqq\displaystyle\bigvee_{(q,(\sigma,\vec{d}),q_0,q_1)\in\Delta_P}\hspace{-2em} \bigl(X_q(x) \land \pr{P}_\sigma(x)\land X_{q_0}(x0)\land X_{q_1}(x1) \land\bigwedge_{1\leq i\leq s}Z^{(\vec{d})_i}_i(x)\bigr)$
            \item $\varphi_{\mathrm{trans}_\omega}\coloneqq\displaystyle\bigvee_{(q,\sigma,q_0,q_1)\in\Delta_\omega}\hspace{-0.5em} \bigl(X_q(x) \land \pr{P}_\sigma(x) \land X_{q_0}(x0)\land X_{q_1}(x1)\land\bigwedge_{\substack{1\leq i\leq s\\ 0\leq d\leq K}}\neg Z^d_i(x)\bigr)$
        \end{itemize}
    \item $\displaystyle\varphi_{\mathrm{acc}}\coloneqq\varphi_\mathrm{path}(P)\Rightarrow \bigvee_{F\in\F}(\bigwedge_{q\in F} \varphi^\cap_\mathrm{inf}(X_q,P)\land\bigwedge_{q\notin F}\neg \varphi^\cap_\mathrm{inf}(X_q,P))$
\end{itemize}
and $\varphi_C$ encodes the Parikh condition as follows. By definition of $C$, there exist some $k,\ell\geq 1$ and linear polynomials $p_1,\ldots,p_k\colon\Nat^\ell\to\Nat^s$ such that $C$ is the union of the images of $p_1,\ldots,p_k$. 
% Assume that $$p_g(m_1,\ldots,m_\ell)=\vec{v}_0+m_1\vec{v}_1+\ldots+m_\ell\vec{v}_l$$
% with $\vec{v}_j=(v_{j,1},\ldots,v_{j,s})$.
% We will represent the values $m_1,\ldots,m_\ell$ by set variables $M_1,\ldots,M_\ell$ and, thus, encode $p_g$ by
% % $$\begin{aligned}[t]
% %   \varphi_{p_t}=&\exists M_1\ldots M_\ell\exists V_1\ldots V_s.\\
% %   &\bigwedge_{1\leq i\leq s} \Bigl(\bigl(\forall v.v\in V_i\leftrightarrow (\bigvee_{1\leq n\leq \vec{v}_0^i}v=1^n)\bigr)\land\\
% %   &\hspace{3em}\sum_{1\leq d\leq K}d|Z_j^d|=|V_i|+\vec{v}_1^i|M_1|+\ldots+\vec{v}_\ell^i|M_\ell|\Bigr).
% % \end{aligned}$$
% $$\begin{aligned}[t]
%   \varphi_{p_g}\coloneqq\exists M_1\ldots M_\ell.\bigwedge_{1\leq i\leq s} \Bigl(\sum_{1\leq d\leq K}\pr{d}\cnts{Z_i^d}=\pr{v_{0,i}}+\pr{v_{1,i}}\cnts{M_1}+\ldots+\pr{v_{\ell,i}}\cnts{M_\ell}\Bigr).
% \end{aligned}$$
% Finally, we let
% \[\varphi_C\coloneqq(\bigwedge_{\substack{1\leq i\leq s\\ 0\leq d\leq K}}\forall x. \neg Z_i^d(x))\lor\varphi_{p_1}\lor\ldots\lor\varphi_{p_k}.\]
% Note that we can pull out the set quantifications of $\varphi_{p_g}$. Hence, $\varphi_{p_g}$ is protected and, thus, $\varphi$ is an \thegoodlogic sentence. \todo{anpassen an aktuellen Beweis}
Assume $p_g(m_1,\ldots,m_l)=\vec{v}_0+m_1\vec{v}_1+\ldots+m_l\vec{v}_l$
with $\vec{v}_j=(v_{j,1},\ldots,v_{j,s})$.
Then, using number variables $\nv{m}_1,\ldots,\nv{m}_l$, we encode $p_g$ by
$$\textstyle\varphi_{p_g}\coloneqq\exists \nv{m}_1\ldots \exists\nv{m}_l.\bigwedge_{i=1}^s \Bigl(\sum_{d=0}^K\pr{d}\cnts{Z_i^d}\tteqfin\pr{v}_{0,i}+\pr{v}_{1,i}\nv{m}_1+\ldots+\pr{v}_{l,i}\nv{m}_l\Bigr),$$
and let
$\textstyle\varphi_C\coloneqq\big(
%\bigwedge_{\tiny\raisebox{1ex}{$\substack{1\leq i\leq s\\ 0\leq d\leq K}$}}
\bigwedge_{i=1}^s
\bigwedge_{d=0}^K
\forall x. \neg Z_i^d(x)\big)\lor\varphi_{p_1}\lor\ldots\lor\varphi_{p_k}$. This finishes the construction of the overall sentence specifying $\L(\A)$, which can be easily shown to be in \thegoodlogic. 
\end{proof}

% \begin{lemma}
% For each Parikh constraint $\chi$ there is a PMTA $\A$ with $\L(\A)=\L(\chi)$.
% \end{lemma}
\parikhtorec*

\begin{proof} 
We assume w.l.o.g.~that $\chi$ is of the form $\pr{c} + \sum_{i\in[r]} \pr{c_i}\,\cnts{X_i} \ttleqfin \pr{d} + \sum_{j\in[k]} \pr{d_j}\,\cnts{Y_j}$ 
where all $X_i$ are pairwise distinct, and all $Y_j$ likewise.
%where on both sides of $\ttleqfin$ duplicates of set variables are resolved by multiplication, respectively. 
Given a subset $\theta\subseteq\mathrm{free}(\chi)$, we denote by $|\theta|_X$ the number $\sum_{X_i\in\theta}c_i$ (and similar for $|\theta|_Y$). Then, assuming $\xi(\rho)=(\sigma^\xi_\rho,\theta^\xi_\rho)$, we get
\[\L(\chi)=\{\xi\in\To_{\Sigma_{\mathrm{free}(\chi)}}\mid  c+ \textstyle\sum_{{\rho\in\pos(\xi)}}|\theta^\xi_{\rho}|_X\leq d + \sum_{{\rho\in\pos(\xi)}} |\theta^\xi_{\rho}|_Y<\infty\}\]
and, by the condition $<\infty$, both sums have finite support.
 Thus, for each $\xi\in\L(\chi)$, there is a non-empty, finite, prefix-closed $Z\subset\{0,1\}^*$ that comprises all positions holding variable assignments and for which $\xi|_Z$ satisfies $\chi$; if $Z$ is a minimal such set, we denote $\xi|_Z$ in the following as $\xi_\chi$. This can be tested by a PMTA as follows.

Let $D=\{(i,j)\mid 0\leq i\leq \sum_{l\in[r]} c_l, 0\leq j\leq \sum_{l\in[k]} d_l\}$. We construct the PMTA $\A=(\{q_I,q_f\},\Xi,q_I,\Delta,\{\{q_f\}\},C)$ with $\Xi=(\Sigma_{\mathrm{free}(\chi)}\times D)\cup\Sigma_{\mathrm{free}(\chi)}$, $\Delta=\Delta_P\cup\Delta_\omega$ where
\begin{itemize}
    \item $\Delta_P=\{(q_I,\bigl((\sigma,\theta),(|\theta|_X,|\theta|_Y)\bigr),q,q)\mid (\sigma,\theta)\in\Sigma_{\mathrm{free}(\chi)}, q\in\{q_I,q_f\}\}$ and
    \item $\Delta_\omega=\{(q_f,(\sigma,\emptyset),q_f,q_f)\mid\sigma\in\Sigma\}$
\end{itemize}
and $C=\{(z_1,z_2)\mid \pr{c}+z_1\ttleqfin \pr{d}+z_2\}$.\footnote{Note that by \cref{lemma-semilin-presburger} we can use this description for a semilinear set.}

$\L(\chi)\subseteq\L(\A):$ Let $\xi\in\L(\chi)$. We prove that there exists a tree $\zeta\in\To_\Xi$ with $(\zeta)_{\Sigma_{\mathrm{free}(\chi)}}=\xi$ and an accepting run $\kappa$ of $\A$ on $\zeta$. Let $\zeta$ be such that, for each position $\rho$,
\[\zeta(\rho)=
\begin{cases}
\bigl((\sigma,\theta),(|\theta|_X,|\theta|_Y)\bigr) &\text{if } \rho\in\pos\{\xi_\chi\}, \xi(\rho)=(\sigma,\theta)\\
\xi(\rho) &\text{otherwise}
\end{cases}.\]
Clearly, $(\zeta)_{\Sigma_{\mathrm{free}(\chi)}}=\xi$ and $\pos(\zeta_\cnt)=\pos(\xi_\chi)$. Moreover, define $\kappa \in \To_{\{q_I,q_f\}}$ such that
\[
\kappa(\rho)=
\begin{cases}
q_I &\text{if } \rho\in\pos\{\xi_\chi\}\\
q_f &\text{otherwise}
\end{cases}\]
for each position $\rho$. Due to the construction of $\Delta$, $\kappa$ is compatible with the transitions of $\A$ in connection with $\zeta$. As $\varepsilon\in\pos\{\xi_\chi\}$ and, thus, $\kappa(\varepsilon)=q_I$, $\kappa$ is a run of $\A$ on $\zeta$.

It remains to argue that $\kappa$ is accepting. We can calculate, by assuming $\zeta(\rho)=(\sigma_\rho,\theta_\rho)$,
\[\Psi(\zeta_\cnt)=\;\sum_{\mathclap{\rho\in\pos(\zeta_\cnt)}}\;(\zeta_\cnt(\rho))_D=\;\sum_{\mathclap{\rho\in\pos(\zeta_\cnt)}}\;(|\theta_\rho|_X,|\theta_\rho|_Y)=(\;\sum_{\mathclap{\rho\in\pos(\zeta_\cnt)}}\;|\theta_\rho|_X,\;\sum_{\mathclap{\rho\in\pos(\zeta_\cnt)}}\;|\theta_\rho|_Y)\ .\]
As $\pos(\zeta_\cnt)=\pos(\xi_\chi)$ and we know from $\xi_\chi$ that $c+ \sum_{\rho\in\pos(\xi)}|\theta_{\rho}|_X\leq d + \sum_{\rho\in\pos(\xi)} |\theta_{\rho}|_Y$, we obtain $\Psi(\zeta_\cnt)\in C$. Finally, as $\xi_\chi$ is finite, we obtain for each path $\pi$ that $\inf(\kappa(\pi))=\{q_f\}$. Thus, $\kappa$ is accepting and $\xi\in\L(\A)$.

$\L(\A)\subseteq\L(\psi):$ Let $\xi\in\L(\A)$. As $\A$ recognizes all assignments of variables with transitions from $\Delta_P$, it is clear that each variable is counted finitely often in $\xi$. Moreover, $\Delta_P$ only contains transitions where the counter values fit to the variable occurrences of each label in $\xi$. Thus, as $\Psi(\zeta_\cnt)\in C$ for each $\zeta$ with $(\zeta)_{\Sigma_{\mathrm{free}(\chi)}}=\xi$ and $\Run_\A(\zeta)\neq\emptyset$, we obtain that $c+ \sum_{\rho\in\pos(\xi)}|\theta_{\rho}|_X\leq d + \sum_{\rho\in\pos(\xi)} |\theta_{\rho}|_Y$ also holds for $\xi_\chi$. Thus, $\xi\in\L(\psi)$.
\end{proof}

% \begin{lemma}
% For every \thegoodlogic formula $\varphi$ there is a PMTA $\A$ with $\L(\A)=\L(\varphi)$.
% \end{lemma}
\deftorec*

\begin{proof}
Let $\varphi$ be an \thegoodlogic formula. By \cref{thm-tree-normal-form}, we can assume that $\varphi$ is in tree normal form, i.e., of the form $\exists X_1.\cdots\exists X_n. \bigvee_{i=1}^k \big( \varphi_i \wedge \bigwedge_{j=1}^{l_i} \chi_{i,j} \big)$,
where $\varphi_i$ are plain MSO sentences and the $\chi_{i,j}$ are (unnegated) Parikh constraints.. Then we proceed by induction on the (now restricted) structure of $\varphi$. 

For the induction base, we consider the following cases:
If $\varphi$ is an MSO-formula, then, by~\cite{Rab69}, $\L(\varphi)$ can be recognized by an MTA. As MTA are an instance of PMTA, we obtain that $\L(\varphi)$ is PMTA-recognizable. Now let $\varphi$ be a Parikh constraint. By \cref{parikh-to-rec} we obtain that $\L(\varphi)$ is PMTA-recognizable.

Now assume that $\varphi$ is of the form $\varphi_1\land\varphi_2$.  Then $\L(\varphi)=\L_{\mathrm{free}(\varphi)}(\varphi_1)\cap\L_{\mathrm{free}(\varphi)}(\varphi_2)$ and, by induction hypothesis, $\L(\varphi_1)$ and $\L(\varphi_1)$ both are PMTA-recognizable. But then also $\L_{\mathrm{free}(\varphi)}(\varphi_1)$ and $\L_{\mathrm{free}(\varphi)}(\varphi_1)$ are PMTA-recognizable\footnote{It can be shown by a standard construction involving a relabeling that, for each $\V\supseteq\mathrm{free}(\varphi)$, $\L(\varphi)$ is recognizable iff $\L_\V(\varphi)$ is recognizable.}, respectively. Finally, by \cref{closure-pmta}, $\L_{\mathrm{free}(\varphi)}(\varphi_1)\cap\L_{\mathrm{free}(\varphi)}(\varphi_2)$ is PMTA-recognizable. With a similar argumentation and with help of \cref{closure-pmta} we obtain that $\varphi_1\lor\varphi_2$ is PMTA-recognizable.

Now assume that $\varphi$ is of the form $\exists X.\varphi_1$. In the usual way, we define the relabeling $\tau\colon\Sigma_{\free(\varphi_1)\cup\{X\}}\to\Pow(\Sigma_{\free(\varphi)})$ given by $\tau(\sigma,\theta)=\{(\sigma,\theta\cap\free(\varphi))\}$ and observe that $\L(\varphi)=\tau(\L_{\free(\varphi_1)\cup\{X\}}(\varphi_1))$. By induction hypothesis, using the same argumentation as above, and by \cref{closure-pmta}, we obtain that $\exists X.\varphi_1$ is PMTA-recognizable.
\end{proof}

\subsection{Parikh-Muller Automata on Words correspond to Reachability-Regular Parikh Automata}
\input{app-grobler-model}

\subsection{Emptiness}

% \begin{theorem}
% Given a PTA $\A$, deciding $\L(\A)\neq\emptyset$ is \textsc{NP}-complete.
% \end{theorem}
\ptaempt*

\begin{proof}
Let $\A=(Q,\Sigma\times D,\delta,q_I,F,C)$ be a PTA of dimension $s$ where all numbers are encoded in binary. W.l.o.g. assume that $D=\{d_1,\ldots,d_n\}$ with $d_i\in\Nat^s$. Moreover, let $\A'=(Q,\Sigma\times D,\delta,q_I,F)$ be a tree automaton\footnote{In our setting, a tree automaton $\A$ over $\Sigma\times D$ is defined to be a PTA with $C=\Nat^s$ which simplifies $\L(\A)$ to $\{\xi\in T_{\Sigma\times D}\mid \Run_\A(\xi)\neq \emptyset\}$.} over $\Sigma\times D$. It is easy to observe that $\L(\A)\neq\emptyset$ iff $\Psi(\L(\A'))\cap C\neq\emptyset$. Our aim is to show that $\Psi(\L(\A))$ can be encoded by a Presburger formula. 

Following and slightly adjusting \cite[proof of Lm. 8.15]{Kla04}, in a first step we construct a context-free grammar that generates all words obtained by ``flattening'' the trees in $\L(\A')$ and where the symbols from $\Sigma\times D$ are projected to their $D$-component. 
% We let $G=(Q,\Sigma\times D,q_i,R)$ where $R$ consists of the following rules:
% \begin{itemize}
%     \item if $(q,(\sigma,d),q_1,q_2)\in\delta$, then $q\to q_1(\sigma,d)q_2\in R$, and
%     \item if $(q,(\sigma,d),q_1,q_2)\in\delta$ and $(q_1,q_2)\in F\times F$, then $q\to(\sigma,d)\in R$.
% \end{itemize}
% We observe that $\Psi(\L(\A'))=\Psi(\L(G))$ (where we sum on the right-hand side over string positions).
We let $G'=(Q,D,q_i,R)$ where $R$ consists of the following rules:
\begin{itemize}
    \item if $(q,(\sigma,d),q_1,q_2)\in\delta$, then $q\to q_1 d q_2\in R$, and
    % \item if $(q,(\sigma,d),q_1,q_2)\in\delta$ and $(q_1,q_2)\in F\times F$, then $q\to d\in R$.
    \item if $q\in F$, then $q\to \varepsilon\in R$.
\end{itemize}
We observe that $\Psi(\L(\A'))=\Psi(\L(G'))$ (where we sum on the right-hand side over string positions and assume that $(d)_D=d$).

Now we proceed as in \cite[Proposition III.2.]{FigLib15a}: as the result used, to construct a Presburger formula for the language of an automaton, was stated for context-free grammars, it works in our more general setting as well.

By \cite[Theorem 4]{VerSeiSch05} (which had originally a small mistake fixed in \cite{HagLin12}) we can construct in linear time an existential Presburger formula $\varphi^P_{G'}(\nv{z}_1,\ldots,\nv{z}_n)$ encoding the Parikh image (counts of occurrences of each symbol) of $\L(G')$. Also due to \cite[proof of Theorem 4]{VerSeiSch05}, one can obtain from $\varphi^P_{G'}(\nv{z}_1,\ldots,\nv{z}_n)$ in polynomial time an existential\footnote{This was not explicitly stated in \cite[proof of Theorem 4]{VerSeiSch05}, however, their proof is based on this assumption.} Presburger formula $\varphi_{G'}(\nv{y}_1,\ldots,\nv{y}_s)$ representing $\Psi(\L(G'))$. Moreover, by \cite[Lemma II.1]{VerSeiSch05}, an existential Presburger formula $\varphi_C(\nv{y}_1,\ldots,\nv{y}_s)$ describing $C$ can be built in linear time. Finally, satisfiability of the Presburger sentence $\exists \nv{y}_1,\ldots,\nv{y}_s.\varphi_{G'}(\nv{y}_1,\ldots,\nv{y}_s)\land\varphi_C(\nv{y}_1,\ldots,\nv{y}_s)$ over $\mathbb{N}$ can be tested in \textsc{NP} \cite{Sca84}.
\textsc{NP}-hardness follows directly from the word case.
\end{proof}

% \begin{theorem}
% Given a PMTA $\A$, deciding $\L(\A)\neq\emptyset$ is \textsc{PSpace}-complete.
% \end{theorem}
\pmtaempt*

\begin{proof}
Let $\A=(\Q,\Xi,q_{I},\Delta,\F,C)$ be a PMTA with $\Q= Q_{P}\cup Q_{\omega}\cup\{q_I\}$, $\Xi=(\Sigma\times D)\cup\Sigma$, and $\Delta=\Delta_{P}\cup \Delta_{\omega}$. As each tree in the language of $\A$ can be decomposed into some finite tree over $\Sigma\times D$, on which the Parikh constraint is tested, and a number of infinite trees from $T_\Sigma$, we can reduce PMTA non-emptiness testing to deciding non-emptiness of Muller tree automata and Parikh tree automata.
To this end, consider 
\begin{itemize}
    \item the Muller tree automaton $\A_{q_I}=(Q_\omega\cup\{q_I\},\Sigma,q_I,\Delta_\omega,\F)$, 
    \item the Muller tree automata $\A_q=(Q_\omega,\Sigma,q,\Delta_\omega,\F)$ for all $q\in Q_\omega$, and
    \item the Parikh tree automaton $\A_P=(\Q,\Sigma\times D\!,q_I,\Delta_P,F_P,C)$ with $F_P=\{q{\,\in\,} Q_\omega \,|\, \L(\A_q){\,\neq\,}\emptyset\}$.
%    \item let $F_P=\{q\in Q_\omega\mid \L(\A_q)\neq\emptyset\}$, and
%    \item $\A_P=(\Q,\Sigma\times D,q_I,\Delta_P,F_P,C)$ be a Parikh tree automaton.
\end{itemize}
As deciding $\L(\A_q)\neq\emptyset$ is \textsc{PSpace}-complete \cite{Rab69,HunDaw05}, $\A_P$ can be constructed in \textsc{PSpace} and, by \cref{emptiness-pta}, its  non-emptiness can be decided in \textsc{NP}. Thus, the overall \textsc{PSpace} complexity follows from the observation that
$\L(\A)\neq\emptyset \quad \text{iff} \quad \L(\A_{q_I})\neq\emptyset\ \text{or}\ \L(\A_P)\neq\emptyset$, shown next.

$\Rightarrow:$ Assume that $\L(\A)\neq\emptyset$. Then there exist some $\xi\in\To_\Sigma$, $\zeta\in\To_\Xi$ with $(\zeta)_\Sigma=\xi$, and $\kappa\in\Run_\A(\zeta)$. We distinguish the following two cases:

Assume $(\zeta)_\Sigma=\zeta$. Then $\kappa(u)\in Q_\omega$ for each $u\in\pos(\zeta)\setminus\{\varepsilon\}$ and, hence, $\zeta$ is recognized by $\A$ only with transitions from $\Delta_\omega$. But then, $\kappa\in\Run_{\A_{q_I}}(\zeta)$ and, thus, $\L(\A_{q_I})\neq\emptyset$.

Now assume $(\zeta)_\Sigma\neq\zeta$. Then $\zeta$ is of the form $\zeta_\cnt[\zeta_1,\ldots,\zeta_m]$ for some $m\in\Nat$ and, similarly, $\kappa$ can be decomposed into $\Bar{\kappa}[\kappa_1,\ldots,\kappa_m]$ with $\pos(\Bar{\kappa})=\pos(\zeta_\cnt)$, $\Bar{\kappa}$ is labeled with states from $Q_P\cup\{q_I\}$, and $\kappa_i$ is labeled with states from $Q_\omega$ for each $i\in[m]$. As, for each $i\in[m]$ and each path $\pi$ reaching $\kappa_i$, we have that $\inf(\kappa(\pi))\in\F$, also $\inf(\kappa_i(\pi\setminus\Bar{\kappa}))\in\F$ with $\pi\setminus\Bar{\kappa}$ being the path obtained by cutting the initial part going through $\Bar{\kappa}$ from $\pi$. Thus, $\kappa_i\in\Run_{\A_{\kappa_i(\varepsilon)}}(\zeta_i)$ and $\kappa_i(\varepsilon)\in F_P$. Moreover, we obtain by construction of $\A_P$ that $\kappa'=\Bar{\kappa}[\kappa_1(\varepsilon),\ldots,\kappa_m(\varepsilon)]$ is a run of $\A_P$ on $\zeta_\cnt$. Finally, as $\Psi(\zeta_\cnt)\in C$, $\kappa'$ is accepting and, thus, $(\zeta_\cnt)_\Sigma\in\L(\A_P)$. Hence, $\L(\A_P)\neq\emptyset$.

$\Leftarrow:$ First assume that $\L(\A_{q_I})\neq\emptyset$. Then there exists some $\xi\in\To_\Sigma$ and $\kappa\in\Run_{\A_{q_I}}(\xi)$. But then, by construction of $\A_{q_I}$, we also have $\kappa\in\Run_\A(\xi)$. Thus, $\L(\A)\neq\emptyset$.

Now assume $\L(\A_P)\neq\emptyset$. Then there is some $\xi\in T_\Sigma$, $\zeta\in T_{\Sigma\times D}$ with $(\zeta)_\Sigma=\xi$, and run $\kappa\in\Run_{\A_P}(\zeta)$. By construction of $\A_P$ we can decompose $\kappa$ into $\Bar{\kappa}[q_1,\ldots,q_m]$ for some $m\in\Nat$ such that $\pos(\Bar{\kappa})=\pos(\zeta)$, $\Bar{\kappa}$ is labeled with states from $Q_P\cup\{q_I\}$, and $q_1,\ldots,q_m\in Q_\omega$. As $\kappa$ is accepting, $q_1,\ldots,q_m\in F_P$ and, hence, $\L(\A_{q_i})\neq\emptyset$ for each $i\in[m]$. Thus, there are trees $\xi_1,\ldots,\xi_m\in\To_\Sigma$ and runs $\kappa_i\in\Run_{\A_{q_i}}(\xi_i)$ starting in $q_i$. But then, $\Bar{\kappa}[\kappa_1,\ldots,\kappa_m]$ is also a run of $\A$ on $\zeta'=\zeta[\xi_1,\ldots,\xi_m]$. Obviously, $\zeta'_\cnt=\zeta$ and, thus, $\Psi(\zeta'_\cnt)\in C$. Finally, as $\inf(\kappa_i(\pi))\in \F$ for each path $\pi$ and this transfers to  $\kappa$, we obtain that $\kappa\in\Run_\A(\zeta')$. Hence, $\xi[\xi_1,\ldots,\xi_m]\in \L(\A)$ and $\L(\A)\neq\emptyset$.
\end{proof}

%% file: app-grobler-model.tex
The set of all infinite words over alphabet $\Sigma$ will be denoted by $\Sigma^\omega$; we use $\Sigma^*$ to denote the set of all finite words over $\Sigma$.

We will use the projections $\cdot_\Sigma$ and $\cdot_D$ as well as the extended Parikh map $\Psi$ with their obvious restrictions to word domains.

\begin{definition}[Parikh $\omega$-Word Automaton] 
Let $\Sigma$ be an alphabet, let $s \in \mathbb{N}\setminus\{0\}$, let $D\subseteq\Nat^s$ be finite, and denote $(\Sigma\times D)\cup\Sigma$ by $\Xi$. A Parikh Word Automaton (of \emph{dimension} $s$) is a tuple $\A=(\Q,\Xi,q_I,\Delta,\F,C)$ where $\Q=Q_P\cup Q_\omega\cup\{q_I\}$ is a finite set of \emph{states} with $Q_P$,$Q_\omega$ disjoint and $q_I$ being the \emph{initial state}, $\Delta=\Delta_P\cup\Delta_\omega$ is the \emph{transition relation} with
 \[\Delta_P\subseteq (Q_P\cup\{q_I\})\times(\Sigma\times D)\times \Q \quad\text{ and }\quad \Delta_\omega\subseteq (Q_\omega\cup\{q_I\})\times\Sigma\times Q_\omega,\]
 $\F$ is the \emph{acceptance condition} given by either
 \begin{itemize}
     \item a set of \emph{final state sets} $\F\subseteq 2^{Q_\omega}$ (Muller acceptance) or
     \item a set of \emph{final states} $\F\subseteq Q_\omega$ (Büchi acceptance),
 \end{itemize}
  and $C\subseteq\Nat^s$ is  a semilinear set named \emph{final constraint}.
\end{definition}

A Parikh $\omega$-word automaton that uses a Muller acceptance condition will be called a \emph{Parikh-Muller word automaton} (PMWA); if it instead uses a Büchi acceptance condition it will be referred to as \emph{Parikh-Büchi word automaton} (PBWA). We note that, by choosing $\Delta_P=\emptyset$, we reobtain the concept of a Muller (respectively Büchi) word automaton.

\begin{definition}[Semantics of PMWA and PBWA]
A \emph{run} of $\A$ on a word $w\in \Xi^\omega$ is a word $\kappa_w\in Q^\omega$ whose initial position carries
%is labeled with 
$q_I$ 
and which respects $\Delta$ jointly with $w$. By definition of $\Delta$, if a run exists, then $w$ is of the form $w_\cnt w'$ with $w_\cnt\in(\Sigma\times D)^*$ and $w'\in\Sigma^\omega$.
%%
%We say that the run 
A run $\kappa_w$ is \emph{accepting} if
 \begin{enumerate}
    \item either $\inf(\kappa_w)\cap \F\neq\emptyset$ (Büchi acceptence) or $\inf(\kappa_w)\in \F$ (for Muller acceptance) and
     \item if $w_\cnt\neq\varepsilon$, then $\Psi(w_\cnt)\in C$.
     
 \end{enumerate}
Note that, by the first condition, $\kappa_w$ being accepting implies finiteness of $w_\cnt$ and, thus, well-definedness of the sum in $\Psi(w_\cnt)$.
 The set of all accepting runs of $\A$ on $w$ will be denoted by $\Run_\A(w)$. Then, the \emph{word language of} $\A$, denoted by $\L(\A)$, is the set \[\L(\A)=\{u\in {\Sigma}^\omega\mid \exists w\in \Xi^\omega \text{ with } \Run_\A(w)\neq\emptyset \text{ and } w_\Sigma=u\}\,.\]
\end{definition}

Now we recall the model of reachability-regular Parikh automata from Grobler et al. \cite{GroSabSie23a}.

\begin{definition}[Reachability-Regular Parikh Automata]
    Let $\Sigma$ be an alphabet and $s\in\Nat\setminus\{0\}$. A reachability-regular Parikh automaton (RRPA) of dimension $s$ is a tuple $\A=(Q,\Sigma,q_0,\Delta,F,C)$ where $Q$ is a finite set of states, $q_0\in Q$ is the initial state, $F\subseteq Q$ is the set of final states, $\Delta\subseteq Q\times\Sigma\times\Nat^s\times Q$ is a finite set of transitions and $C\subseteq\Nat^s$ is  a semilinear set.
\end{definition}

\begin{definition}[Semantics of RRPA]
    A run of $\A$ on an infinite word $w=w_1 w_2 w_3...$ is an infinite sequence $\tau=r_1 r_2 r_3...$ of transitions $r_i=(p_{i-1},w_i,d_i,p_i)$ such that $p_0=q_0$. We say that $\tau$ is accepting if there is an $i\geq 1$ such that $p_i\in F$ and $d_1+\ldots +d_i\in C$, and there are infinitely many $j\geq 1$ such that $p_j\in F$. The set of all accepting runs of $\A$ on $w$ will be denoted by $\Run_\A(w)$. Then, the \emph{word language of} $\A$, denoted by $\L(\A)$, is the set \[\L(\A)=\{w\in {\Sigma}^\omega\mid \Run_\A(w)\neq\emptyset\}\,.\]
\end{definition}

The goal is to show the following theorem:

\begin{theorem}
    Let $\L\subseteq\Sigma^\omega$. Then the following are equivalent:
    \begin{enumerate}
        \item $\L=\L(\A)$ for some PMWA $\A$,
        \item $\L=\L(\A)$ for some PBWA $\A$, and
        \item $\L=\L(\A)$ for some RRPA $\A$.
    \end{enumerate}
\end{theorem}

\begin{proof}
    The theorem immediately follows from \cref{rrpa_to_PBWA}, \cref{PBWA_to_rrba}, and \cref{PBWA_and_PMWA}.
\end{proof}

\subsubsection*{From RRPA to PBWA}

\begin{lemma}\label{rrpa_to_PBWA}
  Given an RRPA $\A$, we can construct a PBWA $\A'$ such that $\L(\A)=\L(\A')$.  
\end{lemma}
\begin{proof}
    Let $\A=(Q,\Sigma,q_0,\Delta,F,C)$ be an RRPA. We construct a PBWA $\A'$ with the following intuition. We duplicate the state space of $\A$ to obtain $Q_P$ and $Q_\omega$. Everytime $\A'$ reaches a final state from $\A$ it can nondeterministically choose if it stays in $Q_P$ and proceeds with counting or if it goes into $Q_\omega$ and reads the remaining word without counting.

    Formally, we let $D=\{d\in\Nat^s\mid (q,w,d,q')\in\Delta\}$ and define $\A'=(Q',(\Sigma\times D)\cup\Sigma,q_I,\Delta_P\cup\Delta_\omega,F',C)$ where 
    \begin{itemize}
        \item $Q'=Q_P\cup Q_\omega\cup\{q_I\}$ with
        \begin{itemize}
            \item $Q_P=\{q^P\mid q\in Q\}$ and
            \item $Q_\omega=\{q^\omega\mid q\in Q\}$,
        \end{itemize}
        \item $\Delta_P=\Delta_I\cup\Delta'$ with
        \begin{itemize}
            \item $\Delta_I=\{(q_I,(\sigma,d),q^P)\mid (q_0,\sigma,d,q)\in\Delta\}\cup\{(q_I,(\sigma,d),q^\omega)\mid (q_0,\sigma,d,q)\in\Delta,q\in F\}$, and
            \item $\Delta'=\{(q^P,(\sigma,d),p^P)\mid(q,\sigma,d,p)\in\Delta\}\cup\{(q^P,(\sigma,d),p^\omega)\mid(q,\sigma,d,p)\in\Delta,p\in F\}$,
        \end{itemize}
        \item $\Delta_\omega=\{(q^\omega,\sigma,p^\omega)\mid(q,\sigma,d,p)\in\Delta\}$,
        \item $F'=\{q^\omega\mid q\in F\}$.
    \end{itemize}
    
    $\L(\A)\subseteq\L(\A')$: Let $\tau=r_1 r_2 r_3\ldots$ with $r_i=(p_{i-1},w_i,d_i,p_i)$ be an accepting run of $\A$ on $w=w_1 w_2\ldots$. Then there is a position $i\geq 1$ such that $p_i\in F$ and $d_1+\ldots+d_n\in C$. Now let $\kappa=q_I\, p_1^P\dots p_{i-1}^P\, p_i^\omega\, p_{i+1}^\omega\ldots$. By construction, $\kappa$ is a run of $\A'$ on $w'=(w_1,d_1)\ldots(w_i,d_i)w_{i+1}w_{i+2}\ldots$ . Moreover, as $\Psi((w_1,d_1)\ldots(w_i,d_i))=d_1+\ldots+d_i\in C$ and $p_j^\omega\in F$ iff $p_j\in F$, we obtain that $\kappa$ is accepting. Finally, as $w'_\Sigma=w$, $w\in\L(\A')$.

    $\L(\A')\subseteq\L(\A)$: Similar argumentation as above, ``in reverse''.
\end{proof}

\subsubsection*{From PBWA to RRPA}

\begin{lemma}\label{PBWA_to_rrba}
  Given a PBWA $\A$, we can construct an RRPA $\A'$ such that $\L(\A)=\L(\A')$.  
\end{lemma}

\begin{proof}
    Let $\A=(\Q,\Xi,q_I,\Delta,\F,C)$ be a PBWA of dimension $s$. We construct a RRBA $\A'$ where we enrich the set of final states by copys of those states that are visited by $\A$ when entering $Q_\omega$. In order to keep apart runs that use transitions from $\Delta_P$ and runs which directly start in $\Delta_\omega$, we use one additional counter. The transitions of $\Delta_\omega$ are simulated by $\A'$ by using the zero-vector $\Vec{0}=(0)^{s+1}$. Formally, let $\A'=(Q',\Sigma,q_I,\Delta',F',C')$ where
    \begin{itemize}
        \item $Q'=Q\cup\{\hat{q}\mid q\in Q_\omega\}$,
        \item $F'=\F\cup\{\hat{q}\mid q\in Q_\omega\}$,
        \item $\Delta'=\Delta_P'\cup\Delta_\omega'$ with
        \begin{itemize}
            \item $\Delta_P'=\{(q,\sigma,d\cdot(0),q')\mid(q,(\sigma,d),q')\in\Delta_P,q'\notin Q_\omega\}\cup\\ \{(q,\sigma,d\cdot(1),\hat{q'})\mid(q,(\sigma,d),q')\in\Delta_P,q'\in Q_\omega\}$ and
            \item $\Delta_\omega'=\{(\hat{q},\sigma,\Vec{0},q')\mid (q,\sigma,q')\in \Delta_\omega\}\cup\{(q,\sigma,\Vec{0},q')\mid (q,\sigma,q')\in \Delta_\omega\}$, and
        \end{itemize}
        \item $C'=C\cdot(1)\cup C_0$ with $C_0=\{\Vec{0}\}$ if $(q_I,\sigma,q)\in \Delta_\omega$ for some $\sigma,q$ and $C_0=\emptyset$ otherwise.
    \end{itemize}

    $\L(\A)\subseteq\L(\A')$: Let $w\in\L(\A)$. Then there is some $u\in\Xi^\omega$ with $(u)_\Sigma=w$ and there is some accepting run $\kappa$ of $\A$ on $u$. Assume that $\kappa=q_I q_1 q_2\ldots$ .

    If $u_\cnt\neq \emptyset$, there is some $i\geq 1$ such that $q_j\in Q_\omega$ for each $j\geq i$ and $q_l\in Q_P$ for each $1\leq l<i$. Moreover, $u$ is of the form $(\sigma_1,d_1)\ldots(\sigma_i,d_i)\sigma_{i+1}\sigma_{i+2}\ldots$ . Let $\tau=r_1 r_2\ldots$ such that for each $k\geq 1$
    \[r_k=\begin{cases}(q_I,\sigma_1,d_1\cdot(0),q_1) &\text{if } k=1,i\neq 1\\
    (q_I,\sigma_1,d_1\cdot(1),\hat{q_1}) &\text{if } k=i=1\\
    (q_{k-1},\sigma_k,d_k\cdot(0),q_k) &\text{if } 1<k< i\\
    (q_{i-1},\sigma_i,d_i\cdot(1),\hat{q_i}) &\text{if } k=i, i>1\\
    (\hat{q_{i}},\sigma_{i+1},\vec{0},q_{i+1}) &\text{if } k= i+1\\
    (q_{k-1},\sigma_k,\vec{0},q_k) &\text{if } k>i+1\,.
    \end{cases}\]
    By construction, each $r_k$ is in $\Delta'$. It is not hard to see that $\tau$ is an accepting run of $\A'$ on $w$: as $d_1\cdot (0)+\ldots +d_i\cdot(1)=\Psi(u_\cnt)\cdot(1)$, $\Psi(u_\cnt)\cdot(1)\in C\cdot(1)$, and $\hat{q_i}\in F'$, the first condition is satisfied. Moreover, as in $r_{i+2}r_{i+3}\ldots$ the same states occur infinitely often as in the sequence $q_{i+2}q_{i+3}\ldots$, by construction of $F'$ the Büchi condition is satisfied as well. 

    If, on the other hand, $u_\cnt= \emptyset$, then $q_i\in Q_\omega$ for each $i\geq 1$ and $u=w=\sigma_1\sigma_2\ldots$ for some $\sigma_i\in\Sigma$. Then construct $\tau=r_1 r_2 \ldots$ with $r_1=(q_I,\sigma_1,\vec{0},q_1)$ and $r_j=(q_{j-1},\sigma_j,\vec{0},q_j)$ for each $j>1$. By construction, $r_i\in\Delta'$ for each $i\geq 1$. Moreover, as $\tau$ uses the same states infinitely often as $\kappa$, by construction of $F'$ the Büchi acceptance condition is satisfied. Finally, as the extended Parikh image is $\vec{0}$ at each position of $\tau$ and, by construction, $\vec{0}\in C_0$, $\tau$ is an accepting run of $\A'$ on $w$.

     $\L(\A')\subseteq\L(\A)$: As the extended Parikh image of $\A'$ does not change anymore after entering the states of $Q_\omega$, we can, as above, construct from each accepting run $\tau$ of $\A'$ a corresponding accepting run $\kappa$ of $\A$.
\end{proof}

\subsubsection*{From PBWA to PMWA and back}

By using standard constructions to transform a Büchi automaton into a Muller automaton and vice versa, we obtain the expressive equivalence of PBWA and PMWA. As any Büchi acceptence condition can be immediately expressed by a Muller acceptance condition, the first direction is straight forward.

\begin{lemma}
    For each PBWA $\A$ there is a PMWA $\A'$ such that $\L(\A)=\L(\A')$.
\end{lemma}

\begin{proof}
    Let $\A=(\Q,\Xi,q_I,\Delta,\F,C)$ be a PBWA. In the usual way (cf. \cite[Prop. 5.3]{thomasLanguagesAutomataLogic1997}), we construct the PMWA $\A'=(\Q,\Xi,q_I,\Delta,\F',C)$ where $\F'=\{F'\subseteq Q_\omega\mid F'\cap\F\neq\emptyset\}$. Then $\L(\A)=\L(\A')$.
\end{proof}

For the other direction, we decompose a given PMWA into a number of Parikh (word) automata and Muller (word) automata (compare with the proof of \cref{pmtaempt}), transform the Muller automata into equivalent Büchi automata (which well-known to be possible in the word case), and finally reverse our decomposition, obtaining an equivalent PBWA.

\begin{lemma}
    For each PMWA $\A$ there is a PBWA $\A'$ such that $\L(\A)=\L(\A')$.
\end{lemma}

\begin{proof}
    Let $\A=(\Q,\Xi,q_I,\Delta,\F,C)$ with $Q=Q_P\cup Q_\omega\cup\{q_I\}$ and $\Delta=\Delta_P\cup\Delta_\omega$ be a PMWA. Let $Q_{\omega,r}=\{q\in Q_\omega\mid (p,(\sigma,d),q)\in\Delta_P\text{ for some } p\in Q_P,(\sigma,d)\in\Sigma\times D\}$. For each $q\in Q_{\omega,r}$ we construct a pair $(\A^P_q,\A^M_q)$ consisting of a Parikh word automaton\footnote{For a definition of Parikh word automata (over finite words), we refer the reader to \cite[page 4]{GuhJecLeh22}.} $\A^P_q$ and a Muller word automaton $\A^M_q$ as follows:
    \begin{itemize}
        \item We let $\A^P_q=(Q_P\cup\{q_I,q\},\Sigma\times D,q_I,\Delta_q,\{q\},C)$ with $\Delta_q=\{(p,(\sigma,d),p')\in\Delta_P\mid p'\notin Q_\omega \text{ or }p'=q\}$.
        \item We let $\A^M_q=(Q_\omega,\Sigma,q,\Delta_\omega,\F)$.
    \end{itemize}
    Moreover, we construct the Muller word automaton $\A^M_{q_I}=(\Q_\omega\cup\{q_I\},\Sigma,q_I,\Delta_\omega,\F)$.

    \begin{claim}\label{claim1}
        $\L(\A)=(\bigcup_{q\in Q_{\omega,r}}\L(\A^P_q)\cdot\L(\A^M_q))\cup\L(\A^M_{q_I})$
    \end{claim}

    \begin{proof}
        $\subseteq$: Let $w\in\L(\A)$. Then there is some $u\in\Xi^\omega$ with $(u)_\Sigma=w$ and $\kappa\in\Run_\A(u)$. If $u_\cnt=\emptyset$, by construction, $\kappa\in\Run_{\A^M_{q_I}}(u)$. If $u_\cnt\neq\emptyset$, $\kappa$ is of the form $q_I q_1 q_2\ldots$ and there is some $i\geq 1$ such that $q_j\in Q_\omega$ for each $j\geq i$ and $q_l\in Q_P$ for each $1\leq l<i$. Moreover, $u$ is of the form $(\sigma_1,d_1)\ldots(\sigma_i,d_i)\sigma_{i+1}\sigma_{i+2}\ldots$ . By construction of $\A^P_{q_i}$, $\kappa'=q_I\ldots q_i$ is a run of $\A^P_{q_i}$ on $(\sigma_1,d_1)\ldots(\sigma_i,d_i)$ and, as $d_1+\ldots+d_i$ is in $C$, $\kappa'$ is accepting. Thus, $\sigma_1\ldots\sigma_i\in\L(\A^P_{q_i})$. Similarly, by construction of $\A^M_{q_i}$, $q_i q_{i+1}\ldots$ is an accepting run of $\A^M_{q_i}$ on $\sigma_{i+1}\sigma_{i+2}\ldots$. Thus, $w\in \L(\A^P_{q_i})\cdot\L(\A^M_{q_i})$.

        $\supseteq$: If $w\in\L(\A^M_{q_I})$ we directly obtain that $w\in\L(\A)$ as the transitions of $\A^M_{q_I}$ are a subset of the transitions of $\A$. Now let $w\in\L(\A^P_q)\cdot\L(\A^M_q)$ for some $q\in Q_{\omega,r}$, i.e., $w=w_1w_2$ for some $w_1\in\L(\A^P_q)$ and $w_2\in\L(\A^M_q)$. As each run $\kappa_1$ of $\A^P_q$ on some $u_1$ with $(u_1)_\Sigma=w_1$ ends in the state $q$ and each run $\kappa_2$ of $\A^M_q$ on $w_2$ starts in $q$, we can easily combine $\kappa_1$ and $\kappa_2$ to a run $\kappa$ of $\A$ on $u_1w_2$.
    \end{proof}

    Now let $q\in Q_{\omega,r}\cup\{q_I\}$. Let $\A^B_q=(Q_{B,q},\Sigma,p_{B,q},\Delta_{B,q},\F_{B,q})$ be the Büchi word automaton equivalent to $\A^M_q$,  obtained by the standard construction from Muller to Büchi word automata \cite[Prop. 5.3]{thomasLanguagesAutomataLogic1997}. Clearly, $\L(\A^P_q)\cdot\L(\A^M_q)=\L(\A^P_q)\cdot\L(\A^B_q)$ for each $q\in Q_{\omega,r}$.

    In the next step, for each $q\in Q_{\omega,r}$ we compose $(\A^P_q,\A^B_q)$ into one PBWA $\hat{\A}_q$. W.l.o.g. we assume that $(Q_P\cup\{q_I\})\cap Q_{B,q}=\emptyset$ and we let $\hat{\A}_q=(Q_P\cup Q_{B,q}\cup\{q_I\},\Xi,q_I,\Delta_{PB},\F_{B,q},C)$ where $\Delta_{PB}=\{(p,(\sigma,d),p')\in\Delta_P\mid p'\notin Q_\omega\}\cup\{(p,(\sigma,d),p_{B,q})\mid (p,(\sigma,d),p')\in\Delta_P,p'=q\}\cup \Delta_{B,q}$. 

    \begin{claim}\label{claim2}
        $\L(\A^P_q)\cdot\L(\A^B_q)=\L(\hat{\A}_q)$
    \end{claim}

    \begin{proof}
        $\subseteq$: Let $w\in\L(\A^P_q)\cdot\L(\A^B_q)$. Then $w=w_1w_2$ with $w_1\in\L(\A^P_q)$ and $w_2\in\L(\A^B_q)$, there is some $u\in(\Sigma\times D)^*$ with $(u)_\Sigma=w_1$ and $\kappa_1\in\Run_{\A^P_q}(u)$, and there is some $\kappa_2\in\Run_{\A^B_q}(w_2)$. By inspecting the construction of $\A^P_q$ and $\A^B_q$, we obtain that $\kappa_1$ is of the form $q_Iq_1\ldots q_i$ with $q_i=q$ and $\kappa_2$ is of the form $p_1 p_2\ldots$ with $p_1=p_{B,q}$. Then we construct the run $\kappa=q_Iq_1\ldots q_{i-1}p_1p_2\ldots$ which, by definition of $\hat{\A}_q$, respects $\Delta_{PB}$ jointly with $uw_2$. Moreover, as $\kappa$ inherits the extended Parikh image from $\kappa_1$ and the states occurring infinitely often from $\kappa_2$, $\kappa$ is accepting.

%\color{purple} (OLD)        $\supseteq$: Let $w\in \L(\hat{\A}_q)$. As above, we can decompose $w$ into $w_1w_2$ and each run $\kappa\in\Run_{\hat{\A}_q}(uw_2)$ for some $u\in(\Sigma\times D)^*$ with $(u)_\Sigma=w_1$ into runs $\kappa_1\in\Run_{\A^P_q}(u)$ and  $\kappa_2\in\Run_{\A^B_q}(w_2)$ by keeping the respective extended Parikh image and the states occurring infinitely often.

%\color{blue} (NEW)       
$\supseteq$: Let $w\in \L(\hat{\A}_q)$. Then, similar to above, there must be a decomposition $w = w_1w_2$ and a run $\kappa\in\Run_{\hat{\A}_q}(uw_2)$ for some $u\in(\Sigma\times D)^*$ with $(u)_\Sigma=w_1$. This run can be easily decomposed into a run $\kappa_1\in\Run_{\A^P_q}(u)$ (with the same extended Parikh image), and a run $\kappa_2\in\Run_{\A^B_q}(w_2)$ (with the same states occurring infinitely often), witnessing $w_1 \in \L(\A^P_q)$ as well as $w_2 \in \L(\A^B_q)$.
    \end{proof}

    Moreover, we construct the PBWA $\hat{\A}_{q_I}=(\Q_{B,q_I}\cup\{q_0\},\Sigma,q_0,\Delta',\F_{B,q},\emptyset)$ where $\Delta'=\Delta_{B,q}\cup\{(q_0,\sigma,q)\mid (p_{B,q},\sigma,q)\in\Delta_{B,q}\}$.

    \begin{claim}\label{claim3}
        $\L(\A^B_{q_I})=\L(\hat{\A}_{q_I})$
    \end{claim}

    \begin{proof}
        As Büchi word automata are a special case of PBWA, it is not hard to see that $\L(\A^B_{q_I})=\L(\hat{\A}_{q_I})$. As PBWA do not allow to use their initial state more than once, we introduce the new initial state $q_0$. 
    \end{proof}

    By taking together \cref{claim1}, \cref{claim2}, and \cref{claim3}, we finally observe that $\L(\A)=\bigcup_{q\in Q_{\omega,r}\cup\{q_I\}}\L(\hat{\A}_q)$. As PBWA are closed under union (which follows either from adapting the proof of \cref{closure-pmta}, or from the closure of RRPA under union \cite[Lm. 2]{GroSabSie23a} together with the expressive equivalence of RRPA and PBWA we already established via \cref{rrpa_to_PBWA} and \cref{PBWA_to_rrba}), we obtain that the language $\bigcup_{q\in Q_{\omega,r}\cup\{q_I\}}\L(\hat{\A}_q)$ is recognizable by some PBWA.
\end{proof}

\begin{theorem}\label{PBWA_and_PMWA}
    PBWA and PMWA are expressively equivalent.
\end{theorem}

%% file: app-interpretations.tex
\section{Details on MSO-Interpretations}

\msointerlemma*

\begin{proof}
The proof uses a variant of the well-known construction for plain MSO.
W.l.o.g. we assume that the \thegoodlogic sentence $\varphi$ is in GNF (cf. \Cref{def:GNF}), i.e. of the shape
$$
\textstyle\exists X_1.\cdots\exists X_n. \bigvee_{i=1}^k \big( \varphi_i \wedge \bigwedge_{j=1}^{l_i} \chi_{i,j} \big), 
$$
where the $\varphi_i$ are CMSO formulae, whereas the $\chi_{i,j}$ are (unnegated) Parikh constraints.
We then define $\varphi^\mathcal{I}$ to be the \thegoodlogic formula
$$
\bigwedge_{\pr{a}\in \SigC}\big( \forall z.\varphi_{\pr{a}}(z)\Leftrightarrow z=\pr{a} \big) \wedge \exists Z_\mathrm{Dom}.\Big(\big(\forall z.Z_\mathrm{Dom}(z) \Leftrightarrow \varphi_\mathrm{Dom}(z)\big) \wedge \varphi' \Big), 
$$
$$
\text{with}\qquad\qquad \varphi' \coloneqq \exists X_1 \subseteq Z_\mathrm{Dom}.\cdots\exists X_n \subseteq Z_\mathrm{Dom}. \bigvee_{i=1}^k \big( \varphi'_i \wedge \bigwedge_{j=1}^{l_i} \chi_{i,j} \big)
$$
where 
$\varphi'_i$ is obtained from $\varphi_i$ by (inside out) replacing
\begin{itemize}
    \item all $\pr{Q}(\iota_1,\ldots,\iota_m)$ by $\varphi_\pr{Q}(\iota_1,\ldots,\iota_m)$
    \item all $\forall x.\psi$ by $\forall x\in Z_\mathrm{Dom}.\psi$ as well as all $\exists x.\psi$ by $\exists x\in Z_\mathrm{Dom}.\psi$ and 
    \item all $\forall X.\psi$ by $\forall X\subseteq Z_\mathrm{Dom}.\psi$ as well as all $\exists X.\psi$ by $\exists X\subseteq Z_\mathrm{Dom}.\psi$  
\end{itemize}

While $\varphi^\mathcal{I}$ is not in GNF, it is easy to check that it is indeed in \thegoodlogic (as all set variables occurring in all $\chi_{i,j}$ are assertive).
Moreover, an easy induction shows $ \mathfrak{A} \models \varphi^\mathcal{I} \Longleftrightarrow \mathfrak{B} \models \varphi$ (assuming such an appropriate $\mathfrak{B}$ exists). 
\end{proof}

We note that this proof gives no garantees for $\varphi^\mathcal{I}$ in case $\mathcal{I}(\mathfrak{A})$ is undefined.
We also want to point out that in the literature, there are subtle differences as to how MSO-interpretations are defined.
Sometimes, an MSO-interpretation is also stipulated to feature an extra MSO formula $\varphi_=(x,y)$ with two free variables, which imposes an equality predicate (and thus a factorization on the elements of $\mathfrak{A}$, whose equivalence classes become the elements of $\mathfrak{B}$). Under such circumstances, however, finiteness and modulo atoms cannot be faithfully translated any more. This justifies our more conservative choice of MSO-interpretation.

\msointerthm*

\begin{proof}
Given an MSO-interpretation $\mathcal{I}$, we make use of the auxiliary formula 
$$
\varphi^\mathrm{def}_{\mathcal{I}} \coloneqq \exists z.\varphi_\mathrm{Dom}(z) \wedge \bigwedge_{\pr{a}\in \SigC}\big( \exists z.\varphi_{\pr{a}}(z) \wedge \forall z'\forall z''.\varphi_{\pr{a}}(z') \wedge \varphi_{\pr{a}}(z'')\Rightarrow z'=z'' \big).
$$
It is easy to see that $\varphi^\mathrm{def}_{\mathcal{I}}$ is satisfied by some $\Sig'$-structure $\mathfrak{A}$ exactly if $\mathcal{I}(\mathfrak{A})$ is defined.\footnote{We recall that we assume a model theory that requires nonempty domains, but note that this choice does not substantially affect our results.} Also, $\varphi^\mathrm{def}_{\mathcal{I}}$ is in plain MSO and hence in \thegoodlogic.

Let now $\varphi$ be an \thegoodlogic formula whose satisfiability over $\mathcal{I}(\mathscr{S})$ we want to check. We then obtain, with the help of \Cref{lem:msointer}, that $\mathcal{I}(\mathscr{S})$ contains a model $\mathfrak{B}$ of $\varphi$ if and only if
there is an $\mathfrak{A} \in \mathscr{S}$ with $\mathfrak{B} = \mathcal{I}(\mathfrak{A})$ and $\mathfrak{A} \models \varphi^\mathcal{I}$. This, in turn, is equivalent to the existence of some $\mathfrak{A} \in \mathscr{S}$ with $\mathfrak{A} \models \varphi^\mathrm{def}_{\mathcal{I}} \wedge \varphi^\mathcal{I}$.
Therefore, satisfiability of $\varphi$ over $\mathcal{I}(\mathscr{S})$ coincides with satisfiability of $\varphi^\mathrm{def}_{\mathcal{I}} \wedge \varphi^\mathcal{I}$ over $\mathscr{S}$. As $\varphi^\mathrm{def}_{\mathcal{I}} \wedge \varphi^\mathcal{I}$ is in \thegoodlogic, this is a decidable problem by assumption, which concludes the proof of our claim.

The particular case for tree-interpretable classes of structures is now a straightforward consequence of \Cref{cor:treedecidable}.
\end{proof}

%% file: app-C2.tex
\section{Coupling \thegoodlogicheading with FO$_\mathbf{Pres}^\mathbf{2}$}
%\todo{choose between phi and varphi}

%The logic of $2$-variable first-order logic with Presburger counting, is
%an extension of $FO^2$, with counting over ultimately periodic sets of integers: \todo{find the macro for FO}
\begin{definition}[Ultimately periodic sets]
    A \emph{(one-dimensional) linear} set is a set of the form $\qty{a +ip \mid i \in \mathbb{N}}$ for fixed natural numbers $a$ and $p$.
    We also consider $\emptyset$ and $\qty{\infty}$ to be linear.
    
    An \emph{ultimately periodic} set is a finite union of linear sets.
\end{definition} 
%\todo{lh: linear set was already used in line 267 for something else ;) -- good to know... ;-)}

\begin{definition}[$\FOpres$]
    The logic $\FOpres$
    consists of formulae written
    using only $2$ variables ($x$ and $y$), atoms 
    (over a signature $\Sig$ containing only unary and binary predicates),
    equality, conjunction,
    disjunction, negation and ultimately periodic 
    existential quantification. 
    For $S$ an ultimately periodic set, the ultimately periodic quantifier
    $\exists^S x$ has the following semantics: 
    for $\varphi$ a formula in $\FOpres$,
    the sentence $\exists^S x. \varphi$ is true for a model $\mathfrak{A}$ and variable assignment $\nu$ if and only
    if the cardinal of the set 
    $\qty{a \in A \mid \mathfrak{A},\nu_{x \mapsto a} \models \varphi}$
    is in $S$. 
\end{definition}

This very expressive logic has recently been shown to be decidable \cite{Ben20}. In fact,
the satisfiability problem for $\FOpres$ is reduced to the satisfiability of a 
formula in existential Presburger Arithmetic. 
Given that $\thegoodlogic$ can express existential Presburger
Arithmetics, and that the translation from $\FOpres$ to 
Presburger Arithmetic is transparent enough to partially 
account for the $\FOpres$ models, we can pair them both to obtain a handier
result. For the following, recall that, for a $\Sig$-structure $\mathfrak{C}$ and a signature $\Sig' \subseteq \Sig$,
$\mathfrak{C}|_{\Sig'}$ denotes the $\Sig'$-reduct of $\mathfrak{C}$, i.e., the
structure obtained by forgetting from $\mathfrak{C}$ the predicates that are not in $\Sig'$.

%\begin{lemma}\label{lemmaC2}
%Let $\psi$ be a sentence in $\FOpres$, over a signature
%$\Sig_1 \cup \Sig_2$, where $\Sig_1$ contains only unary predicates and $\Sig_2$ only
%binary predicates. 
%There is an unary extension $\Sig_1' \supseteq \Sig_1$ and 
%a sentence
%$\psi' \in \thegoodlogic$ over the signature $\Sig_1'$ such that
%\begin{itemize}
%    \item for every structure $\mathfrak{B} \models \psi$, there is a structure $\mathfrak{A} \models \psi'$ such that %$\mathfrak{A}|_{\Sig_1} = \mathfrak{B}|_{\Sig_1}$ and
%    \item for every structure $\mathfrak{A} \models \psi'$, there is a structure $\mathfrak{B} \models \psi$ such that %$\mathfrak{A}|_{\Sig_1} = \mathfrak{B}|_{\Sig_1}$,
%\end{itemize}
%where, for a $\Sig$-structure $\mathfrak{C}$ and a signature $\Sig' \subseteq \Sig$,
%$\mathfrak{C}|_{\Sig'}$ denotes the $\Sig'$-reduct of $\mathfrak{C}$, i.e., the
%structure obtained by forgetting from 
%$\mathfrak{C}$ the predicates that are not in $\Sig'$.
%\end{lemma}

\begin{lemma}\label{lemmaC2}
Let $\psi$ be a sentence in $\FOpres$, over a signature
$\Sig$ of unary and binary predicates and let $\Sig_1 \subseteq \Sig$ be a subset of the unary predicates of $\Sig$.
Then there is a sentence $\psi_{\Sig_1} \in \thegoodlogic$ over the signature $\Sig_1$ such that
\begin{itemize}
    \item every structure $\mathfrak{B} \models \psi$ satisfies $\mathfrak{B} \models \psi_{\Sig_1}$ and
    \item for every structure $\mathfrak{A} \models \psi_{\Sig_1}$, there is a structure $\mathfrak{B} \models \psi$ such that $\mathfrak{A}|_{\Sig_1} = \mathfrak{B}|_{\Sig_1}$.
\end{itemize}
\end{lemma}

See further below for the corresponding proof. From this lemma, we can deduce the wanted theorem:
%, together with the observation that adding or removing arbitrary unary relations does not change a structure's treewidth/cliquewidth/partitionwidth \cite{feller2023decidability}

\mixed*

\begin{proof}
    Let $\Sig_1 = \Sig_a \cap \Sig_b$.
    By applying \Cref{lemmaC2} for $\psi$, obtain $\psi_{\Sig_1} \in \thegoodlogic$.
    Then, $\varphi \wedge \psi_{\Sig_1}$ and $\varphi \wedge \psi$ are equisatisfiable over the class of structures $\mathfrak{C}$ with $w(\mathfrak{C}|_{\Sig_a} ) \leq n$:
    \begin{itemize}
        \item Let $\mathfrak{A}$ be a model of $\varphi \wedge \psi_{\Sig_1}$ with $w(\mathfrak{A}|_{\Sig_a} ) \leq n$.
        By the second item of \Cref{lemmaC2}, we find a $\mathfrak{B} \models \psi$ such that
        $\mathfrak{A}$ and $\mathfrak{B}$ have the same domain and agree on all ``shared'' predicates of $\Sig_a \cap \Sig_b$ (all of them unary by assumption).
        Thus, there exists a unique structure $\mathfrak{C}$ over $\Sig_a \cup \Sig_b$ which is an expansion of both $\mathfrak{A}$ and $\mathfrak{B}$ (the ``superposition'' of the two structures).
        By virtue of this property, $\mathfrak{C}$ satisfies all valid $\Sig_b$-sentences of $\mathfrak{A}$ (in particular $\psi$) and all valid $\Sig_a$-sentences of $\mathfrak{B}$ (in particular $\varphi$), therefore $\mathfrak{C}' \models \varphi \wedge \psi$. Now, from $\mathfrak{C}|_{\Sig_a} = \mathfrak{A}|_{\Sig_a} $ also follows $w(\mathfrak{C}|_{\Sig_a} ) \leq n$
        \item Let $\mathfrak{C}$ be a model of $\varphi \wedge \psi$ with $w(\mathfrak{C}|_{\Sig_a} ) \leq n$.
        Then, by the first item of \Cref{lemmaC2}, $\mathfrak{C}$ is also a model of $\varphi \wedge \psi_{\Sig_1}$. 
    \end{itemize}
Concluding, we have shown that our decision problem can be reduced to the question if an \thegoodlogic sentence (namely $\varphi \wedge \psi_{\Sig_1}$) has a countable model $\mathfrak{C}$ with $w(\mathfrak{C}) \leq n$, which is decidable by \Cref{cor:whateverwidth}.    
\end{proof}

What remains to be proven is \Cref{lemmaC2}.
To this end, we need to introduce the notion of atomic $1$-types.

\begin{definition}[Atomic $1$-type]
    Let $\Sig$ be a signature. An \emph{atomic $1$-type} (short: $1$-type) $\pi$ over $\Sig$ is
    a maximal coherent set of atoms and negation of atoms with only one variable $x$.
    We often see $\pi$ as the conjunction of its elements.

    Note that, since $\Sig$ is finite, the set of $1$-types over $\Sig$ is finite.
    Moreover, every element $b$ in a structure $\mathfrak{B}$ \emph{realizes} 
    exactly one certain $1$-type, that is, the $1$-type $\pi$ satisfying $\mathfrak{B},\{x \mapsto b\} \models \pi$. Thus, the $1$-types form a finite partition of the
    universe of a structure.
\end{definition}

Now, let us prove \Cref{lemmaC2} by adapting the proof from \cite{Ben20}.
Let $\psi$ be a sentence in $\FOpres$ over a signature $\Sig$.
The first step is to translate $\psi$ into a sentence of $\FOpres$ in normal form, obtaining a conservative extension
$\psi^*$ using an extended signature $\Sig' \supseteq \Sig$.
That is, every model of $\psi^*$ is a model of $\psi$ and
every model of $\psi$ can be expanded to a model of $\psi^*$. 
The details of this effective translation can be found in the appendix of the full version of
\cite{Ben20}.

The second step is to introduce so-called
\emph{behaviors} $g_1, \dots, g_m$ depending on $\psi^*$. Every element $a \in A$ in a structure $\mathfrak{A}$ will be assigned exactly one behavior out of $g_1, \dots, g_m$. 
The assigned behavior quantitatively represents how $a$ is related to
the other elements in the structure. 
For our purposes, it is sufficient to account for the behaviors by defining, 
for every $\Sig'$-structure $\mathfrak{A}$,
a partition over its domain $A$, which we will represent %Which is the verb : define or consider?
by means of set variables $P_{g_1}, \dots, P_{g_m}$.

We extract the following lemma from \cite[Lemma 6]{Ben20}.

\begin{lemma}
Let $\psi^* \in \FOpres$ be a sentence in normal form over a signature $\Sig$, 
let $\pi_1, \dots, \pi_n$ be an enumeration of the atomic $1$-types over $\Sig$, 
and let $g_1, \dots, g_m$ be the list of behaviors corresponding to $\psi^*$.
    
Then one can compute a Presburger arithmetic formula $\consistent(\nv{z}_1,\ldots,\nv{z}_{n \cdot m})$ over $n \cdot m$ variables that is satisfied
by $(\abs{A_{\pi_1, g_1}}, \abs{A_{\pi_1, g_2}}, \dots, \abs{A_{\pi_n, g_m}})$ if and only if there is a structure $\mathfrak{A} \models \psi^*$ where 
$A_{\pi_i, g_j}$ is the set of the elements $a$ of $A$ of $1$-type $\pi_i$ and behavior $g_j$.
\end{lemma}

%Moreover, the formula $\consistent(\nv{z}_1,\ldots,\nv{z}_{n \cdot m})$ does not contain any
%non-assertive variable, in the sense of this paper.

To conclude, we encode this satisfiability problem with fresh set variables
$P_{\pi_1}, \dots, P_{\pi_n}$ and $P_{g_1}, \dots, P_{g_m}$, and write
the following $\thegoodlogic$-sentence:

\begin{align*}
    \psi_{\Sig_1} = \exists P_{\pi_1}. \dots \exists P_{\pi_n}. 
    \exists P_{g_1}. \dots \exists P_{g_m}.&
    \consistent \qty(\abs{P_{\pi_1} \cap P_{g_1}}, \dots, 
    \abs{P_{\pi_n} \cap P_{g_m}})\\
    \wedge& \forall x. \bigvee_{i = 1}^m P_{g_i}(x) \wedge
    \bigwedge_{1 \leq i < j \leq m} \neg (P_{g_i}(x) \wedge P_{g_j}(x))\\
    \wedge& \forall x. \bigvee_{i = 1}^n P_{\pi_i}(x) \wedge
    \bigwedge_{1 \leq i < j \leq n} \neg (P_{\pi_i}(x) \wedge P_{\pi_j}(x))\\
    \wedge& \forall x. \bigwedge_{i=1}^n P_{\pi_i}(x) \to 
    \bigwedge \{(\neg)\pr{P}(x) \in \pi_i \mid \pr{P} \in \Sig_1 \}%|_{\Sig} %To account that pi might contain R(x, x) and the likes
\end{align*}

It is now clear that the sentence $\psi_{\Sig_1}$ is as required by \Cref{lemmaC2}:
in any structure $\mathfrak{A} \models \psi_{\Sig_1}$, there are assignments of the 
$1$-types and behavior functions that are partitions, such that the $1$-types
are coherent with the unary predicates of ${\Sig_1}$, and these assignments satisfy $\consistent$,
which means that they can be completed with binary and further unary relations in a way that
$\psi^*$ is satisfied. On the other hand, any structure satisfying $\psi^*$ also satisfies
$\psi$.

%% file: app-FEmu.tex
\section{Details on Section 9}

\newcommand{\bulbul}{{\scriptscriptstyle\bullet\!\!-\!\!\bullet}}
\newcommand{\Bulbul}{\mbox{$\bullet\!\!-\!\!\bullet$}}

\begin{lemma}
Let $\mathfrak{A}$ be a tame structure over $\Sig = \Sig_\Consts \cup \Sig_{\Preds,1} \cup \Sig_{\Preds,2}$.
Then the treewidth of $\mathfrak{A}$ is at most $|\Sig_\Consts|+1$.
\end{lemma}

\begin{proof}
%We show the claim by exhibiting a tree decomposition $T=(V,E)$ of $\mathfrak{A}$ that has a width $\leq |\Sig_\Consts|+1$.
Recall that $A$ is a prefix-closed subset of of $\{rw \mid r\in Roots, w \in \mathbb{N}^*\}$ where $Roots = \{\pr{a}^\mathfrak{A} \mid \pr{a} \in \SigC\}$.
Let $V = \{ Roots \} \cup \{ Roots \cup \{a,a'\} \mid a \text{ is child of } a'\}$ and let 
$$E = \{ (Roots,A') \mid A' \in V,\  |A'| = |Roots|+1 \} \cup \{ (A',A'') \in V \times V \mid A'\cap A'' \setminus Roots \neq \emptyset  \}.$$
Then it is easy to see that $T=(V,E)$ is a tree decomposition of $\mathfrak{A}$ and the maximal size of any bag $A' \in V$ is $\leq |\Sig_\Consts|+2$, witnessing that $\mathfrak{A}$ that has a treewidth $\leq |\Sig_\Consts|+1$. 
\end{proof}

\begin{definition}
Given a fixed signature $\Sig = \Sig_\Consts \cup \Sig_{\Preds,1} \cup \Sig_{\Preds,2}$, we define
\begin{align*}
\varphi_\mathrm{anon}(z) & \coloneqq \textstyle\bigwedge_{\pr{c} \in \SigC} z \neq \pr{c} \\ 
\varphi_{ \bulbul}(z,z') & \coloneqq z\neq z' \wedge \varphi_\mathrm{anon}(z) \wedge \varphi_\mathrm{anon}(z') \wedge \textstyle\bigvee_{\pr{R} \in \Sig_{\Preds,2}} \pr{R}(z,z') \vee \pr{R}(z',z) \\
\varphi_\mathrm{closedres}(Y,X) & \coloneqq \forall y,y'\in X.\big(y\in Y \wedge \varphi_{ \bulbul}(y,y') \Rightarrow y'\in Y \big) \\
\varphi_\mathrm{reachvia}(z,z',X) & \coloneqq z\in X \wedge z'\in X  \wedge \forall Y \subseteq X. \Big( z \in Y \wedge \varphi_\mathrm{closedres}(Y,X) \Big) \Rightarrow z' \in Y \\ 
\varphi_\mathrm{connected}(X) & \coloneqq \forall y,y'\in X.\varphi_\mathrm{reachvia}(y,y',X)\\ 
\varphi_\mathrm{pthbtw}(z,z',X) & \coloneqq \varphi_\mathrm{reachvia}(z,z',X) \wedge \varphi_\mathrm{connected}(X) \ \wedge \\
& \ \ \ \ \exists^{=1}y{\in} X.\varphi_{ \bulbul}(z,y) \wedge \exists^{=1}y{\in} X.\varphi_{ \bulbul}(z',y) \wedge \forall x {\in} X.\exists^{\leq 2}y{\in} X.\varphi_{ \bulbul}(x,y)  %\\ 
%\varphi_\mathrm{dob}(z,z',X,X') & \coloneqq  X \neq X' \wedge \forall y \in X\cap X'.\big( y = z \vee y=z'\big) \wedge \varphi_\mathrm{reach}(z,z',X) \wedge \varphi_\mathrm{reach}(z,z',X')
\end{align*}
Then we define $\psi^\Sig_\mathrm{tame}$ to be the $MSO$ sentence 
$$
\neg\exists z,z',X,X'\!.\,\varphi_\mathrm{pthbtw}(z,z'\!,X) \wedge \varphi_\mathrm{pthbtw}(z,z'\!,X') \wedge X {\neq} X' %NOT NEEDEDE: \wedge \forall y {\in} X{\cap} X'.\big( y {=} z \vee y{=}z'\big).
%\varphi_\mathrm{dob}(z,z',X,X')
$$
\end{definition}

For better readability, the above definition uses the counting quantifiers $\exists^{=1}$ and $\exists^{\leq 2}$ as abbreviations, which can be easily expressed:

\smallskip

\noindent$\begin{array}{@{}rl}
~\exists^{=1}y{\in}X.\varphi_{\bulbul}(x,y) & 
\equiv \ \  \exists y.\varphi_{\bulbul}(x,y) 
\ \wedge \ \forall y_1,y_2{\in}X. \big( \varphi_{\bulbul}(x,y_1) 
\wedge \varphi_{\bulbul}(x,y_2) \big) 
\Rightarrow y_1{=}y_2 \\[1ex]   
\exists^{\leq 2}y{\in}X.\varphi_{ \bulbul}(x,y) & \equiv \ \  \forall y_1,y_2,y_3{\in}X. \big( \varphi_{ \bulbul}(x,y_1) \wedge \varphi_{ \bulbul}(x,y_2) \wedge \varphi_{ \bulbul}(x,y_3) \big) \\
& \qquad\qquad\qquad\qquad\qquad\qquad\qquad\qquad\qquad \Rightarrow y_1{=}y_2 \vee y_1{=}y_3 \vee y_2{=}y_3  
\end{array}$

\smallskip

Now we observe that $\varphi_\mathrm{anon}$ holds for all those $\{z \mapsto a\}$ where $a$ is an ``anonymous'' element, i.e., one not ``named'' by any of the constants from $\SigC$.
Further, $\varphi_{ \bulbul}$ is satisfied by those $\{z \mapsto a,z' \mapsto a'\}$ where both $a$ and $a'$ are anonymous, they are distinct, and co-occur in some relation pair of some $\pr{R}^\mathfrak{A}$ with  
$\pr{R} \in \Sig_{\Preds,2}$; we will abbreviate this situation by $a \Bulbul a'$. Then, $\varphi_\mathrm{closedres}$ characterizes those $\{Y \mapsto A',X \mapsto A''\}$, where $A'$ is closed under the restricted co-occurence relation $\Bulbul \cap (A'' \times A'')$.
Moreover, $\varphi_\mathrm{reachvia}$ holds for $\{z \mapsto a,z' \mapsto a',X \mapsto A'\}$ whenever $a = a'$ or there is an entirely anonymous $\Bulbul$-path from $a$ to $a'$ which only traverses elements of $A'$.
Then, $\varphi_\mathrm{connected}$ will identify those $\{X \mapsto A'\}$ for which the set $A'$ is $\Bulbul$-connected.
Consequently, $\varphi_\mathrm{pthbtw}$ will be satisfied by $\{z \mapsto a,z' \mapsto a',X \mapsto A'\}$ iff $A'$ consists of the elements of a simple (i.e., repetition-free) $\Bulbul$-path between $a$ and $a'$. 
Finally, $\psi^\Sig_\mathrm{tame}$ disallows the case that two anonymous domain elements are connected by two distinct entirely anonymous simple $\Bulbul$-paths. Yet, this property precisely characterizes structures isomorphic to tame structures: a structure is tame iff removing all relation instances with participating named elements and considering the Gaifman graph of the ensuing structure, we find it to be a disjoint union of undirected trees, which is characterized by the property that between any two elements, there is at most one simple path. 

\begin{lemma}
Given a signature $\Sig = \Sig_\Consts \cup \Sig_{\Preds,1} \cup \Sig_{\Preds,2}$ a $\Sig$-structure $\mathfrak{A}$ satisfies $\psi^\Sig_\mathrm{tame}$ iff it is isomorphic to a tame structure over $\Sig$.
\end{lemma}

Toward the proof of \Cref{thm:FEmu}, we first provide a translation from FEµ formulae to \thegoodlogic. 

\begin{definition}
Given an individual variable $x$, we define the translation function $\mathrm{trans}_x$ mapping FEµ formulae to MSO formulae:

\medskip

$
\begin{array}[t]{r@{\ \mapsto \ }l}
~\mathbf{true} & \mathbf{true} \\ 
~\mathbf{false} & \mathbf{false} \\[1ex] 
X & X(x)\\[1ex] 
\pr{c} & x=\pr{c}\\ 
\neg \pr{c} & x \neq \pr{c} \\[1ex] 
\pr{P} & \pr{P}(x)\\ 
\neg\pr{P} & \neg \pr{P}(x) \\
\end{array}$\ 
$
\begin{array}[t]{r@{\ \mapsto \ }l}
\varphi \wedge \varphi' & \mathrm{trans}_x(\varphi) \wedge \mathrm{trans}_x(\varphi') \\  
\varphi \vee \varphi' & \mathrm{trans}_x(\varphi) \vee \mathrm{trans}_x(\varphi') \\[1ex] 
~\langle n,\alpha \rangle\varphi & \exists Y.\pr{n}\ttleq\cnts{Y} \wedge \forall x' \in Y.\alpha(x,x') \wedge \mathrm{trans}_{x'}(\varphi) \\ 
~[n,\alpha ]\varphi & \forall Y.\pr{n}\plus\pr{1}\ttleq\cnts{Y} \Rightarrow \exists x' \in Y.\neg\alpha(x,x') \vee \mathrm{trans}_{x'}(\varphi) \\[1ex] 
\upmu X.\varphi & \forall X.(\forall x'.\mathrm{trans}_{x'}(\varphi)  \Rightarrow  X(x')) \Rightarrow X(x)\\ 
\upnu X.\varphi  & \exists X.(\forall x'.X(x') \Rightarrow \mathrm{trans}_{x'}(\varphi)) \wedge X(x) \\
\end{array}
$

\medskip

\noindent where $\pr{R}^-(x,x')$ is understood to paraphrase $\pr{R}(x',x)$. Note that, in the cases for $\upmu$ and $\upnu$, the set variable $X$ is supposed to occur freely in $\varphi$ (and hence also in $\mathrm{trans}_{x'}(\varphi)$). Contrarily, $Y$ is always meant to be a fresh set variable (in particular one not occurring in $\mathrm{trans}_{x'}(\varphi)$).
\end{definition}

The provided translation closely follows and slightly extends established canonical translations from the $\mu$-calculus to MSO logic, and along the same lines, it is straightforward to obtain the following correspondence. 

\begin{lemma}\label{lem:femuform}
Let $\varphi$ be a closed FEµ formula over a signature $\Sig$ and let $\mathfrak{A}$ be a structure over $\Sig$. Then $\mathfrak{A},\{x \mapsto a\} \models \mathrm{trans}_x(\varphi)$ iff $a \in \semof{\varphi}^{\mathfrak{A}}_\emptyset$.
\end{lemma}

In words, given a structure $\mathfrak{A}$, the semantics of FEµ assigns to each closed FEµ formula $\varphi$ a subset $A'=\semof{\varphi}^{\mathfrak{A}}_\emptyset$ of $A$ (intuitively, $A'$ represents the elements/states/worlds, ``wherein  
$\varphi$ holds''). The translation function $\mathrm{trans}_x$ maps each such $\varphi$ to an \thegoodlogic formula $\varphi'(x)$ (with one free individual variable $x$) which holds precisely for those $a \in A$ that are contained in the set $A'$ described by $\varphi$. Then, the expression $\cnts{\varphi}$ inside a global FEµ Presburger constraint $PC$ (which is the only context such a $\varphi$ can occur in) simply stands for $|A'|$. With this insights and in view of  \Cref{lem:femuform}, it is then not hard to come up with an \thegoodlogic counterpart of $PC$ where all $\cnts{\varphi_i}$ inside $PC$ are replaced back by $\cnts{X_i}$ for fresh set variables $X_i$, which are then axiomatized to be instantiated with the corresponding $\semof{\varphi_i}^{\mathfrak{A}}_\emptyset$.   

\begin{definition}
Given a global FEµ Presburger constraint $PC(\cnts{\varphi_1},\ldots,\cnts{\varphi_k})$ let 
$$
\mathrm{trans}(PC)= \exists X_1,\ldots,X_k. PC(\cnts{X_1},\ldots,\cnts{X_k}) \wedge 
\bigwedge_{1\leq i \leq k} \forall z.X_i(z) \Leftrightarrow \mathrm{trans}_z(\varphi_i). 
$$
For a finite set $\Pi$ of global FEµ Presburger constraints, let $\mathrm{trans}(\Pi) = \{\mathrm{trans}(PC) \mid PC \in \Pi\}$.
\end{definition}

\begin{lemma}
Let $\Pi$ be a finite set of global FEµ Presburger constraints. Then $\mathfrak{A} \models \Pi$ iff $\mathfrak{A} \models \mathrm{trans}(\Pi)$.
\end{lemma}

We now have assembled all ingredients to establish the wanted result.

\FEmu*

\begin{proof}[Proof]
Let $\Sig$ be a finite signature, $\varphi$ a corresponding sentence of the fully enriched $\mu$-calculus, and $\Pi$ a finite set of global FEµ Presburger constraints.
Then, by virtue of the above lemmas, tame satisfiability of  $(\varphi,\Pi)$ corresponds to satisfiability of the \thegoodlogic sentence $\psi^\Sig_\mathrm{tame} \wedge \exists x.\mathrm{trans}_x(\varphi) \wedge \bigwedge \mathrm{trans}(\Pi)$ over all countable structures of treewidth $\leq |\SigC|+1$, which is decidable by \Cref{cor:whateverwidth}.
\end{proof}